\documentclass[final]{lmcs}
\pdfoutput=1

\usepackage{lastpage}
\lmcsdoi{18}{3}{17}
\lmcsheading{}{\pageref{LastPage}}{}{}%
{Sep.~08,~2020}{Aug.~09,~2022}{}


\usepackage{lineno,hyperref}
\modulolinenumbers[5]

\usepackage[utf8]{inputenc}
\usepackage{amsfonts,amssymb,amsmath,amsthm}
\usepackage{xspace}
\usepackage{xparse}
\usepackage{relsize}
\usepackage{tikz}
\usetikzlibrary{arrows,automata,shapes,calc,positioning,intersections}
\tikzset{>=latex}
\usepackage[strict]{changepage}
\usepackage[ruled,boxed,vlined,linesnumbered]{algorithm2e}
\usepackage{hyperref}

\newcommand\N{\ensuremath{\mathbb{N}}}
\newcommand\R{\ensuremath{\mathbb{R}}}
\newcommand\Rplus{\ensuremath{\R^+}}
\newcommand\Rbar{\ensuremath{\overline{\R}}}
\newcommand\Z{\ensuremath{\mathbb{Z}}}
\newcommand\Q{\ensuremath{\mathbb{Q}}}
\renewcommand\vec[1]{\boldsymbol{#1}}
\renewcommand\epsilon{\varepsilon}

\renewcommand\leq{\leqslant}
\renewcommand\geq{\geqslant}

\newcommand\slope{\textrm{slope}}

\newcommand\MinPl{\ensuremath{\mathsf{Min}}\xspace} 
\newcommand\MaxPl{\ensuremath{\mathsf{Max}}\xspace} 

\newcommand\Guard[1]{\ensuremath{\mathsf{Guard}(#1)}} 
\newcommand\regcst[1]{\ensuremath{\mathsf{Reg}_{#1}}} 
\newcommand\den[1]{\ensuremath{[\![#1]\!]\xspace}} 

\newcommand\valuation{\ensuremath{\nu}\xspace}
\newcommand\CF[1]{\ensuremath{{\sf CF}_{#1}}} 
\newcommand\opcf{\ensuremath{\rhd}} 

\newcommand\PTG{\textrm{PTG}\xspace}
\newcommand\NRAPTG{\textrm{NRAPTG}\xspace}
\newcommand\Locs{\ensuremath{L}}
\newcommand\locs{\Locs}
\newcommand\LocsMin{\Locs_{\MinPl}}
\newcommand\LocsMax{\Locs_{\MaxPl}}
\newcommand\LocsFin{\ensuremath{\Locs_f}} 
\newcommand\LocsUrg{\ensuremath{\Locs_u}} 
\newcommand\loc{\ensuremath{\ell}}

\newcommand\fgoal{\ensuremath{\varphi}}
\newcommand\fgoalvec{\ensuremath{\boldsymbol{\varphi}}}
\newcommand\transitions{\ensuremath{\Delta}}
 
\newcommand\price{\ensuremath{\pi}}
\newcommand\transition{\delta}
\newcommand\game{\ensuremath{\mathcal G}\xspace}
\newcommand\hame{\ensuremath{\mathcal H}\xspace}
\newcommand\cstgame{\ensuremath{S_\game}}
\newcommand\reggame{\regcst{\game}}

\newcommand\maxPriceLoc{\ensuremath{\Pi^{\mathrm{loc}}}\xspace}
\newcommand\maxPriceTrans{\ensuremath{\Pi^{\mathrm{tr}}}\xspace}
\newcommand\maxPriceFin{\ensuremath{\Pi^{\mathrm{fin}}}\xspace}

\newcommand\confs[1]{\ensuremath{{\sf Conf}_{#1}}\xspace}
\newcommand\confgame{\confs{\game}}
\newcommand\costname{{\sf Price}}
\newcommand\costgame[2]{\ensuremath{\costname_{#1}(#2)}}
\newcommand\cost[1]{\ensuremath{\costname(#1)}}
\newcommand\costTrans[1]{\ensuremath{{\sf Cost}^{\mathrm{tr}}(#1)}}
\newcommand\costLoc[1]{\ensuremath{{\sf Cost}^{\mathrm{loc}}(#1)}}
\newcommand\len[1]{\ensuremath{|#1|}}
\newcommand\run{\rho}

\newcommand\pname{cost\xspace}
\newcommand\pnames{costs\xspace}
\newcommand\puse{{\sf Cost}}

\newcommand\strat{\ensuremath{\sigma}\xspace}
\newcommand\stratmin{\strat_{\MinPl}}
\newcommand\stratmax{\strat_{\MaxPl}}
\newcommand\stratsofmin{\ensuremath{{\sf Strat}_{\MinPl}}}
\newcommand\stratsofmax{\ensuremath{{\sf Strat}_{\MaxPl}}}
\newcommand\stratsofmingame[1]{\stratsofmin(#1)}
\newcommand\stratsofmaxgame[1]{\stratsofmax(#1)}
\newcommand\outcomes{\mathsf{Play}}
\newcommand\coutcomes{\mathsf{CPlay}}
\newcommand\Play[1]{\ensuremath{\outcomes(#1)}}
\newcommand\CPlay[1]{\ensuremath{\coutcomes(#1)}}

\newcommand\points{\mathsf{pts}}
\newcommand\intervals{\mathsf{int}}
\newcommand\fakeValue{\mathsf{fake}}

\newcommand\val{\ensuremath{{\sf Val}}}
\newcommand\Val{\val}
\newcommand\uval{\ensuremath{\overline{\val}}}
\newcommand\lval{\ensuremath{\underline{\val}}}
\newcommand\uppervalue{\uval}
\newcommand\lowervalue{\lval}
\newcommand\Value{\val}

\newcommand\valgs[2]{\val_{#1}^{#2}}

\newcommand\valgame{\valgs{\game}{}}

\newcommand\SPTG{\textrm{SPTG}\xspace}
\newcommand\rightpoint{\ensuremath{r}\xspace}

\SetKwFunction{Waiting}{wait}
\SetKwFunction{SolveInstant}{solveInstant}
\SetKwFunction{Solve}{solve}
\newcommand\locMin{\ensuremath{\loc^\star}}
\newcommand\Next{\textsf{left}}

\newcommand\operator{\mathcal F}

\newcommand\bupval[1]{\overline{\Value}^{\leq #1}}
\newcommand\minstrategy{\stratmin}
\newcommand\maxstrategy{\stratmax}
\newcommand\Zstar{\ensuremath{\Z_{\valuation,\fgoalvec}}}
\newcommand\Zstarinf{\ensuremath{\Z_{\valuation,\fgoalvec}^{+\infty}}}
\newcommand\costbound[1]{\costname^{\leq#1}}
\newcommand\possval{\ensuremath{{\sf PossVal}}}
\newcommand\posscp{{\sf PossCP}}
\newcommand{\F}{{\sf F}}
\newcommand\WaitTr{\ensuremath{\mathrm{WaitTr}}}

\newcommand\he{it\xspace}

\newcommand\him{it\xspace}
\newcommand\his{its\xspace}
\newcommand\His{Its\xspace}

\begin{document}


\title{One-Clock Priced Timed Games with Negative Weights}
\titlecomment{A preliminary version of this work has been published in
  the proceedings of FSTTCS 2015~\cite{BGHLM15a}. The research leading
  to these results was funded by the European Union Seventh Framework
  Programme (FP7/2007-2013) under Grant Agreement
  $\textrm{n}^{\textrm{o}}$601148 (CASSTING). This work was also
  partly supported by the Fonds de la Recherche Scientifique - FNRS
  under grant $\textrm{n}^{\textrm{o}}$T.0027.21. During part of the
  preparation of this article, the last author was (partially) funded
  by the ANR project DeLTA (ANR-16-CE40-0007) and the ANR project
  Ticktac (ANR-18-CE40-0015).}

\author[T.~Brihaye]{Thomas Brihaye\lmcsorcid{0000-0001-5763-3130}}[a]

\author[G.~Geeraerts]{Gilles Geeraerts}[b]

\author[A.~Haddad]{Axel Haddad}[a]

\author[E.~Lefaucheux]{Engel Lefaucheux\lmcsorcid{0000-0003-0875-300X}}[c]

\author[B.~Monmege]{Benjamin Monmege\lmcsorcid{0000-0002-4717-9955}}[d]

\address{Université de Mons, Belgium}
\email{thomas.brihaye@umons.ac.be, axel.haddad@umons.ac.be}

\address{Université libre de Bruxelles, Belgium}
\email{gigeerae@ulb.ac.be}

\address{Max-Planck Institute for Software Systems, Germany}
\email{elefauch@mpi-sws.org}

\address{Aix Marseille Univ, LIS, CNRS, Marseille, France}
\email{benjamin.monmege@univ-amu.fr}


\begin{abstract}
  Priced timed games are two-player zero-sum games played on priced
  timed automata (whose locations and transitions are labeled by
  weights modelling the \pname of spending time in a state and
  executing an action, respectively). The goals of the players are to
  minimise and maximise the \pname to reach a target location,
  respectively. We consider priced timed games with one clock and
  arbitrary integer weights and show that, for an important subclass
  of them (the so-called \emph{simple} priced timed games), one can
  compute, in pseudo-polynomial time, the optimal values that the
  players can achieve, with their associated optimal strategies. As
  side results, we also show that one-clock priced timed games are
  determined and that we can use our result on simple priced timed
  games to solve the more general class of so-called
  \emph{negative-reset-acyclic} priced timed games (with arbitrary
  integer weights and one clock). The decidability status of the full
  class of priced timed games with one-clock and arbitrary integer
  weights still remains open.
\end{abstract}

\maketitle



\section{Introduction} 

Game theory is nowadays a
well-established framework in theoretical computer science, enabling
computer-aided design of computer systems that are
correct-by-construction. It allows one to describe and analyse the
possible interactions of antagonistic agents (or players) as in the
\emph{controller synthesis} problem, for instance. This problem asks,
given a model of the environment of a system, and of the possible
actions of a controller, to compute a controller that constraints the
environment to respect a given specification. Clearly, one cannot
assume in general that the two players (the environment and the
controller) will collaborate, hence the need to find a \emph{strategy
  for the controller} that enforces the specification \emph{whatever
  the environment does}. This question thus reduces to computing a
so-called winning strategy for the corresponding player in the game
model.

In order to describe precisely the features of complex computer
systems, several game models have been considered in the
literature. In this work, we focus on the model of Priced Timed Games
(\PTG{s} for short), which can be regarded as an extension (in several
directions) of classical finite automata. First, like timed
automata~\cite{AluDil94}, \PTG{s} have \emph{clocks}, which are
real-valued variables whose values evolve with time elapsing, and
which can be tested and reset along the transitions. Second, the
locations are associated with weights representing rates and transitions are labeled
by discrete weights, as in priced timed
automata~\cite{BehFeh01,AluLa-04,BouBri07}. These weights allow one to
associate a price with each play (or run), which depends on the
sequence of transitions traversed by the play, and on the time spent in
each visited location. Finally, a \PTG is played by two players,
called $\MinPl$ and $\MaxPl$, and each location of the game is owned
by either of them (we consider a turn-based version of the game). The
player who controls the current location decides how long to wait, and
which transition to take.

In this setting, the goal of $\MinPl$ is to reach a given set of
target locations, while minimising the price of the play to reach such
a location. Player $\MaxPl$ has an antagonistic objective: \he~tries to
avoid the target locations, and, if not possible, to maximise the
accumulated \pname up to the first visit of a target location. To
reflect these objectives, we define the upper value~$\uval$ of the
game as a mapping of the configurations of the \PTG to the least price
that $\MinPl$ can guarantee while reaching the target, whatever the
choices of $\MaxPl$. Similarly, the lower value~$\lval$ returns the
greatest price that $\MaxPl$ can ensure (letting the price be
$+\infty$ in case the target locations are not reached).

\begin{figure}[tbp]
  \centering
  \begin{tikzpicture}[minimum size=5mm,node distance=1.3cm]
        \everymath{\footnotesize}
        
        \node[draw,circle,label={[label distance=-1mm]above:$-2$}] (q1) {\makebox[0pt][c]{$\loc_1$}};
        \node[draw,rectangle,below of=q1,label={[label distance=-1mm]above left:$-14$}] (q2) {\makebox[0pt][c]{$\loc_2$}};
        \node[draw,circle,below right of=q1,xshift=8mm,yshift=3mm,label={[label distance=-1mm]above:$4$}] (q3) {\makebox[0pt][c]{$\loc_3$}};
        \node[draw,rectangle,above right of=q3,xshift=8mm,yshift=-3mm,label={[label distance=-0.5mm,yshift=-1mm]above left:$3$}] (q4) {\makebox[0pt][c]{$\loc_4$}};
        \node[draw,circle,left of=q2,label=below:$8$,xshift=-5mm] (q5) {\makebox[0pt][c]{$\loc_5$}};
        \node[draw,circle,above of=q5,label=above:$-12$] (q6) {\makebox[0pt][c]{$\loc_6$}};
        \node[draw,circle,below of=q4,label={[label distance=-1mm]below:$-16$}] (q7){\makebox[0pt][c]{$\loc_7$}};
        \node[draw,circle,below right of=q4,xshift=8mm,yshift=3mm,accepting] (qf) {\makebox[0pt][c]{$\loc_f$}};

        \path[->] (q1) edge (q2) 
        (q2) edge (q3) 
        (q2) edge (q5)
        (q5) edge (q6) 
        (q6) edge node[above] {$1$} (q1)
        (q5) edge[bend right=15] node[pos=.7,above,xshift=2mm] {$2$} (q7) 
        (q3) edge node[above right,xshift=-2mm] {$6$} (q7)
        (q3) edge (q1)
        (q3) edge (q4) 
        (q4) edge node[below left,xshift=1mm,yshift=1mm] {$-7$} (qf)
        (q7) edge (qf) 
        (q1) edge[bend left=40] (qf);
        
        \draw[dotted] (1.3,-2) rectangle (5.6,1) ;
      \end{tikzpicture}
     \qquad
         \begin{tikzpicture}[xscale=.8,yscale=0.55]
           \draw[->] (6,-5) -- (10.5,-5) node[anchor=north] {$\valuation$};
           \draw	(6,-5) node[anchor=south] {$0$}
           (7,-5) node[anchor=south] {$\frac 1 4$}
           (8,-5) node[anchor=south] {$\frac 1 2$}
           (9,-5) node[anchor=south] {$\frac 3 4$}
           (9.6,-5) node[anchor=south] {$\frac 9 {10}$}
           (10,-5) node[anchor=south] {$1$};
           
           \draw[->] (6,-5) -- (6,-9) node[anchor=west] {$\lval(\loc_1,\valuation)$};
           \draw	(6,-8.3) node[anchor=east] {$-9.5$}
           (6,-7.3) node[anchor=east] {$-6$}
           (6,-6.6) node[anchor=east] {$-5.5$}
           (6,-5.8) node[anchor=east] {$-2$}
           (6,-5.1) node[anchor=east] {$-0.2$};

           \draw[thick] (6,-8.3) -- (7,-7) --
           (8,-6.85)--(9,-5.6)--(9.6,-5.1)--(10,-5);

           \draw[dotted] (7,-7) -- (6,-7);
           \draw[dotted] (7,-7) -- (7,-5);
           \draw[dotted] (8,-6.85) -- (6,-6.85) ;
           \draw[dotted] (8,-6.85) -- (8,-5) ;
           \draw[dotted] (9,-5.6) -- (6,-5.6) ;
           \draw[dotted] (9,-5.6) -- (9, -5) ;
           \draw[dotted] (9.6,-5.1) -- (6,-5.1) ;
           \draw[dotted] (9.6,-5.1) -- (9.6,-5) ;

         \end{tikzpicture}
         \caption{A simple priced timed game (left) and the lower
           value function of location $\loc_1$ (right). Transitions
           without label have weight $0$.}
    \label{fig:ex-ptg2}
\end{figure}
An example of \PTG is given in \figurename~\ref{fig:ex-ptg2}, where
the locations of $\MinPl$ and $\MaxPl$ are represented by circles and
rectangles respectively. The integers next to the locations are their
rates, i.e.~the \pname of spending one time unit in the
location. Moreover, there is only one clock $x$ in the game, which is
never reset, and all guards on transitions are $x\in[0,1]$ 
which force every player to keep the clock value below or equal to 1, but do not 
hinder the choice of transition
(hence this guard is not displayed and transitions are only labelled by their
respective discrete weight): this is an example of a \emph{simple
  priced timed game} (we will define them properly later). It is easy
to check that $\MinPl$ can force reaching the target location $\ell_f$
from all configurations~$(\ell,\valuation)$ of the game, where $\ell$
is a location and $\valuation$ is a real value of the clock in
$[0,1]$. Let us comment on the optimal strategies for both players.
From a configuration~$(\loc_4,\valuation)$, with $\valuation\in[0,1]$,
$\MaxPl$ better waits until the clock takes value $1$, before taking
the transition to $\loc_f$ (\he is forced to move, by the rules of the
game). Hence, $\MaxPl$'s optimal value
is~$3(1-\valuation)-7 = -3\valuation -4$ from all configurations
$(\loc_4,\valuation)$. Symmetrically, it is easy to check that
$\MinPl$ better waits as long as possible in $\loc_7$, hence \his
optimal value is~$-16(1-\valuation)$ from all configurations
$(\loc_7,\valuation)$. However, optimal value functions are not always
\emph{that simple}, see for instance the lower value function of
$\loc_1$ on the right of Figure~\ref{fig:ex-ptg2}, which is a
piecewise affine function. To understand why value functions can be
piecewise affine, consider the sub-game enclosed in the dotted
rectangle in Figure~\ref{fig:ex-ptg2}, and consider the value that
$\MinPl$ can guarantee from a configuration of the form
$(\loc_3,\valuation)$ in this sub-game. Clearly, $\MinPl$ must decide
how long \he will spend in $\loc_3$ and whether \he will go to
$\loc_4$ or $\loc_7$. \His optimal value from all
$(\loc_3,\valuation)$ is thus
\[\inf_{0\leq t\leq 1-\valuation} \min\big(4t+ 3 (1-(\valuation+t))-7, 4t
  + 6 -16(1-(\valuation+t))\big) = \min
  (-3\valuation-4,16\valuation-10)\,.\] Since
$16\valuation-10\geq -3\valuation-4$ if and only if
$\valuation \geq 6/19$, an optimal choice of $\MinPl$ is to move
instantaneously to $\loc_7$ if $\valuation \in [0,6/19]$ and to move
instantaneously to $\loc_4$ if $\valuation \in (6/19,1]$, hence the
value function of $\loc_3$ (in the sub-game) is a piecewise affine
function with two pieces.

\subsection*{Related work.} \PTG{s} are a special case of hybrid games~\cite{AlfHen01,MalPnu95,Won97}, independently investigated in~\cite{BouCas04} and~\cite{AluBer04}. For (non-necessarily turn-based)
\PTG{s} with \emph{non-negative} weights, semi-algorithms are given to
decide the \emph{value problem}, that is to say, whether the upper
value of a location (the best price that $\MinPl$ can guarantee
starting with a clock value $0$), is below a given threshold. It
was also shown that, under the \emph{strongly non-Zeno assumption} on
weights (asking the existence of $\kappa>0$ such that every cycle in
the underlying region graph has a weight at least $\kappa$), the
proposed semi-algorithms always terminate. This assumption was
justified in~\cite{BriBru05,BouBri06} by showing that, without it, the
\emph{existence problem}, that is to decide whether $\MinPl$ has a
strategy guaranteeing to reach a target location with a price below a
given threshold, is indeed undecidable for \PTG{s} with non-negative
weights and three or more clocks. This result was recently extended in~\cite{BJMr14} to show that the value problem is also
undecidable for \PTG{s} with non-negative weights and four or more
clocks. In~\cite{BCJ09}, the undecidability of the existence problem
has also been shown for \PTG{s} with arbitrary weights on locations
(without weights on transitions), and two or more clocks.  
Finally, PSPACE-hardness of the value problem has been established for 
one-clock \PTG{s} in~\cite{FeaIbs20}. On a
positive side, the value problem was shown decidable by~\cite{BouLar06} for \PTG{s} with one clock when the weights are
non-negative: a 3-exponential time algorithm was first proposed,
further refined in~\cite{Rut11,DueIbs13} into an exponential time
algorithm. The key point of those algorithms is to reduce the problem
to the computation of optimal values in a restricted family of \PTG{s}
called \emph{Simple Priced Timed Games} (\SPTG{s} for short), where
the underlying automata contain no guard, no reset, and the play is
forced to stop after one time unit.  More precisely, the \PTG is
decomposed into a sequence of \SPTG{s} whose value functions are
computed and re-assembled to yield the value function of the original
\PTG.\@ Alternatively, and with radically different techniques, a
pseudo-polynomial time algorithm to solve one-clock \PTG{s} with
arbitrary weights on transitions, and rates restricted to two values
amongst $\{-d,0,+d\}$ (with $d\in\N$) was given in~\cite{BGKMMT14}. More recently, a large subclass of PTGs with
arbitrary weights and no restrictions on the number of clocks was
introduced in~\cite{BMR17a}, whose value can be computed in
double-exponential time: they are defined via a partition of strongly
connected components with respect to the sign of all the cycles they
contain. A survey summarising results on PTGs can be found in~\cite{Bou15}.

\subsection*{Contributions.} Following the decidability results
sketched above, we consider \PTG{s} with one clock. We extend those
results by considering arbitrary (positive and negative)
weights. Indeed, all previous works on \PTG{s} with only one clock
(except~\cite{BGKMMT14}) have considered non-negative weights only,
and the status of the more general case with arbitrary weights has so
far remained elusive. Yet, arbitrary weights are an important
modelling feature. Consider, for instance, a system which can consume
but also produce energy at different rates. In this case, energy
consumption could be modelled as a positive rate, and production by a
negative rate. As another example, imagine the billing system for
electrical power, in a smart house that itself produces energy: the
money spent/earned while using/producing energy contains both timed
components (the more energy the house produces, the more money the
owner gets), and discrete ones (the electricity provider charges
discrete costs per month). Such a model has been studied in~\cite{BDGHM16}. In the untimed setting, such extension to negative
weights has been considered in~\cite{BGHM14,BGHM16}: our result
heavily builds upon techniques investigated in these works, as we will
see later. Our main contribution is a \emph{pseudo-polynomial time
  algorithm to compute the value of one-clock \SPTG{s} with arbitrary
  weights}. While this result might sound limited due to the
restricted class of simple \PTG{s} we can handle, we recall that the
previous works mentioned above~\cite{BouLar06,Rut11,DueIbs13} have
demonstrated that solving \SPTG{s} is a key result towards solving
more general \PTG{s}. Moreover, this algorithm is, as far as we know,
the first to handle the full class of \SPTG{s} with arbitrary weights,
and we note that the solutions (either the algorithms or the proofs)
known so far do not generalise to this case. Notice also
that previous algorithms provided exponential-time algorithms
($2^{O(n^2)}$ for~\cite{Rut11} and $O(12^n)$ for~\cite{DueIbs13}, with
$n$ the number of locations of the game), whereas we obtain a
pseudo-polynomial time complexity (see Theorem~\ref{the:ExpCutpoints}
for the exact bound).\footnote{In the shorter version of this article,
  published in the proceedings of FSTTCS 2015~\cite{BGHLM15a}, only
  exponential-time complexity was provided: new techniques have
  allowed us to obtain the pseudo-polynomial time complexity.}
Finally, as a side result, this algorithm allows us to solve the more
general class of \emph{negative-reset-acyclic} one-clock \PTG{s} that
we introduce. This also improves the previous exponential complexity
for one-clock \PTG{s} with only non-negative prices to a
pseudo-polynomial time complexity. However, the
decidability (and thus complexity) of the whole class of one-clock
\PTG{s} with arbitrary weights remains open so far: our result may be
seen as a potentially important milestone towards this goal.

\section{Quantitative reachability games}\label{sec:quant-reach-games}

The semantics of the priced timed games we study in this work can be
expressed in the setting of \emph{quantitative reachability games} as
defined below.  Intuitively, in such a game, two players ($\MinPl$ and
$\MaxPl$) play by changing alternatively the current configuration of
the game. The game ends when it reaches a final configuration, and
$\MinPl$ has to pay to $\MaxPl$ a price associated with the sequence of
configurations and of transitions taken (hence, $\MinPl$ is trying to
minimise this price while $\MaxPl$ wants to maximise it).

Note that the framework of quantitative reachability games that
we develop here (and for which we prove a determinacy result, see
Theorem~\ref{thm:determined}) can be applied to other settings than
priced timed games. For example, special cases of quantitative
reachability games are \emph{finite} quantitative reachability
games---where the set of configurations is finite---that have been
thoroughly studied in~\cite{BGHM16} under the name of \emph{min-cost
  reachability games}. In this article, we will rely on quantitative
reachability games with \emph{uncountably} many states as the
semantics of priced timed games. Similarly, our quantitative
reachability games could be used to formalise the semantics of hybrid
games~\cite{BouBri06a,BouBri08} or any (non-probabilistic) game with a
reachability objective.
We start our discussion by defining formally those games:
\begin{defi}[Quantitative reachability games]\label{def:QRG}
  A \emph{quantitative reachability game} is a tuple
  $G=(C_\MinPl, C_\MaxPl, F, \Sigma, E, \puse)$, where $C=C_\MinPl \uplus C_\MaxPl\uplus F$ is the
  set of \emph{configurations} (that does not need to be finite, nor
  countable), partitioned into the set $C_\MinPl$ of configurations of
  player~\MinPl, the set $C_\MaxPl$ of configurations of
  player~\MaxPl, and the set $F$ of \emph{final configurations}; 
  $\Sigma$ is a (potentially infinite) alphabet whose
  elements are called letters; $E\subseteq (C\setminus F)\times \Sigma \times C$ is
  the transition relation; and $\puse\colon (C\Sigma)^* C\to \R$ maps each
  finite sequence $c_1a_1 \cdots a_{n-1}c_{n}$ to a real number called
  the \emph{\pname} of $c_1a_1 \cdots a_{n-1}c_{n}$.
\end{defi}

A \emph{play} is a finite or infinite sequence $\rho=c_1a_1c_2\cdots$
alternating between configurations and letters, and such that for all
$i\geq 0$: $(c_i,a_i,c_{i+1}) \in E$. The length of a play $\rho$, denoted $\len{\rho}$
is the number of configuration occurring in it. As such, if $\rho$ is infinite,
$\len{\rho}= + \infty$. For the sake of clarity, we denote a play
$c_1 a_1 c_2 \cdots$ as $c_1\xrightarrow{a_1}c_2\cdots$.
A \emph{completed play} is either (1) an
infinite play, or
(2) a finite play ending in a
deadlock, i.e.~a configuration $c$ such that the set
$\{(c,a,c')\in E\mid a\in \Sigma \wedge c'\in C\}$ is empty. Note that
every play reaching a final state ends in a deadlock, hence infinite
plays never visit $F$.  

We take the viewpoint of player \MinPl who wants to reach a final
configuration. Thus, the \emph{price} of a completed play
$\rho = c_1\xrightarrow{a_1} c_2 \cdots$, denoted $\cost{\rho}$ is
either $+\infty$ if either $\len{\rho}=+\infty$ or $\rho$ is a finite 
play that does not end in a final state (this is the worst situation
for \MinPl, which explains why the price is maximal in this case);
or $\puse(c_1\xrightarrow{a_1} c_2 \cdots c_n)$ if $\len{\rho}=n$ and $c_n\in F$.

A \emph{strategy} for player $\MinPl$ is a function $\stratmin$
mapping all finite plays ending in a configuration~$c\in C_\MinPl$
(excluding deadlocks) to a transition $(c,a,c')\in E$.  Strategies
$\stratmax$ of player $\MaxPl$ are defined accordingly.  We let
$\stratsofmingame{G}$ and $\stratsofmaxgame{G}$ be the sets of
strategies of $\MinPl$ and $\MaxPl$, respectively. A pair
$(\stratmin,\stratmax)\in \stratsofmingame{\game}\times
\stratsofmaxgame{\game}$ is called a \emph{profile of
  strategies}. Together with an initial configuration $c_1$, it
defines a unique completed play
$\CPlay{c_1,\stratmin,\stratmax}=c_1\xrightarrow{a_1} c_2\cdots$ such
that for all $i\geq 0$:
$(c_i, a_i,c_{i+1}) = \stratmin (c_1\xrightarrow{a_1}c_2 \cdots c_i)$
if $c_i \in C_\MinPl$; and
$(c_i, a_i,c_{i+1}) = \stratmax (c_1\xrightarrow{a_1} c_2 \cdots c_i)$
if $c_i \in C_\MaxPl$.  We let $\Play{\stratmin}$
(resp.\ $\CPlay{\stratmin}$) be the set of plays (resp.\ completed
plays) that conform with $\stratmin$. That is,
$c_1\xrightarrow{a_1} c_2\cdots\in \Play{\stratmin}$ iff for all
$i\geq 0$: $c_i \in C_\MinPl$ implies
$(c_i, a_i,c_{i+1}) = \stratmin (c_1\xrightarrow{a_1}c_2 \cdots c_i)$.
We let $\Play{c_1,\stratmin}$ (resp.\ $\CPlay{c_1,\stratmin}$) be the
subset of plays from $\Play{\stratmin}$ (resp.\ $\CPlay{\stratmin}$)
that start in $c_1$.  We define $\Play{\stratmax}$ and
$\Play{c_1,\stratmax}$ as well as the completed variants
accordingly. Given an initial configuration~$c_1$, the price of a
strategy $\stratmin$ of $\MinPl$ is:
\[ \cost{c_1,\stratmin}=\sup_{\rho\in
  \CPlay{c_1,\stratmin}}\cost{\rho}\,.\]
It matches the intuition to be the largest price that $\MinPl$ may pay
while following strategy~$\stratmin$. This definition is equal to
$\sup_{\stratmax}\cost{\CPlay{c_1,\stratmin,\stratmax}}$, which is
intuitively the highest price that $\MaxPl$ can force $\MinPl$ to pay
if $\MinPl$ follows $\stratmin$. Similarly, given a strategy
$\stratmax$ of $\MaxPl$, we define the price of $\stratmax$ as
\[ \cost{c_1,\stratmax}=\inf_{\rho\in
  \CPlay{c_1,\stratmax}}\cost{\rho}=\inf_{\stratmin}\cost{\CPlay{c_1,\stratmin,\stratmax}}\,.\]
It corresponds to the least price that $\MinPl$ can achieve once
$\MaxPl$ has fixed its strategy $\stratmax$.

From there, two different definitions of the value of a configuration
$c_1$ arise, depending on which player chooses its strategy first. The
\emph{upper value} of $c_1$, defined as
\[\uppervalue(c_1) = \inf_{\stratmin} \sup_{\stratmax}
\cost{\CPlay{c_1,\stratmin,\stratmax}}\]
corresponds to the least price that $\MinPl$ can ensure when choosing
its strategy \emph{before} $\MaxPl$, while the \emph{lower value},
defined as
\[\lowervalue(c_1) = \sup_{\stratmax}\inf_{\stratmin}
\cost{\CPlay{c_1,\stratmin,\stratmax}}\]
corresponds to the least price that $\MinPl$ can ensure when choosing
its strategy \emph{after} $\MaxPl$.  It is easy to see that
$\lowervalue(c_1) \leq \uppervalue(c_1)$, which explains the chosen
names. Indeed, if $\MinPl$ picks its strategy after $\MaxPl$, \he has
more information, and then can, in general, choose a better response.

In general, the order in which players choose their strategies can
modify the outcome of the game. However, for quantitative reachability
games, this makes no difference, and the value is the same whichever
player picks \his strategy first. This result, known as the
\emph{determinacy property}, is formalised here:

\begin{thm}[Determinacy of quantitative reachability
  games]\label{thm:determined}
  For all quantitative reachability games $G$ and configurations
  $c_1$, $\uppervalue(c_1) = \lowervalue(c_1)$.
\end{thm}
\begin{proof}
  To establish this result, we rely on a general determinacy result of
  Donald Martin~\cite{Mar75}. This result concerns \emph{qualitative}
  games (i.e.~games where players either win or lose the game, and do
  not pay a price), called Gale-Stewart games. So, we first explain how
  to reduce a quantitative reachability game
  $G=(C_\MinPl , C_\MaxPl, F, \Sigma, E,  \puse)$ to a family of such
  Gale-Stewart games $\textit{Threshold}(G,r)$ parametrised by a
  threshold $r\in \R$.

  The Gale-Stewart game $\textit{Threshold}(G,r)$ is played on an
  infinite tree whose vertices are owned by either of the players. A
  play is then a maximal branch in this tree, built as follows: the
  player who owns the root of the tree first picks a successor of the
  root that becomes the current vertex. Then, the player who owns this
  vertex gets to choose a successor that becomes the current one,
  etc. The game ends when a leaf is reached, where the winner is
  declared thanks to a given set $\mathit{Win}$ of winning leaves.
  
  In our case, the vertices of $\textit{Threshold}(G,r)$ are the finite plays 
  $c_1\xrightarrow{a_1}c_2\cdots c_n$ of $G$ starting from configuration
  $c_1$. Such a vertex $v=c_1\xrightarrow{a_1}c_2\cdots c_n$ is owned by 
  $\MinPl$ iff~$c_n\in C_\MinPl$; otherwise $v$ belongs to $\MaxPl$. A
  vertex $v=c_1\xrightarrow{a_1}c_2\cdots c_n$ has successors only if 
  $c_n\not\in F$. In
  this case, the successors of $v$ are all the vertices $v\xrightarrow{a}c$
  such that $(c_n,a,c)\in E$. Finally, a leaf $c_1\xrightarrow{a_1}c_2\cdots c_n$ 
  (thus, with $c_n\in F$) is winning for $\MinPl$ iff
  $\puse(c_1\xrightarrow{a_1}c_2\cdots c_n)\leq r$.

  As a consequence, the set of winning plays in
  $\textit{Threshold}(G,r)$ is:
 \[
    \mathit{Win}=\bigcup_{v\in L\text{ s.t. } \puse(v)\leq r} \{\textit{branch}(v)\}
  \]
  where $L$ is the set of leaves of $\textit{Threshold}(G,r)$, and
  $\textit{branch}(v)$ is the (unique) branch from $c_1$ to $v$. Then,
  we rewrite the definition of $\mathit{Win}$ as:
  \[
    \mathit{Win}= \bigcup_{v\in L\text{ s.t. } \puse(v)\leq r}
    \textit{Cone}(v)
  \]
  where $\textit{Cone}(v)$ is the set of plays in
  $\textit{Threshold}(G,r)$ that visit $v$. Indeed, when $v$ is a
  leaf, the set $\textit{Cone}(v)$ reduces to the singleton containing
  only $\textit{branch}(v)$. Thus, the set of winning plays (for
  $\MinPl$) is an open set, defined in the topology generated from the
  $\textit{Cone}(v)$ sets, and we can apply~\cite{Mar75} to conclude
  that $\textit{Threshold}(G,r)$ is a determined game for all
  quantitative reachability games $G$ and all thresholds $r\in \R$
  i.e.~either $\MinPl$ or $\MaxPl$ has a winning strategy from the
  root of the tree. Moreover, notice that \MinPl wins the game
  $\textit{Threshold}(G,r)$ iff it guarantees in $G$ an upper value
  $\uppervalue(c_1)\leq r$. Similarly, \MaxPl wins the game
  $\textit{Threshold}(G,r)$ iff it guarantees in $G$ a lower value
  $\lowervalue(c_1)> r$.

  We rely on this result to prove that
  $\lowervalue(c_1) \geq \uppervalue(c_1)$ in $G$ (the other
  inequality holds by definition of $\lowervalue(c_1)$ and
  $\uppervalue(c_1)$). We consider two cases:
  \begin{enumerate}
  \item If $\uppervalue(c_1) = -\infty$, then, since
    $\lowervalue(c_1)\leq\uppervalue(c_1)$, we have
    $\lowervalue(c_1)=-\infty$ too.
  \item If $\uppervalue(c_1) > -\infty$, consider any real number $r$
    such that $r < \uppervalue(c_1)$. 
    Therefore, \MinPl loses in the game
    $\textit{Threshold}(G,r)$. By determinacy, \MaxPl wins in this
    game, i.e.~$\lowervalue(c_1)\geq r$. Therefore, $r < \uppervalue(c_1)$
    implies $r\leq \lowervalue(c_1)$. Thus, we have shown that, for
    all $r$: $r < \uppervalue(c_1)$ implies $r\leq
    \lowervalue(c_1)$. This is equivalent to saying that for all $r$:
    either $r\geq \uppervalue(c_1)$, or $r\leq \lowervalue(c_1)$. This
    can happen only when
    $\uppervalue(c_1)\leq \lowervalue(c_1)$.\qedhere
  \end{enumerate}
\end{proof}

Now that we have showed that quantitative reachability games are
determined, we can denote by $\val$ the value of the game, defined
as $\val=\uppervalue=\lowervalue$.


\section{Priced timed games}
\label{sec:ptg}
We are now ready to formally introduce the core model of this article:
priced timed games. We start by the formal definition, then study some
properties of the value function of those games
(Section~\ref{sec:properties-value}). Next, we introduce the
restricted class of \emph{simple priced timed games}
(Section~\ref{sec:simple-priced-timed}) and close this section by
discussing some special strategies (called \emph{switching
  strategies}) that we will rely upon in our algorithms to solve
priced timed games.

\subsection{Notations and definitions} As usual, we let $\N$, $\Z$,
$\Q$, $\R$, and $\Rplus$ be the set of non-negative integers,
integers, rational numbers, real numbers, and non-negative real
numbers respectively. We also let $\Rbar=\R\cup
\{+\infty,-\infty\}$. Let $x$ denote a non-negative real-valued variable
called \emph{clock}. A \emph{guard} (or \emph{clock constraint}) is an
interval with endpoints in $\Q\cup\{+\infty\}$. We often abbreviate
guards, writing for instance $x\leq 5$ instead of $[0,5]$. The set of
all guards on the clock $x$ is called $\Guard{x}$. Let
$S\subseteq \Guard{x}$ be a finite set of guards. We let
$\den{S}=\bigcup_{I\in S} I$. Assuming $M_0=0<M_1<\cdots<M_k$ are all
the endpoints of the intervals in~$S$ (to which we add $0$ if needed),
we let
\[\regcst{S}=\{(M_i,M_{i+1})\mid 0\leq i\leq k-1\} \cup \{ \{M_i\}\mid
  0\leq i\leq k\}\] be the set of \emph{regions} of $S$. Thus,
intuitively, a region of $S$ is either an open interval whose
endpoints are consecutive endpoints the intervals in $S$, or singletons
containing one endpoint from $S$. Observe that $\regcst{S}$ is also a
set of guards.

We rely on the notion of \emph{cost function} to formalise the notion
of optimal value function sketched in the introduction.  Formally, for
a set of guards $S\subseteq \Guard{x}$, a \emph{cost function} over
$S$ is a function $f\colon \den{\regcst{S}} \to \Rbar$ such that over
each region $r\in\regcst{S}$, $f$ is either infinite or it is a
continuous piecewise affine function with rational slopes and a finite set of 
rational cutpoints (points where the first derivative is not defined).
In particular, if
$f(r)=\{f(\valuation)\mid \valuation\in r\}$ contains $+\infty$
(respectively, $-\infty$) for some region $r$, then $f(r)=\{+\infty\}$
($f(r)=\{-\infty\}$). We denote by $\CF{S}$ the set of all cost
functions over $S$.

In our algorithm to solve \SPTG{s}, we will need
to combine cost functions thanks to the $\opcf$ operator. Let
$f\in\CF{S}$ and $f'\in\CF{S'}$ be two cost functions on sets of
guards $S,S'\subseteq \Guard{x}$, such that $\den S \cap \den{S'}$ is
a singleton. We let $f\opcf f'$ be the cost function in
$\CF{S\cup S'}$ such that $(f\opcf f')(\valuation)=f(\valuation)$ for
all $\valuation\in\den{\regcst S}$, and
$(f\opcf f')(\valuation)=f'(\valuation)$ for all
$\valuation\in\den{\regcst{S'}}\setminus\den{\regcst{S}}$.  For
example, let $S=\{\{0\},(0,1),\{1\}\}$ and $S'=\{\{1\}\}$. We define the cost
functions $f_1$ and~$f_2$ such that $f_1$ is equal to $+\infty$ on the
set of regions $\regcst{S}$ and $f_2$ is equal to $0$ on the set of
regions $\regcst{S'}$. The cost function $ f_2\opcf f_1\in\CF{S\cup S'}$ is
equal to $+\infty$ on $[0,1)$ and to $0$ on $\{1\}$ and the cost
function $f_1\opcf f_2\in\CF{S'}$ is equal to $+\infty$ on
$[0,1]$. Thus $f_1\opcf f_2$ is equal to $f_1$ while $ f_2\opcf f_1$
extends $f_2$ with a $+\infty$ value on $[0,1)$.

We consider an extended notion of one-clock priced timed games
(\PTG{s} for short) allowing for the use of \emph{urgent locations},
where only a zero delay can be spent, and \emph{final cost functions}
which are associated with all final locations and incur an extra cost
to be paid when ending the game in this location:

\begin{defi}
  A \emph{priced timed game} (\PTG for short) $\game$ is a tuple
  $(\LocsMin, \LocsMax, \LocsFin,\allowbreak \LocsUrg, \fgoalvec, \transitions,
  \price)$ where:
  \begin{itemize}
  \item $\LocsMin$ and $\LocsMax$ are finite sets of \emph{locations}
    belonging respectively to player $\MinPl$ and $\MaxPl$. 
    $\LocsFin$ is a finite set of \emph{final} locations.
    We assume that $\LocsMin$, $\LocsMax$ and $\LocsFin$ are disjoint and denote
    $\Locs = \LocsMin \uplus \LocsMax \uplus \LocsFin$ 
    the set of all locations of the \PTG;
  \item $\LocsUrg \subseteq \Locs\setminus\LocsFin$ is the set of
    \emph{urgent} locations\footnote{Here we differ
      from~\cite{BouLar06} where $\LocsUrg \subseteq
      \LocsMax$.}; 
  \item
    $\transitions \subseteq (\Locs\setminus \LocsFin)\times \Guard{x}
    \times \{\top,\bot\} \times \Locs$ is a finite set of
    \emph{transitions}. We denote by
    $\cstgame = \{I\mid \exists \loc, R, \loc': (\loc,I,R,\loc')\in
    \transitions\}$ the set of all guards occurring on some
    transitions of the \PTG;
  \item $\fgoalvec = (\fgoal_\loc)_{\loc\in\LocsFin}$ associates to
    all locations $\loc\in \LocsFin$ a \emph{final cost function},
    that is an affine\footnote{In our one-clock setting, an affine
      function is of the form $\fgoal_\loc(\nu)=a\times \nu +b$.} cost
    function $\fgoal_\loc$ with rational coefficients;
  \item $\price\colon (\Locs\setminus \LocsFin) \cup \transitions \to \Z$ is a mapping
    associating an integer weight to all non-final locations and transitions. 
  \end{itemize}
\end{defi}

Intuitively, a transition $(\loc,I,R,\loc')$ changes the current
location from $\loc$ to $\loc'$ if the clock has value in $I$ and the
clock is reset according to the Boolean $R$. We assume that, in all
\PTG{s}, the clock $x$ is \emph{bounded}, i.e.~there is $M\in\N$ such
that for all guards $I\in \cstgame$,
$I\subseteq [0,M]$.\label{sec:notat-defin}\footnote{This last
  restriction is \emph{not} without loss of generality in the case of
  \PTG{s}. While all timed automata $\mathcal{A}$ can be turned into
  an equivalent (with respect to reachability properties)
  $\mathcal{A}'$ whose clocks are bounded~\cite{BehFeh01}, this
  technique cannot be applied to \PTG{s}, in particular with arbitrary
  weights.} We denote by $\reggame$ the set $\regcst{\cstgame}$ of
\emph{regions of} $\game$.  We further denote\footnote{Throughout the
  paper, we often drop the $\game$ in the subscript of several
  notations when the game is clear from the context.} by
$\maxPriceTrans_\game$, $\maxPriceLoc_\game$ and $\maxPriceFin_\game$
respectively the values
$\max_{\transition\in\transitions}|\price(\transition)|$,
$\max_{\loc\in(\Locs\setminus \LocsFin)} |\price(\loc)|$ and
$\sup_{\valuation\in [0,M]} \max_{\loc\in\LocsFin}
|\fgoal_\loc(\valuation)|=\max_{\loc\in\LocsFin} \max(
|\fgoal_\loc(0)|,|\fgoal_\loc(M)|)$ (the last equality holds because
we have assumed that $\fgoal_\loc$ is affine). That is,
$\maxPriceTrans_\game$, $\maxPriceLoc_\game$ and $\maxPriceFin_\game$
are the largest absolute values of the transition weights, location
weights and final cost functions.

As announced in the first section, the semantics of a \PTG
$\game = (\LocsMin, \LocsMax, \LocsFin, \allowbreak\LocsUrg, \fgoalvec,
\transitions, \price)$ is given by a quantitative reachability game
\[
  G_\game=\big(C_\MinPl=(\LocsMin\times \Rplus), C_\MaxPl=(\LocsMax\times \Rplus),  F=(\LocsFin\times \Rplus), 
  \Sigma=(\Rplus\times \transitions\times\R),
  E, \puse\big)\] that we describe now. Note
that, from now on, we often confuse the \PTG $\game$ with its
semantics $G_\game$, writing, for instance `the configurations of
$\game$' instead of: `the configurations of $G_\game$'. We also lift
the functions $\costname$, $\uval$, $\lval$ and $\Val$, and the
notions of plays from $G_\game$ to $\game$. A \emph{configuration} of
$\game$ is a pair $s=(\loc,\valuation)\in \Locs \times \Rplus$, where
$\loc$ and $\valuation$ are respectively the current location and
clock value of $\game$. We denote by \confgame the set of all
configurations of~$\game$. 
Let $(\loc,\valuation)$ and
$(\loc',\valuation')$ be two configurations, let
$\transition=(\loc,I,R,\loc')\in\transitions$ be a transition
of~$\game$ and~$t\in\Rplus$ be a delay. Then,
$((\loc,\valuation),(t,\delta,c),(\loc',\valuation'))\in E$, iff:
\begin{enumerate}
\item $\loc\in\LocsUrg$ implies $t=0$ (no time can elapse in urgent
  locations);
\item $\valuation+t\in I$ (the guard is satisfied); 
\item $R=\top$ implies $\valuation'=0 $ (when the clock is reset); 
\item $R=\bot$ implies $\valuation'=\valuation+t$ (when the clock is
  not reset);
\item $c=\price(\transition)+t\times\price(\loc)$ (the cost of
  $(t,\transition)$ takes into account the weight of~$\loc$, the delay
  $t$ spent in $\loc$, and the weight of $\transition$).
\end{enumerate}
In this case, we say that there is a $(t,\transition)$-transition from
$(\loc,\valuation)$ to $(\loc',\valuation')$ with cost $c$, and we
denote this by
$(\loc, \valuation) \xrightarrow{t,\transition,c}(\loc',
\valuation')$.  For two configurations $s$ and $s'$, we also write
$s\xrightarrow{c} s'$ whenever there are $t$ and $\transition$ such
that $s\xrightarrow{t,\transition,c} s'$. Observe that, since the
alphabet of $G_\game$ is $\Rplus\times \transitions\times\R$, and its
set of configurations is $\confgame$, plays of $\game$ are of the form
$\rho=(\loc_1,\valuation_1)\xrightarrow{t_1, \transition_1,c_1}
(\loc_2,\valuation_2)\cdots$. Finally, the \pname function $\puse$ is
obtained by summing the costs of the play (transitions and time spent
in the locations) and the final cost function if applicable.
Formally, let
$\rho=(\loc_1,\valuation_1)\xrightarrow{t_1, \transition_1,c_1}
(\loc_2,\valuation_2) \cdots (\loc_n,\valuation_n)$ be a finite play
such that for all $k<n$, $\loc_k \notin \LocsFin$. Then,
$\puse(\rho)= \sum^{n-1}_{i=1} c_i + \fgoal_{\loc_{n}}(\valuation_n)$
if $\loc_n\in \LocsFin$, and $\puse(\rho)= \sum^{n-1}_{i=1} c_i$
otherwise.
%

As sketched in the introduction, we consider optimal reachability-price
games on \PTG{s}, where the aim of player $\MinPl$ is to reach a
location of $\LocsFin$ while minimising the price. Since the semantics
of \PTG{s} is defined in terms of quantitative reachability games, we
can apply Theorem~\ref{thm:determined}, and deduce that all \PTG{s}
$\game$ are determined. Hence, for all \PTG{s} the value function
$\Val$ is well-defined, and we denote it by $\valgame$ when we need to
emphasise the game it refers to.

For example, consider the \PTG on the left of
Figure~\ref{fig:ex-ptg2}. Using the final cost function~$\fgoal$
constantly equal to 0, its value function for location $\loc_1$ is
represented on the right.  The completed play
$\rho=(\loc_1,0)\xrightarrow{0, t_{1,2} ,0} (\loc_2,0)
\xrightarrow{1/4, t_{2,3},-3.5} (\loc_3,1/4) \xrightarrow{0,
  t_{3,7},6} (\loc_7,1/4) \xrightarrow{3/4, t_{7,f},-12} (\loc_f,1)$
where $t_{n,m}= (\loc_n,[0,1],\bot,\loc_m)$ ends in the unique final
location $\loc_f$ and its price is $\cost{\rho}=0-3.5+6-12=-9.5$.

Let us fix a \PTG $\game$ with initial
configuration $c_1$. We say that a strategy $\stratmin$ of $\MinPl$ is
\emph{optimal} if $\cost{c_1, \stratmin}=\valgame(c_1)$, i.e.~it
ensures $\MinPl$ to enforce the value of the game, whatever $\MaxPl$
does. Similarly, $\stratmin$ is $\varepsilon$-\emph{optimal}, for
$\varepsilon > 0$, if
$\cost{c_1, \stratmin}\leq\valgame(c_1)+\varepsilon$. And,
symmetrically, a strategy $\stratmax$ of $\MaxPl$ is \emph{optimal}
(respectively, $\varepsilon$-\emph{optimal}) if
$\cost{c_1, \stratmax}=\valgame(c_1)$ (respectively,
$\cost{c_1, \stratmax}\geq\valgame(c_1)-\varepsilon$).

\subsection{Properties of the value}\label{sec:properties-value}
Let us now discuss useful preliminary properties of the value
functions of \PTG{s}.  We have already shown the determinacy of the
game, ensuring the existence of the value function. We will now
establish a stronger (and, to the best of our knowledge, original)
result. For all locations $\ell$, let $\valgame(\loc)$ denote the
function such that
$\valgame(\loc)(\valuation)=\valgame(\loc,\valuation)$ for all
$\valuation\in\Rplus$. Then, we show that, for all~$\loc$,
$\valgame(\loc)$ is a \emph{piecewise continuous function} that might
exhibit discontinuities \emph{only on the borders of the regions}
of~$\reggame$.

\begin{thm}\label{prop:continuity-of-val} 
  For all (one-clock) \PTG{s} $\game$, for all $r\in \reggame$, for
  all $\loc\in\locs$, $\valgame(\loc)$ is either infinite or
  continuous over $r$.
\end{thm}
\begin{proof}
  The main ingredient of our proof is, given a
  strategy $\stratmin$ of \MinPl, a location $\loc$ of the game,
  a region $r \in \reggame$ and valuations
  $\valuation,\valuation'\in r$, to show how to build a strategy $\stratmin'$ and
  a function $g$ such that $g$ maps plays starting in $(\loc,\valuation')$ and 
  consistent with $\stratmin'$ to plays starting in $(\loc,\valuation)$ consistent
  with $\stratmin$ with similar behaviour and cost. 
  More precisely, we define $\stratmin'$ and $g$ by induction on
  the length of the finite play that is given as argument and rely on the following set of 
  induction hypothesis:
%

  \textbf{Induction hypothesis:} 
There exists a strategy $\stratmin'$ and a function $g$ from the plays of length (respectively) 
$k-1$ and $k$, starting in $(\loc,\valuation')$ and consistent with $\stratmin'$, 
to (respectively) transitions in 
$\game$ and plays starting in $(\loc,\valuation)$, consistent with $\stratmin$, 
such that for all plays
  $\run'=(\loc_1,\valuation'_1)\xrightarrow{c'_1}\cdots\xrightarrow{c'_{k-1}}
  (\loc_k,\valuation'_k)$  starting in $(\loc,\valuation')$ and consistent with
  $\stratmin'$, denoting
  $(\loc_1,\valuation_1)\xrightarrow{c_1}\cdots\xrightarrow{c_{{k'}-1}}
  (\loc_{{k'}},\valuation_{{k'}})$ the play $g(\run')$ starting in $(\loc, \valuation)$ we have:
  \begin{enumerate}
  \item \label{item:1}$\run'$ and $g(\run')$ have the same length,
    i.e.~$k'=k$,
  \item \label{item:2} for every
    $i\in\{1,\ldots,k\}$, $\valuation_i$ and $\valuation'_i$ are in
    the same region, i.e.~there exists a region $r'\in \reggame$ such
    that $\valuation_i\in r'$ and $\valuation'_i\in r'$,
  \item \label{item:3}
    $|\valuation_k-\valuation'_k|\leq |\valuation-\valuation'|$,
  \item \label{item:4}
    $\puse{(\run')} \leq \puse{(g(\run'))} +
    \maxPriceLoc(|\valuation-\valuation'|
    -|\valuation_k-\valuation'_k|)$.
  \end{enumerate}
  Notice that no property is required on the strategy $\stratmin'$ for
  finite plays that do not start in~$(\loc,\valuation')$.

 Let us explain how this result would imply the theorem before going through the induction itself.
 Let $r\in \reggame$ be a region of the game and $\loc$ be a location.
 Remark first that the result directly implies that if the value of the game is finite for some 
 valuation $\valuation$ in $r$, then it is finite for all other valuation $\valuation'$ in $r$. 
 Indeed, a finite value of the game in $(\loc,\valuation)$
 implies that there exists a strategy $\stratmin$ such that
 every play consistent with it and starting in $(\loc,\valuation)$ reaches a final location
 with a time valuation such that the final cost function is finite. Moreover, denoting $\stratmin'$ the strategy obtained
 from $\stratmin$ thanks to the above result, any play $\rho'$ starting in $(\loc,\valuation')$
 and consistent with $\stratmin'$
 reaches a final location (since $g(\rho')$ does) and the final cost function is finite as 
 the final time valuation of $\rho'$ and $g(\rho')$ sit in the same region and, by 
 definition, a final cost function is either always finite or always infinite within a region.

 Now, assuming the value of the game is finite over $r$, in order to show that 
 $\valgame(\loc)$ is continuous over $r$, we need to show that,
  for all $\valuation\in r$, for all $\varepsilon>0$, there exists
  $\delta>0$ s.t.\ for all $\valuation' \in r$ with
  $|\valuation-\valuation'|\leq \delta$, we have
  $|\val(\loc,\valuation)-\val(\loc,\valuation')|\leq \varepsilon$. To
  this end, we can show that:
  \begin{align}
    \label{eq:1}
    |\val(\loc,\valuation)-\val(\loc,\valuation')| &\leq
    (\maxPriceLoc + K_{fin})|\valuation-\valuation'|
  \end{align}
  where $K_{fin}$ is the greatest absolute value of the slopes appearing
  in the piecewise affine functions within $\fgoalvec$.
  Indeed, assume that this inequality holds, and consider a clock
  value $\valuation\in r$ and a positive real number
  $\varepsilon$. Then, we let
  $\delta=\frac{\varepsilon}{\maxPriceLoc+ K_{fin}}$, and we consider a
  valuation $\valuation'$ s.t.\ $|\valuation -
  \valuation'|\leq\delta$. In this case, equation~\eqref{eq:1}
  becomes:
  \[|\val(\loc,\valuation)-\val(\loc,\valuation')| \leq
    (\maxPriceLoc+ K_{fin})|\valuation-\valuation'|\leq (\maxPriceLoc+ K_{fin})
    \frac{\varepsilon}{\maxPriceLoc+ K_{fin}}\leq\varepsilon\,.\]

  Thus, proving equation~\eqref{eq:1} is sufficient to establish
  continuity.  On the other hand, equation~\eqref{eq:1} is
  equivalent to:
  \[\val(\loc,\valuation)\leq\val(\loc,\valuation') +
    (\maxPriceLoc+ K_{fin})|\valuation-\valuation'|\quad \textrm{and}\quad
    \val(\loc,\valuation')\leq\val(\loc,\valuation) +
    (\maxPriceLoc+ K_{fin})|\valuation-\valuation'|\,.\]
  As those two last equations are symmetric with respect to
  $\valuation$ and $\valuation'$, we only have to show either of
  them. 
  We thus focus on the latter, which, by using the upper
  value, can be reformulated as: for all strategies $\stratmin$ of
  \MinPl, there exists a strategy $\stratmin'$ such that
  $\cost{(\loc,\valuation'),\stratmin'}\leq \cost{(\loc,\valuation),
    \stratmin}+(\maxPriceLoc+ K_{fin})|\valuation-\valuation'|$.  Note that this
  last equation is equivalent to say that there exists a function $g$
  mapping plays $\rho'$ from $(\loc,\valuation')$, consistent with
  $\stratmin'$ (i.e.~such that
  $\rho'=\Play{(\loc,\valuation'),\stratmin',\stratmax}$ for some
  strategy $\stratmax$ of $\MaxPl$) to plays from~$(\loc,\valuation)$,
  consistent with $\stratmin$, such that, for all such $\run'$ the final time valuations
  of $\run'$ and $g(\run')$ differ by at most $|\valuation-\valuation'|$ and:
  \[\puse{(\run')}\leq
    \puse{(g(\run'))}+\maxPriceLoc|\valuation-\valuation'|\]
which is exactly what our claimed induction achieves. 
Thus, to conclude this proof, let us now define $\stratmin'$ and $g$, by induction
  on the length $k$ of $\run'$.

  \textbf{Base case $k=1$: } In this case, $\stratmin'$ does not have
  to be defined. Moreover, in that case, $\run' = (\loc,\valuation')$
  and $g(\run')=(\loc,\valuation)$. Both plays have length~$1$,
  $\valuation$ and $\valuation'$ are in the same region by 
  hypothesis, and $\puse{(\run')} = \puse{(g(\run'))} =0$, therefore all
  four properties are true.

  \textbf{Inductive case: } Let us suppose now that the construction
  is done for a given $k\geq 1$, and perform it for $k+1$. We start
  with the construction of $\stratmin'$. To that extent, consider a
  play
  $\run'=(\loc_1,\valuation'_1)\xrightarrow{c'_1}\cdots\xrightarrow{c'_{k-1}}
  (\loc_k,\valuation'_k)$ from $(\loc,\valuation')$, consistent with
  $\stratmin'$ (provided by induction hypothesis)
  such that $\loc_{k}$ is a location of player
  $\MinPl$. Let $t$ and $\transition$ be the choice of delay and
  transition made by $\stratmin$ on $g(\run')$,
  i.e.~$\stratmin(g(\run'))=(t,\transition)$. Then, we define
  $\stratmin'(\run')=(t',\transition)$ where
  $t' = \max(0, \valuation_{k}+t - \valuation'_{k})$.  The delay $t'$
  respects the guard of transition $\transition$, as can be seen from
  \figurename~\ref{fig:deltaPrime}. Indeed, either
  $\valuation_{k}+t = \valuation'_{k}+t'$ (cases (a) and (b) in the
  figure) or
  $\valuation_{k}\leq \valuation_{k}+t\leq \valuation'_{k}$, in which
  case $\valuation'_{k}$ is in the same region as $\valuation_{k}+t$
  since $\valuation_k$ and $\valuation'_k$ are in the same region by
  induction hypothesis.

  \begin{figure}
    \centering
    \begin{tikzpicture}
      \begin{scope}
        \draw[->,dashed] (-0.5,0) -- (2.5,0);  
        \node  at (0,-0.3) {$\valuation'_{k}$};
        \node (vkp) at (0,0) {$\bullet$};
        \node  at (1,0.3) {$\valuation_{k}$};
        \node (vk) at (1,0) {$\bullet$};
        \node (vkdelta) at (2,0) {$\bullet$};
        
        \node at (1,-1.5) {(a)};
        
        \draw[->] (vk) to[bend left] node[midway, above] {$t$} (vkdelta);
        \draw[->] (vkp) to[bend right]  node[midway, below]  {$t'$} (vkdelta);
      \end{scope}
      
      \begin{scope}[xshift = 4cm]
        \draw[->,dashed] (-0.5,0) -- (2.5,0);  
        \node  at (0,-0.3) {$\valuation_{k}$};
        \node (vk) at (0,0) {$\bullet$};
        \node  at (1,0.3) {$\valuation'_{k}$};
        \node (vkp) at (1,0) {$\bullet$};
        \node (vkdelta) at (2,0) {$\bullet$};
        
        \node at (1,-1.5) {(b)};
        
        \draw[->] (vk) to[bend right] node[midway, below] {$t$} (vkdelta);
        \draw[->] (vkp) to[bend left]  node[midway,above]  {$t'$} (vkdelta);
      \end{scope}

      \begin{scope}[xshift = 8cm]
        \draw[->,dashed] (-0.5,0) -- (2.5,0);  
        \node  at (0,-0.3) {$\valuation_{k}$};
        \node (vk) at (0,0) {$\bullet$};
        \node (vkdelta) at (1,0) {$\bullet$};
        \node  at (2,-0.3) {$\valuation'_{k}$};
        \node (vkp) at (2,0) {$\bullet$};
        
        \node at (1,-1.5) {(c)};
        
        \draw[->] (vk) to[bend right] node[midway, below] {$t$} (vkdelta);
        \draw[->]  (vkp) edge[loop above] node[midway, above] {$t'$} (vkp);
      \end{scope}
      
    \end{tikzpicture}
    
    \caption{The definition of $t'$ when (a)
      $\valuation'_k\leq \valuation_k$, (b)
      $\valuation_k < \valuation'_k < \valuation_k+t$, (c)
      $\valuation_k < \valuation_k+t < \valuation'_k $.}
    \label{fig:deltaPrime}
  \end{figure}

  Let us now build the mapping $g$. Let
  $\run'=(\loc_1,\valuation'_1)\xrightarrow{c'_1}\cdots\xrightarrow{c'_{k}}
  (\loc_{k+1},\valuation'_{k+1})$ be a play from~$(\loc,\valuation')$
  consistent with $\stratmin'$ and let
  $\tilde{\run}'
  =(\loc_1,\valuation'_1)\xrightarrow{c'_1}\cdots\xrightarrow{c'_{k-1}}
  (\loc_k,\valuation'_k)$ its prefix of length~$k$.
  Let~$(t', \transition)$ be the delay and transition taken after
  $\tilde{\run}'$. Using the construction of~$g$ over plays of length
  $k$ by induction, the play
  $g(\tilde{\run}') =
  (\loc_1,\valuation_1)\xrightarrow{c_1}\cdots\xrightarrow{c_{k-1}}
  (\loc_k,\valuation_k)$ (with
  $(\loc_1,\valuation_1)=(\loc,\valuation)$) verifies properties
  (\ref{item:1}), (\ref{item:2}), (\ref{item:3}) and (\ref{item:4}). Then:
  \begin{itemize}
  \item if $\loc_k$ is a location of \MinPl and
    $\stratmin(g(\tilde{\run}')) = (t,\transition)$, then
    $g(\run')=g(\tilde{\run}')\xrightarrow{c_{k}}
    (\loc_{k+1},\valuation_{k+1})$ is obtained by applying those
    choices on $g(\tilde{\run}')$;
  \item if $\loc_k$ is a location of \MaxPl, the last clock value
    $\valuation_{k+1}$ of $g(\run')$ is rather obtained by choosing
    action $(t,\transition)$ verifying
    $t = \max (0, \valuation'_{k}+t'-\valuation_k)$. Note that
    transition $\transition$ is allowed since both $\valuation_k+t$
    and~$\valuation'_k + t'$ are in the same region (for similar
    reasons as above).
  \end{itemize}

  By induction hypothesis $|\tilde{\run}'| = |g(\tilde{\run}')|$,
  thus:~\ref{item:1} holds, i.e.~$|\run'|=|g(\run')|$. Moreover,
  $\valuation_{k+1}$ and $\valuation'_{k+1}$ are also in the same
  region as either they are equal to $\valuation_k+t$ and
  $\valuation'_k + t'$, respectively, or $\transition$ contains a
  reset in which case $\valuation_{k+1}=\valuation'_{k+1}=0$ which
  proves~\ref{item:2}. To prove~\ref{item:3}, notice that we always
  have either $\valuation_k+t = \valuation'_k + t'$ or
  $\valuation_k\leq \valuation_k+t \leq
  \valuation'_k=\valuation'_k+t'$
  or
  $\valuation'_k\leq \valuation'_k+t \leq
  \valuation_k=\valuation_k+t$.
  In all of these possibilities, we have
  $|(\valuation_k+t) - (\valuation'_k+t')|\leq
  |\valuation_k-\valuation'_k|$.
  We finally check property~\ref{item:4}. In both
  cases:
  \begin{align*}
  \puse{(\run')} &= \puse{(\tilde{\run}')}+\price(\transition)+t' \price(\loc_k) \\
               &\leq
                 \puse{(g(\tilde{\run}'))}+\maxPriceLoc(|\valuation-\valuation'|
                 -|\valuation_k-\valuation'_k|)+\price(\transition)+t'
                 \price(\loc_k)\\ 
               &=\puse{(g(\run'))}
                 +(t'-t)\price(\loc_k)+\maxPriceLoc(|\valuation-\valuation'|
                 -|\valuation_k-\valuation'_k|)\,. 
  \end{align*}
  Let us prove that
  \begin{equation}
    |t'-t|\leq |\valuation_k-\valuation'_k| -
    |\valuation'_{k+1}-\valuation_{k+1}|\,.\label{eq:t'-t}
  \end{equation}
  If $\transition$ contains no reset,
  $t' = \valuation'_{k+1}-\valuation'_{k}$ and
  $t = \valuation_{k+1}-\valuation_{k}$, we have
  $|t'-t|
  =|\valuation'_{k+1}-\valuation'_{k}-(\valuation_{k+1}-\valuation_{k})|$.
  Then, two cases are possible: either
  $t' = \max(0, \valuation_{k}+t - \valuation'_{k})$ or
  $t = \max (0, \valuation'_{k}+t'-\valuation_k)$. So we have three
  different possibilities:
  \begin{itemize}
  \item if $t'+\valuation'_k=t +\valuation_k$, then
    $\valuation'_{k+1}=\valuation_{k+1}$, thus
    $|t'-t|= |\valuation_k-\valuation'_k| =
    |\valuation_k-\valuation'_k| -
    |\valuation'_{k+1}-\valuation_{k+1}|$.
  \item if $t=0$, then
    $\valuation_k=\valuation_{k+1}\geq \valuation'_{k+1}\geq
    \valuation'_k$,
    thus
    $|\valuation'_{k+1}-\valuation'_{k}-(\valuation_{k+1}-\valuation_{k})|=
    \valuation'_{k+1}-\valuation'_k = (\valuation_{k}-\valuation'_k) -
    (\valuation_{k+1}-\valuation'_{k+1}) = |\valuation_k-\valuation'_k| -
    |\valuation'_{k+1}-\valuation_{k+1}|$.
  \item if $t'=0$, then
    $\valuation'_k=\valuation'_{k+1}\geq \valuation_{k+1}\geq
    \valuation_k$,
    thus
    $|\valuation'_{k+1}-\valuation'_{k}-(\valuation_{k+1}-\valuation_{k})|=
    \valuation_{k+1}-\valuation_k = (\valuation'_{k}-\valuation_k) -
    (\valuation'_{k+1}-\valuation_{k+1}) = |\valuation_k-\valuation'_k| -
    |\valuation'_{k+1}-\valuation_{k+1}|$.
  \end{itemize}
  If $\transition$ contains a reset, then
  $\valuation'_{k+1}=\valuation_{k+1}$. If
  $t' = \valuation_{k}+t - \valuation'_{k}$, we have that
  $|t'-t| = |\valuation_k-\valuation'_k| $. Otherwise, 
  $t'=0$ and
  $t \leq \valuation'_k-\valuation_k$.
  In all cases, we have proved~\eqref{eq:t'-t}.

  Together with the fact
  that $|\price(\loc_k)|\leq \maxPriceLoc$, we conclude that:
  \[\puse{(\run')}\leq \puse{(g(\run'))}
  +\maxPriceLoc(|\valuation-\valuation'|
  -|\valuation_{k+1}-\valuation'_{k+1}|)\,.\]

  Now that $\stratmin'$ and $g$ are defined (noticing that $g$ is
  stable by prefix, we extend naturally its definition to infinite
  plays), notice that for all plays $\run'$ from $(\loc,\valuation')$
  consistent with $\stratmin'$, either $\run'$ does not reach a final
  location and its price is $+\infty$, but in this case $g(\run')$ has
  also price $+\infty$; or $\run'$ is finite. In this case, let
  $\valuation'_k$ be the clock value of its last configuration,
  and $\valuation_k$ be the clock value of the last configuration
  of $g(\run')$. Combining (\ref{item:3}) and (\ref{item:4}) we have
  $\puse{(\run')}\leq
  \puse{(g(\run'))}+\maxPriceLoc|\valuation-\valuation'|$
  which concludes the induction.
\end{proof}

\begin{figure}[tbp]
    \begin{center}
      \begin{tikzpicture}[node distance=3cm,minimum size=5mm]
        \node[draw,circle,label=above:$5$] (q0)
        {\makebox[0pt][c]{$\loc_0$}};%
        \node[draw,circle,below left of=q0, label=left:$-5$] (c2)
        {\makebox[0pt][c]{$\loc_2$}};%
        \node[draw,circle,right of=q0,accepting] (q1)
        {\makebox[0pt][c]{$\loc_f$}};%
        \node[draw,circle,below right of=q0,label=right:$5$] (c1)
        {\makebox[0pt][c]{$\loc_1$}};
        
        \path[->] (q0) edge node[above] {$x=0$} (q1)%
        (q0) edge node[above right] {$x=0$} (c1) %
        (c1) edge node[below] {$y=1, y:=0$} (c2) %
        (c2) edge node[above left] {$x=1, x:=0$} (q0);
      \end{tikzpicture}
    \end{center}
    \caption{A PTG with 2 clocks whose value function is not
      continuous inside a region.}
    \label{fig:pas-continu2}
  \end{figure}
\begin{rem}
  \label{app:pas-continu2}
  Let us consider the example in \figurename~\ref{fig:pas-continu2}
  (that we describe informally since we did not properly define games
  with multiple clocks), with clocks $x$ and $y$.  One can easily
  check that, starting from a configuration $(\loc_0,0,0.5)$ in
  location $\loc_0$ and where $x=0$ and $y= 0.5$, the following cycle
  can be taken:
  $(\loc_0, 0, 0.5)\xrightarrow{0,\transition_0, 0}(\loc_1,0,
  0.5)\xrightarrow{0.5, \transition_1, 2.5}(\loc_2,0.5,0)
  \xrightarrow{0.5, \transition_2, -2.5}(\loc_0, 0, 0.5)$, %
  where $\transition_0$, $\transition_1$ and $\transition_2$ denote
  respectively the transitions from $\loc_0$ to $\loc_1$; from
  $\loc_1$ to $\loc_2$; and from $\loc_2$ to $\loc_0$.  Observe that
  the \pname of this cycle is null, and that no other delays can be
  played, hence $\uval(\loc_0,0,0.5)=0$.  However, starting from a
  configuration $(\loc_0, 0, 0.6)$, and following the same path,
  yields the cycle %
  $(\loc_0, 0, 0.6)\xrightarrow{0,\transition_0, 0}(\loc_1, 0,
  0.6)\xrightarrow{0.4, \transition_1, 2}(\loc_2,0.4,0) \xrightarrow{0.6, \transition_2,
    -3}(\loc_0, 0, 0.6)$ %
  with \pname $-1$. Hence, $\uval(\loc_0,0,0.6)=-\infty$, and the
  function is not continuous although both clocks values $(0, 0.5)$ and
  $(0, 0.6)$ are in the same region. Observe that this holds even for
  priced timed \emph{automata}, since our example requires only one
  player.
\end{rem}

\subsection{Simple priced timed games}\label{sec:simple-priced-timed}
As sketched in the introduction, our main contribution is to solve the
special case of simple one-clock priced timed games with arbitrary
weights, where the clock is never reset and takes values in some fixed
interval $[0,r]$ only.
Formally, an $\rightpoint$-\SPTG, with
$\rightpoint\in \Q^+\cap [0,1]$, is a \PTG
$\game = (\LocsMin, \LocsMax, \LocsFin, \LocsUrg, \fgoalvec,
\transitions, \price)$ such that for all transitions
$(\loc,I,R,\loc')\in \transitions$, $I=[0,\rightpoint]$ (the clock is
also bounded by $r$) and $R=\bot$. Hence, transitions of
$\rightpoint$-\SPTG{s} are henceforth denoted by $(\loc,\loc')$,
dropping the guard and the reset. Then, an \SPTG is a $1$-\SPTG.  This
paper is mainly devoted to prove the following result on \SPTG{s}.
\begin{thm}\label{thm:main-result}
  Let $\game$ be an \SPTG. Then, for all locations $\loc \in \locs$,
  either $\val(\loc)\in\{-\infty,+\infty\}$, or $\val(\loc)$ is
  continuous and piecewise-affine with at most a pseudo-polynomial
  number of cutpoints (in the size of \game), i.e.~polynomial if the
  prices of the game are encoded in unary. The value functions
  $\val(\loc)$ for all locations $\ell$, as well as a pair of optimal
  strategies $(\stratmin,\stratmax)$ (that always exist when no values
  are infinite) can be computed in pseudo-polynomial time.
\end{thm}


\subsubsection{Proof strategy} Let us now highlight the main steps
that will allow us to establish this theorem. The central argument
consists in showing that all \SPTG{s} admit `well-behaved' optimal
strategies for both players, in the sense that these strategies can be
finitely described (and computed in pseudo-polynomial time). To this
end, we rely on several new definitions that we are about to introduce
and that we first describe informally.

We start by the case of \MaxPl: we will show that \MaxPl always has a
\emph{positional} (aka memoryless) optimal strategy. However, this is
not sufficient to show that \MaxPl has an optimal strategy that can be
\emph{finitely} described: indeed, in the case of \SPTG{s}, a
positional strategy associates a move to each \emph{configuration} of
the game, and there are uncountably many such configurations because
of the possible values of the clock. Thus, we introduce the notion of
\emph{finite positional strategies} (FP-strategies for short). Such
strategies partition the set $[0,1]$ of possible clock values into
finitely many intervals, and ensure that the same move is played
throughout each interval: this move can be either to \emph{wait} until
the clock reaches the end of the interval, or to \emph{take
  immediately} a given transition.

The case of \MinPl is more involved, as shown in the
next example taken from~\cite{BGHM16}.  Consider the \SPTG of
\figurename~\ref{fig:Weighted-game}, where $W$ is a positive integer,
and every location has weight~$0$ (thus, it is an \emph{untimed} game,
as originally studied). We claim that the values of locations $\loc_1$
and $\loc_2$ are both $-W$. Indeed, consider the following strategy
for $\MinPl$: during each of the first $W$ visits to $\loc_2$ (if
any), go to $\loc_1$; else, go to $\loc_f$. Clearly, this strategy
ensures that the final location $\loc_f$ will eventually be reached,
and that: 
\begin{enumerate}
\item either transition $(\loc_1,\loc_f)$ (with weight~$-W$) will
  eventually be traversed; 
\item or transition $(\loc_1,\loc_2)$ (with weight~$-1$) will be
  traversed at least~$W$ times.
\end{enumerate}
Hence, in all plays following this strategy, the price will be at
most~$-W$. This strategy allows $\MinPl$ to secure~$-W$, but \he
cannot ensure a lower price, since $\MaxPl$ always has the opportunity
to take the transition $(\loc_1,\loc_f)$ (with \pname~$-W$) instead of
cycling between $\loc_1$ and~$\loc_2$. Hence, $\MaxPl$'s optimal
choice is to follow the transition $(\loc_1,\loc_f)$ as soon as
$\loc_1$ is reached, securing a price of $-W$. The strategy we have
just given is optimal for \MinPl, and there are \emph{no optimal
  memoryless strategies} for $\MinPl$. Indeed, always playing
$(\loc_2,\loc_f)$ does not ensure a price at most $-W$; and, always
playing $(\loc_2,\loc_1)$ does not guarantee to reach the target, and
this strategy has thus value $+\infty$.
\begin{figure}
\begin{center}
  \begin{tikzpicture}[node distance=2cm,minimum size=5mm]
      \node[draw,rectangle](1){\makebox[0mm][c]{$\loc_1$}}; 
      \node[draw,circle](2)[right of=1,node distance=3cm]{\makebox[0mm][c]{$\loc_2$}}; 
      \node[draw,circle,accepting](3)[below of=1,xshift=1.5cm]{\makebox[0mm][c]{$\loc_f$}};
      

      \path[->]
      (1) edge node[below left]{$-W$} (3)
      (1) edge[bend left=10] node[above]{$-1$} (2) 
      (2) edge[bend left=10] node[below]{$0$} (1)
      (2) edge node[below right]{$0$} (3)
      (3) edge[loop right] node[right]{$0$} (3);
    \end{tikzpicture}
\end{center}
\caption{An \SPTG where \MinPl needs memory to play optimally}
\label{fig:Weighted-game}
\end{figure}

This example shows the kind of strategies that we will prove
sufficient for \MinPl to play optimally: \emph{first} play an
FP-strategy to obtain a play prefix with a sufficiently low cost, by
forcing negative cycles (if any); \emph{second} play another
FP-strategy that ensures that the target will eventually be
reached. Such strategies have been introduced as \emph{switching
  strategies} in~\cite{BGHM16}, and can be finitely described by a
pair $(\stratmin^1,\stratmin^2)$ of FP-strategies and a threshold~$K$
to trigger the switch when the length 
of the play prefix is at least $K$.

Computing the latter of these two strategies is easy: $\stratmin^2$ is
basically an \emph{attractor} strategy, which guarantees \MinPl to
reach the target (when possible) at a bounded cost. Thus, the main
difficulty in identifying \emph{optimal switching strategies} is to
characterise $\stratmin^1$. To do so, we further introduce the notion
of \emph{negative cycle strategies} (NC-strategies for short). Those
strategies are FP-strategies which guarantee that all cycles taken
have cost of $-1$ at most, without necessarily guaranteeing to
eventually reach the target (as this will be taken care of by
$\stratmin^2$). Among those NC-strategies, we identify so-called
\emph{fake optimal} strategies. Those are the NC-strategies that
guarantee \MinPl to obtain the optimal value (or better) when the
target is reached, but do not necessarily guarantee to reach the
target (they are thus not \emph{really} optimal, hence the name
fake-optimal).

Based on these definitions, we will show that \emph{all \SPTG{s}} with
only finite values admit such optimal switching strategies for \MinPl
and optimal FP-strategies for \MaxPl. By definition, these strategies
can be finitely described (as a matter of fact, we will show that we
can compute them in pseudo-polynomial time). Let us now give the
formal definitions of those notions.

\subsubsection{Finite positional strategies} 

We start with the notion of
finite positional strategies, that will formalise a class of optimal
strategies for \MaxPl:
\begin{defi}[FP-strategies]
  A strategy $\strat$ is a \emph{finite positional strategy}
  (FP-strategy for short) iff it is a memoryless strategy (i.e.~for
  all finite plays $\rho_1 = \rho_1' \xrightarrow{c_1} s$ and
  $\rho_2 = \rho_2' \xrightarrow{c_2} s$ ending in the same
  configuration, we have $\strat(\rho_1) = \strat(\rho_2)$) and for
  all locations $\loc$, there exists a finite sequence of rationals
  $0=\valuation^\loc_0< \valuation^\loc_1<\valuation^\loc_2<\cdots <
  \valuation^\loc_k=1$
  and a finite sequence of transitions
  $\transition_0,\ldots,\transition_{2k}\in\transitions$ such that
  \begin{enumerate}
  \item for all $1\leq i \leq k$ and
    $\valuation\in (\valuation^\loc_{i-1},\valuation^\loc_{i})$,
    either $\strat(\loc,\valuation) = (0,\transition_{2i-1})$, or
    $\strat(\loc,\valuation) = (\valuation^\loc_i-\valuation,
    \transition_{2i-1})$;
    \item for all $0\leq i \leq k-1$,
    either $\strat(\loc,\valuation^\loc_i) = (0,\transition_{2i})$, or
    $\strat(\loc,\valuation^\loc_i) = (\valuation^\loc_{i+1}-\valuation^\loc_i,
    \transition_{2i})$; and
    \item $\strat(\loc,\valuation^\loc_k) = (0,\transition_{2k})$.
  \end{enumerate}
\end{defi}
We let $\points(\strat)$ be the set of $\valuation^\loc_i$ for all
$\loc$ and $i$, and $\intervals(\strat)$ be the set of all successive
open intervals and singletons generated by $\points(\strat)$.
Finally, we let $|\strat|=|\intervals(\strat)|$ be the size
of~$\strat$. Intuitively, in each location $\loc$ and interval
$(\valuation^\loc_{i-1},\valuation^\loc_i)$, $\strat$ always returns
the same move: either to take \emph{immediately} $\transition_{2i-1}$
or to wait until the clock reaches the endpoint $\valuation^\loc_i$
and then take~$\transition_{2i-1}$ (point 1 of the definition above).
A similar behaviour also happens on the endpoints (point 2).

\subsubsection{Switching strategies} On top of the definition of
FP-strategies, we can now define the notion of switching strategy:

\begin{defi}[Switching strategies]\label{def:switchstrat}
  A switching strategy is described by a pair
  $(\stratmin^1,\stratmin^2)$ of FP-strategies and a switch threshold
  $K$. It consists in playing $\stratmin^1$ until 
  the play contains $K$ transitions (i.e.~the length of the play 
  prefix is at least $K+1$); and then to
  \emph{switch} to strategy~$\stratmin^2$.
\end{defi}
The role of $\stratmin^2$ is to ensure reaching a final location: it
is thus a (classical) attractor strategy. The role of $\stratmin^1$,
on the other hand, is to allow \MinPl to decrease the cost low enough
(possibly by forcing negative cycles) to secure a price 
sufficiently low
, and the computation of
$\stratmin^1$ is thus the critical point in the computation of an
optimal switching strategy. In the \SPTG of
\figurename~\ref{fig:Weighted-game}, for example, $\stratmin^1$ is the
strategy that goes from $\loc_2$ to $\loc_1$, $\stratmin^2$ is the
strategy going directly to $\loc_f$ and the switch occurs after the
threshold of $K=2W$
. The value of the game under this strategy is
thus $-W$.

\subsubsection{Negative cycle strategies} To characterise $\stratmin^1$,
we introduce now the notion of negative cycle strategy (NC-strategy):
\begin{defi}[Negative cycle strategies]
  An NC-strategy $\stratmin$ of \MinPl is an FP-strategy such that
  for all plays
  $\rho = (\loc_1,\valuation)\xrightarrow{c_1} \cdots
  \xrightarrow{c_{k-1}} (\loc_k,\valuation')\in \Play{\stratmin}$ with
  $\loc_1=\loc_k$, and $\valuation,\valuation'$ in the same interval
  of $\intervals(\stratmin)$, the sum of weights of \emph{discrete
    transitions} is at most $-1$,
  i.e.~$\price(\loc_1,\loc_2)+\cdots+\price(\loc_{k-1},\loc_k) \leq
  -1$.
\end{defi}

Let us now show that this definition allows one to find an upper bound
on the \pnames of the plays following such an NC-strategy
$\stratmin$.
\begin{lem}\label{lem:bound-price-run}
  Let $\stratmin$ be an NC-strategy, and let $\run\in\Play{\stratmin}$
  be a finite play. Then:
  \begin{align*}
  \puse{(\run)} &\leq \maxPriceLoc +
  (2|\stratmin|-1)\times|\locs|\maxPriceTrans
  -\frac{|\run|}{|\locs|} +2|\stratmin|\,. 
\end{align*}
\end{lem}
\begin{proof}
  We start by proving a bound on the \pname of a finite play
  $\tilde{\run} \in \Play{\stratmin}$ such that all clock values are in
  the same interval $I$ of $\intervals(\stratmin)$ and ending in $(\loc_f,\nu_f)$. In this case, 
  we claim that:
  \begin{align}
    \puse{(\tilde{\run})} &\leq |I|\maxPriceLoc + |\locs|\maxPriceTrans
                  -\Big\lfloor(|\tilde{\run}|-1)/|\locs|\Big\rfloor\, \label{eq:3}
  \end{align}
  (where $\lfloor \cdot \rfloor$ is the floor function).
  Indeed, 
  the \pname of~$\tilde{\run}$ is the sum of the weights generated
  while spending time in locations, plus the discrete weights of
  taking transitions. The former is bounded by $|I|\maxPriceLoc$ since
  the total time spent in locations is bounded by $|I|$. For the
  discrete weights, one can delete from $\tilde{\run}$ at least
  $\lfloor(|\tilde{\run}|-1)/|\locs|\rfloor$ cycles, which all have a
  weight bounded above by $-1$, since $\stratmin$ is an
  NC-strategy. After removing those cycles from $\tilde{\run}$, one
  ends up with a play of length at most~$|\locs|$ (otherwise, the same
  location would be present twice and the remaining play would still
  contain a cycle). This ensures that the total \pname of all
  transitions is bounded by
  $|\locs|\maxPriceTrans
  -\lfloor(|\tilde{\run}|-1)/|\locs|\rfloor$. Hence the bound given
  above.

  Then, we consider the general case of a finite play
  $\run\in\Play{\stratmin}$ that might cross several intervals. We
  achieve this by splitting the play along intervals of
  $\intervals(\stratmin)$. Let~$I_1,I_2,\ldots,I_{k}$ be the intervals
  of $\intervals(\stratmin)$ visited during $\run$ (with
  $k\leq |\stratmin|$). We split $\run$ into $k$ plays $\run_1$,
  $\run_2,\ldots, \run_k$ such that
  $\run = \run_1 \xrightarrow{c_1} \run_2 \xrightarrow{c_2}
  \cdots\run_{k}$; and, for all $i$, all clock values along~$\run_i$
  are in~$I_i$ (remember that \SPTG{s} contain no reset transitions).
  Then, we have:
  \begin{align}
    \puse{(\run)} &= \sum_{i=1}^{k}\puse{(\run_i)} + \sum_{i=1}^{k-1} c_i\,.\label{eq:5}
  \end{align}
  Let us bound these two sums. We start with the rightmost one. Since
  $c_i\leq \maxPriceTrans$ for all $i$, and since $k\leq |\stratmin|$,
  we have:
  \begin{align}
    \label{eq:4}
      \sum_{i=1}^{k-1} c_i &\leq (k-1) \maxPriceTrans\leq
      (|\stratmin|-1) \maxPriceTrans\,.
  \end{align}
  Now let us bound the leftmost sum in~\eqref{eq:5}. Using~\eqref{eq:3},
  we obtain:
  \begin{align}
    \sum_{i=1}^{k}\puse{(\run_i)} &\leq  \maxPriceLoc
                                  \sum_{i=1}^{k}|I_i| +
                                  \sum_{i=1}^{k}|\locs|\maxPriceTrans
                                  -\sum_{i=1}^{k}\lfloor(|\run_i|-1)/|\locs|\rfloor\label{eq:11}
  \end{align}
  Now, we can further bound these three new sums. 
  The
  intervals $I_i$ are consecutive, hence $\sum_{i=1}^{k}|I_i|\leq
  1$. Next,
  $\sum_{i=1}^{k}|\locs|\maxPriceTrans= k |\locs|\maxPriceTrans$. But
  since $k\leq |\stratmin|$, we obtain that
  $\sum_{i=1}^{k}|\locs|\maxPriceTrans\leq
  |\stratmin||\locs|\maxPriceTrans$. For the last sum, we observe
  that, by definition of the split of $\run$ into $\run_1$,
  $\run_2,\ldots,\run_k$ (with $k-1$ extra transitions in-between),
  $|\run|=\sum_{i=1}^{k}|\run_i|$, hence
  $\sum_{i=1}^{k}\lfloor(|\run_i|-1)/|\locs|\rfloor\geq\frac{|\run|}{|\locs|} -2|\stratmin|$,  
  since $|\stratmin|\geq k$. Plugging these three bounds in~\eqref{eq:11},
  we obtain:
  \begin{align}
    \sum_{i=1}^{k}\puse{(\run_i)} &\leq  \maxPriceLoc + |\stratmin||\locs|\maxPriceTrans
                                  -\frac{|\run|}{|\locs|} +2|\stratmin|\,.\label{eq:12}
  \end{align}
  Finally, using the bounds~\eqref{eq:4} and~\eqref{eq:12}
  in~\eqref{eq:5}, we obtain:
  \begin{align*}
    \puse{(\run)} &\leq  \maxPriceLoc + (|\stratmin|-1) \maxPriceTrans  + |\stratmin||\locs|\maxPriceTrans
                  -\frac{|\run|}{|\locs|} +2|\stratmin|\,,
  \end{align*}
  which concludes the proof, using the fact that $|\locs|\geq 1$. 
\end{proof}

\subsubsection{Fake-optimal strategies} Next, to characterise the fact
that $\stratmin$ must allow \MinPl to reach a \pname which is
\emph{small enough, without necessarily reaching a target location}, we
define the \emph{fake value} of an NC-strategy $\stratmin$ from a
configuration $s$ as:
\[\fakeValue_\game^{\stratmin}(s) = \sup \{\cost{\run} \mid \rho \in
  \Play{s,\stratmin}, \rho \textrm{ reaches a target}\}\] i.e.~the
value obtained when \emph{ignoring} the $\stratmin$-induced plays that
\emph{do not} reach the target: we let $\sup\emptyset=-\infty$.  Thus,
clearly, $\fakeValue_\game^{\stratmin}(s) \leq
\val^{\stratmin}(s)$. We say that an NC-strategy $\stratmin$ is
\emph{fake-optimal} if
$\fakeValue_\game^{\stratmin}(s)= \val_\game(s)$ for all
configurations $s$.

Let us now explain why this notion of fake-optimal strategy is
important. As we are about to show, we can combine any fake-optimal
NC-strategy $\stratmin^1$ with an attractor strategy $\stratmin^2$
into a switching strategy $\stratmin$, which forces to eventually
reach the target with a price that we can make as small as desired
(since $\stratmin^1$ is an NC-strategy) when (negative) cycles can be
enforced in the game by $\MinPl$.

\begin{lem}\label{lem:fake-optimality}
  Let $\game$ be an \SPTG such that $\val_\game(s)\neq +\infty$, for
  all $s$.  Let $\stratmin^1$ be an NC-strategy of $\MinPl$ in
  $\game$, and $\stratmin^2$ be an attractor strategy. Then, for
  all
  $n\in\N$, the switching strategy $\stratmin$ described by the pair
  $(\stratmin^1, \stratmin^2)$ and the switching threshold
  \[      K= |\locs|\times\big(2\maxPriceLoc +
      (2|\stratmin^1|)\times|\locs|\maxPriceTrans +3|\stratmin^1| -\max(-n,
      \fakeValue_\game^{\stratmin^1}(s))\big)\] is such
  that
  $\Value_\game^{\stratmin}(s)\leq\max(-n,
  \fakeValue_\game^{\stratmin^1}(s))$  
  for all configurations
  $s$.  
\end{lem}
\begin{rem}
  In particular, in the case where $\stratmin^1$ is a fake-optimal
  NC-strategy, and $\val_\game(s)\neq -\infty$, and when we choose the
  parameter $n$ in the definition of the threshold s.t.
  $n>-\val_\game(s)$, then, we obtain a strategy $\stratmin^1$ that is
  optimal for $\MinPl$ from configuration $s$.
\end{rem}
\begin{proof}[Proof of Lemma~\ref{lem:fake-optimality}]
  In order to establish that
  $\Value_\game^{\stratmin}(s)\leq\max(-n,
  \fakeValue_\game^{\stratmin^1}(s))$ (under the assumptions of the
  lemma), we will consider any play $\rho$ in $\Play{s,\stratmin}$
  and show that
  $\cost{\rho}\leq\max(-n, \fakeValue_\game^{\stratmin^1}(s))$.

  There are two possibilities regarding $\rho\in\Play{s,\stratmin}$,
  depending on whether the switch has happened or
  not: 
  \begin{enumerate}
  \item If $\rho$ reaches the target without switching from
    $\stratmin^1$ to $\stratmin^2$, then
    $\rho\in \Play{s,\stratmin^1}$ and thus
    $\cost{\rho} \leq \fakeValue_\game^{\stratmin^1}(s)\leq \max(-n,
    \fakeValue_\game^{\stratmin^1}(s))$.  
  \item If $\rho$ reaches the target after the switch happened from
    $\stratmin^1$ to $\stratmin^2$, we can decompose $\rho$ into the
    concatenation of a prefix $\rho_1$ of length $K+1$ conforming to
    $\stratmin^1$, and a play prefix~$\rho_2$ conforming to
    $\stratmin^2$.  As $\stratmin^1$ is an NC-strategy, every
    $|\locs|$ steps, either the time valuation of the play changed of
    interval of $\intervals(\stratmin^1)$ or a cycle occured within
    the same interval in which the cost of the discrete transitions is
    at most $-1$. In other words, if $(m + |\stratmin^1|)|L|$ steps
    occured, at least $m$ cycle occured reducing the
    discrete \pname of the play by at least $m$.  As a
    consequence, thanks to Lemma~\ref{lem:bound-price-run}, since
    $\rho_1$ has $K$ steps with
    \[
      K= |\locs|\times\big(\maxPriceLoc +
      (2|\stratmin^1|-1)\times|\locs|\maxPriceTrans +3|\stratmin^1| -[\max(-n,
      \fakeValue_\game^{\stratmin^1}(s))-\maxPriceLoc -
      |\locs|\maxPriceTrans]\big)
    \]
    we know that
    $\puse{(\rho_1)}\leq \max(-n, \fakeValue_\game^{\stratmin^1}(s))
    -\maxPriceLoc - |\locs|\maxPriceTrans$. Moreover,
    $\puse{(\rho_2)}\leq \maxPriceLoc + |\locs|\maxPriceTrans$ since
    $\stratmin^2$ follows an attractor computation and must thus reach
    the target in at most $|\locs|$ transitions (and at most $1$ unit of
    time).  Hence,
    \[\cost{\rho} = \puse{\rho_1}+\puse{\rho_2}
      \leq \max(-n, \fakeValue_\game^{\stratmin^1}(s))\,. \qedhere
    \]
  \end{enumerate}
\end{proof}

This result allows us to identify the conditions we need to check to
make sure than an \SPTG admits optimal strategies that can be
described in a \emph{finite way}. Formally, we say that an \SPTG is
\emph{finitely optimal} if:
\begin{enumerate}
\item\label{item:finitely-optimal1} $\MinPl$ has a fake-optimal
  NC-strategy;
\item\label{item:finitely-optimal2} $\MaxPl$ has an optimal
  FP-strategy; and
\item\label{item:finitely-optimal3} $\Value_\game(\loc)$ is a cost
  function, for all locations~$\loc$.
\end{enumerate} 
The central point in establishing Theorem~\ref{thm:main-result} will thus be to
prove that \textbf{all \SPTG{s} are finitely optimal}
(Theorem~\ref{the:finiteOptimality}), as this guarantees the existence of
well-behaved optimal strategies and value functions. We will also show that
these have a pseudo-polynomial number of cutpoints
(Theorem~\ref{the:ExpCutpoints}), which easily induces that they can be computed
in pseudo-polynomial time.
%
%
%
%
The proof is by induction on the number of non-urgent locations of the
\SPTG.\@ In Section~\ref{sec:urgentSPTG}, we address the base case of
\SPTG{s} with urgent locations only (where no time can elapse). Since
these \SPTG{s} are very close to the \emph{untimed} min-cost
reachability games of~\cite{BGHM16}, we adapt the algorithm in this
work and obtain the \SolveInstant function
(Algorithm~\ref{algo:value-iteration-fixed}). This function can also
compute $\valgame(\loc,1)$ for all $\loc$ and all games $\game$ (even
with non-urgent locations) since time cannot elapse anymore when the
clock has value $1$. Next, using the continuity result of
Theorem~\ref{prop:continuity-of-val}, we can detect locations $\loc$
where $\valgame(\loc,\valuation)\in\{+\infty,-\infty\}$, for all
$\valuation\in[0,1]$, and remove them from the game. Finally, in
Section~\ref{sec:solving-sptg} we handle \SPTG{s} with non-urgent
locations by refining the technique of~\cite{BouLar06,Rut11} (that
work only on \SPTG{s} with non-negative weights).


\section{SPTGs with only urgent locations}
\label{sec:urgentSPTG}

Throughout this section, we consider an $\rightpoint$-\SPTG
$\game= (\LocsMin, \LocsMax, \LocsFin, \LocsUrg, \fgoalvec,
\allowbreak \transitions, \price)$ where all non-final locations are
urgent, i.e.~$\LocsUrg = \LocsMin\cup\LocsMax$. We also fix an initial
clock value $\valuation$. Since all locations in $\game$ are urgent,
no time will elapse, all configurations will have the same clock value
$\valuation$ and no cost will be incurred by staying in the different
locations of the plays. Hence, we can simplify their representation:
in this section, a play
$\rho=(\loc_0,\valuation)\xrightarrow{c_0}
(\loc_1,\valuation)\xrightarrow{c_1}\cdots$ will be represented simply
as $\loc_0\loc_1\cdots$. The price of this play is
$\cost{\rho}=+\infty$ if $\loc_k\not\in\LocsFin$ for all $k\geq 0$;
and
$\cost{\rho}=\sum_{i=0}^{k-1}\price(\loc_i,\loc_{i+1}) +
\fgoal_{\loc_k}(\valuation)$ if $k$ is the least position such that
$\loc_k\in\LocsFin$.

\subsection{Computing the game value for a particular clock value}
\label{sec:particular-valuation} 

We first explain how we can compute the value function of the game for
a \emph{fixed} clock value $\valuation\in [0,\rightpoint]$: more
precisely, we will compute the vector
$(\val(\loc,\valuation))_{\loc\in\Locs}$ of values for all
locations. We will denote by $\val_\valuation(\loc)$ the value
$\val(\loc,\valuation)$, so that $\val_\valuation$ is the vector we
want to compute.  Since no time can elapse, it consists in an
adaptation of the techniques developed in~\cite{BGHM16} to solve
(untimed) \emph{min-cost reachability games}. The main difference
concerns the weights being rational (and not integers) and the presence
of final cost functions.

Following the arguments of~\cite{BGHM16}, we first observe that
locations $\loc$ with values $\val_\valuation(\loc)=+\infty$ and
$\val_\valuation(\loc)=-\infty$ can be pre-computed (using
respectively attractor and mean-payoff techniques) and removed from
the game without changing the other values. Then, because of the
particular structure of the game $\game$ (where a real \pname is paid
only on the target location, all other weights being integers), for all
plays $\rho$, $\cost{\rho}$ is a value from the set
$\Zstar =\Z+\{\fgoal_\loc(\valuation)\mid \loc\in \LocsFin\}$. We
further define $\Zstarinf = \Zstar\cup\{+\infty\}$. Clearly, $\Zstar$
contains at most $|\LocsFin|$ values between two consecutive integers,
i.e.\
\begin{equation}
  \forall i\in \Z\quad |[i,i+1)\cap\Zstar| \leq
  |\LocsFin| \label{eq:at-least-Qfin} 
\end{equation}

Then, we define an operator
$\operator\colon (\Zstarinf)^\Locs \to (\Zstarinf)^\Locs$ mapping
every vector $\vec x = (x_\loc)_{\loc\in\Locs}$ of $(\Zstarinf)^\Locs$
to $\operator(\vec x) = (\operator(\vec x)_\loc)_{\loc\in\Locs}$
defined by
\[\operator(\vec x)_\loc =
\begin{cases}
  \fgoal_\loc(\valuation) &\textrm{if } \loc\in\LocsFin\\
  \displaystyle{\max_{(\loc,\loc')\in \transitions}}
  \big(\price(\loc,\loc')+\vec x_{\loc'}\big)
  &\textrm{if } \loc\in \LocsMax\\
  \displaystyle{\min_{(\loc,\loc')\in\transitions}}
  \big(\price(\loc,\loc')+\vec x_{\loc'}\big) &\textrm{if } \loc\in
  \LocsMin\,.
\end{cases}\] We will obtain $\val_{\valuation}$ as the limit of the
sequence $(\vec x^{(i)})_{i\geq 0}$ defined by
$\vec x^{(0)}_\loc= +\infty$ if $\loc\not\in\LocsFin$, and
$\vec x^{(0)}_\loc=\fgoal_\loc(\valuation)$ if $\loc\in\LocsFin$, and
then $\vec x^{(i)}=\operator(\vec x^{(i-1)})$ for $i\geq 1$.

The intuition behind this sequence is that \emph{$\vec x^{(i)}$ is the
  value of the game (when the clock takes value~$\valuation$) if we
  impose that $\MinPl$ must reach the target within $i$ steps} (and
pays a price of $+\infty$ if it fails to do so). Formally, for a play
$\rho=\loc_0 \loc_1 \cdots$, we let $\costbound{i}(\rho)=\cost{\rho}$
if $\loc_k\in\LocsFin$ for some $k\leq i$, and
$\costbound{i}(\rho)=+\infty$ otherwise. We further let
\[\bupval{i}_\valuation(\loc)=\inf_{\minstrategy} \sup_{\maxstrategy}
\costbound{i}(\outcomes((\loc,\valuation),\maxstrategy,\minstrategy))\]
\noindent where $\minstrategy$ and $\maxstrategy$ are respectively
strategies of $\MinPl$ and $\MaxPl$. Lemma~6 of~\cite{BGHM16} allows
us to easily obtain that:
\begin{lem}
  For all $i\geq 0$, and $\loc\in\Locs$:
  $\vec x^{(i)}_\loc=\bupval{i}_\valuation(\loc)$.
\end{lem}
\begin{proof}[Sketch of proof]
  This is proved by induction on $i$. It is trivial for $i=0$, and
  playing one more step amounts to computing one more iterate of
  $\operator$.
\end{proof}

Now, let us study how the sequence $(\bupval{i}_\valuation)_{i\geq 0}$
behaves and converges to the finite values of the game. Using again
the same arguments as in~\cite{BGHM16} (in particular, that
$\operator$ is a monotonic and Scott-continuous operator over the
complete lattice $(\Zstarinf)^\Locs$), the sequence
$(\bupval{i}_\valuation)_{i\geq 0}$ converges towards the greatest
fixed point of $\operator$. Let us now show that $\val_\valuation$ is
actually this greatest fixed point. First, Lemma~7 of~\cite{BGHM16} can be adapted to obtain
\begin{lem}\label{lem:after-n-steps-no-infty}
  For all $\loc\in\Locs$:
  $\bupval{|\Locs|-1}_\valuation(\loc)\leq (|\Locs|-1) \maxPriceTrans+\maxPriceFin\,$.
\end{lem}
\begin{proof}
  \newcommand\Attr{\mathsf{Attr}} Denoting by $\Attr_i(S)$ the
  $i$-steps attractor of set $S$ (i.e.~the set of locations where player
  $\MinPl$ can enforce reaching $S$ in at most $i$ steps), 
  and assuming that
  $\Attr_{-1}(S)=\emptyset$ for all~$S$, we can establish by induction
  on $j$ that: for all locations $\loc\in \Locs$ with
  $0\leq k\leq |\Locs|$ such that
  $\loc\in\Attr_k(\LocsFin)\setminus \Attr_{k-1}(\LocsFin)$, and for
  all $0\leq j\leq |\Locs|$:
  \begin{enumerate}
  \item $j<k$ implies $\bupval{j}_\valuation(\loc)=+\infty$ and
  \item\label{item:attr2} $j\geq k$ implies
    $\bupval{j}_\valuation(\loc)\leq j \maxPriceTrans
    +\maxPriceFin$
    and $\bupval{j}_\valuation(\loc)\in\Zstar$.
  \end{enumerate}
  Then, the result is obtained by taking $j=|\Locs|-1$
  in~\ref{item:attr2}.
\end{proof}

The next step is to show that the values that can be computed along
the sequence (still assuming that $\val(\loc,\valuation)$ is finite
for all $\loc$) are taken from a finite set:
\begin{lem}
  For all $i\geq 0$ and for all $\loc\in\Locs$:
  \[\bupval{|\Locs|+i}_\valuation(\loc) \in \possval_{\valuation}=
    [-(|\Locs|-1) \maxPriceTrans-\maxPriceFin, (|\Locs|-1)
    \maxPriceTrans+\maxPriceFin] \cap \Zstar\] where
  $\possval_{\valuation}$ has cardinality bounded by
  $|\LocsFin|\times \big( 2(|\Locs|-1)\maxPriceTrans + 2\maxPriceFin +
  1\big)$.
\end{lem}
\begin{proof}
  Following the proof of~\cite[Lemma~9]{BGHM16}, it is easy to show
  that if $\MinPl$ can secure, from some vertex $\loc$, a price less
  than $-(|\Locs|-1) \maxPriceTrans-\maxPriceFin$,
  i.e.~$\val(\loc,\valuation)<-(|\Locs|-1)
  \maxPriceTrans-\maxPriceFin$, then it can secure an arbitrarily
  small price from that configuration,
  i.e.~$\val(\loc,\valuation)=-\infty$, which contradicts our
  hypothesis that the value is finite.


  Hence, for all $i\geq 0$, for all $\loc$:
  $\bupval{i}_\valuation(\loc)\geq \val(\loc,\valuation)> -(|\Locs|-1)
  \maxPriceTrans-\maxPriceFin$.
  By Lemma~\ref{lem:after-n-steps-no-infty} and since the sequence is
  non-increasing, we conclude that, for all $i\geq 0$ and for all
  $\loc\in\Locs$:
  \[-(| \Locs|-1) \maxPriceTrans-\maxPriceFin <
    \bupval{|\Locs|+i}_\valuation(\loc)\leq (|\Locs|-1)
    \maxPriceTrans+\maxPriceFin\,.\] Since all
  $\bupval{|\Locs|+i}_\valuation(\loc)$ are also in $\Zstar$, we
  conclude that
  $\bupval{|\Locs|+i}_\valuation(\loc)\in\possval_{\valuation}$ for
  all $i\geq 0$. The upper bound on the size of
  $\possval_{\valuation}$ is established by equation~\eqref{eq:at-least-Qfin}.
\end{proof}

This allows us to bound the number of iterations needed for the
sequence to stabilise. 
Indeed, at each step after the first $|\Locs|$ steps, the value of at
least one location must decrease, while remaining in a set of
values that contains
$2(|\Locs|-1)\maxPriceTrans + 2\maxPriceFin + 1$
elements.  
\begin{cor}
  The sequence $(\bupval{i}_\valuation)_{i\geq 0}$ stabilises after a number of
  steps at most
  $|\LocsFin|\times |\Locs|\times \big( 2(|\Locs|-1)\maxPriceTrans +
  2\maxPriceFin + 1\big)+ |\Locs|$.
\end{cor}

\begin{algorithm}[tbp]
  \caption{\texttt{solveInstant}($\game$,$\valuation$)}
  \label{algo:value-iteration-fixed}
  \DontPrintSemicolon%
  \KwIn{$\rightpoint$-\SPTG
    $\game= (\LocsMin, \LocsMax, \LocsFin, \LocsUrg, \fgoalvec, \transitions,
    \price)$, a clock value $\valuation\in [0,\rightpoint]$}%
  \SetKw{value}{\ensuremath{\mathsf{X}}}
  \SetKw{prevvalue}{\ensuremath{\mathsf{X}_{pre}}}
  
  \BlankLine

  \ForEach{$\loc\in\Locs$}{%
    \leIf{$\loc\in \LocsFin$}%
    {$\value(\loc) :=
      \fgoal_\loc(\valuation)$}{$\value(\loc):=+\infty$} %
  }%

  \Repeat{$\value = \prevvalue$}{%
    $\prevvalue := \value$\;%
    \lForEach{$\loc\in\LocsMax$}{$\value(\loc) :=
      \max_{(\loc,\loc')\in\transitions}
      \big(\price(\loc,\loc')+\prevvalue(\loc')\big)$} %
    \lForEach{$\loc\in
      \LocsMin$}{$\value(\loc)
      := \min_{(\loc,\loc')\in\transitions}
      \big(\price(\loc,\loc')+\prevvalue(\loc')\big)$}%
    \lForEach{$\loc\in \Locs$ \emph{such that}
      $\value(\loc) < -(|\Locs|-1)
      \maxPriceTrans-\maxPriceFin$\label{line-infty-RT}\label{line-infty}}%
    {$\value(\loc) := -\infty$\label{line-update}}%
  } %
  \Return{$\value$}
\end{algorithm}

Next, the proofs of~\cite[Lemma~10 and
Corollary~11]{BGHM16} allow us to conclude that
this sequence converges towards the value $\val_{\valuation}$ of the
game (when all values are finite), which proves that the value
iteration scheme of \algorithmcfname~\ref{algo:value-iteration-fixed}
computes exactly $\val_{\valuation}$ for all
$\valuation\in [0,\rightpoint]$.  Indeed, this algorithm also works
when some values are not finite. As a corollary, we obtain a
characterisation of the possible values of $\game$:
\begin{cor}\label{cor:possible-values}
  For all $\rightpoint$-\SPTG{s} $\game$ with only urgent locations,
  locations $\loc\in\Locs$ and values $\valuation\in [0,\rightpoint]$,
  $\val(\loc,\valuation)$ is contained in the set
  $\possval_{\valuation}\cup\{-\infty,+\infty\}$ of cardinal
  polynomial in $|\Locs|$, $\maxPriceTrans$, and $\maxPriceFin$,
  i.e.~pseudo-polynomial with respect to the size of~$\game$.
\end{cor}

Finally, Section~3.4 of~\cite{BGHM16} explains how to compute
simultaneously optimal strategies for both players. In our context,
this allows us to obtain for every clock value
$\valuation\in[0,\rightpoint]$ and location $\loc$ of an
$\rightpoint$-\SPTG, such that
$\val(\loc,\valuation)\notin \{-\infty,+\infty\}$, an optimal
FP-strategy for $\MaxPl$, and an optimal switching strategy for
$\MinPl$.
In case of a configuration of value $-\infty$, the switching
strategies built in Lemma~\ref{lem:fake-optimality}, for all
parameters $n>0$, give a sequence of strategies of $\MinPl$ that
ensure a value as low as possible. 


\subsection{Study of the complete value functions: $\game$ is finitely
  optimal}
So far, we have been able to compute $\val_\game(\loc,\valuation)$ for
a fixed value $\valuation$. In practice, this can be achieved by
calling \SolveInstant
(\algorithmcfname~\ref{algo:value-iteration-fixed}). Clearly, running this
algorithm for all possible valuations $\valuation$ is not feasible, so
let us now explain how we can reduce the computation of
$\val_\game(\loc)\colon \valuation\in[0,\rightpoint] \mapsto
\val(\loc,\valuation)$ (for all $\loc$) to a \emph{finite number of
  calls} to \SolveInstant. We first study a precise characterisation
of these functions, in particular showing that these are cost
functions of $\CF{\{[0,\rightpoint]\}}$.

We first define the set $\F_{\game}$ of affine functions over
$[0,\rightpoint]$ as follows:\label{page:FG}
\[\F_{\game} = \{k+\fgoal_\loc\mid \loc\in\LocsFin \land 
k\in[-(|\Locs|-1)\maxPriceTrans,(|\Locs|-1)\maxPriceTrans]\cap\Z\}\]
Observe that this set is finite and that its cardinality is bounded above by
$2|\Locs|^2\maxPriceTrans$, pseudo-polynomial in the size of
$\game$. Moreover, as a direct consequence of
Corollary~\ref{cor:possible-values}, this set contains enough
information to compute the value of the game in each possible
value of the clock, in the following sense:
\begin{lem}\label{lem:for-all-val-there-is-f-in-F}
  For all $\loc\in \Locs$, for all $\valuation\in [0,\rightpoint]$: if
  $\val(\loc,\valuation)$ is finite, then there is $f\in\F_{\game}$
  such that $\val(\loc,\valuation)=f(\valuation)$.
\end{lem}

\begin{figure}[tbp]
  \centering
  \newdimen\Xabs
  \newdimen\Xord
  \begin{tikzpicture}[>=latex]
    \node[below] at (7,0) {$\valuation$};
    \draw[->,very thin] (-0.3,0) -- (7,0);
    \draw[->,very thin] (0,-1.7) -- (0,5);
    \node[below left] at (0,0) {0};
    \node[below right] at (6,0) {$\rightpoint$};
    \draw[dotted] (6,-1.7) -- (6,4.5);

    \draw[very thin,color=purple,draw,name path = fu] (0,3) -- (6,4);
    \path[very thick,color=purple,draw,name path = f] (0,2) -- (6,3);
    \draw[very thin,color=purple,draw,name path = fd] (0,1) -- (6,2);

    \draw[very thin,color=cyan,draw,name path = gu] (0,4.5) -- (6,0.5);
    \draw[very thick,color=cyan,draw,name path = g] (0,3.5) -- (6,-0.5);
    \draw[very thin,color=cyan,draw,name path = gd] (0,2.5) -- (6,-1.5);
    
    \draw[very thin,color=olive,draw,name path = hu] (0,2.5) -- (6,2.5);
    \draw[very thick,color=olive,draw,name path = h] (0,1.5) -- (6,1.5);
    \draw[very thin,color=olive,draw,name path = hd] (0,0.5) -- (6,0.5);

    \foreach \fA/\fB in {gd/hu, g/fu, g/hu, fu/gu, f/gu, gu/fd, h/gu, gu/hd} {
      \path[name intersections={of={\fA} and {\fB},by={X}}] (X);
      \pgfgetlastxy{\Xabs}{\Xord};
      \node[fill,circle,inner sep=0pt,minimum size=1.5mm] (Y) at
      (\Xabs,0) {};
      \draw[dotted,thick] (X) -- (Y);
    }

  \end{tikzpicture}
  \caption{Network of affine functions defined by $\F_\game$:
    functions in bold are final affine functions of $\game$, whereas
    non-bold ones are their translations with weights
    $k\in[-(|\Locs|-1)\maxPriceTrans,(|\Locs|-1)\maxPriceTrans]\cap\Z=\{-1,0,1\}$.
    $\posscp_\game$ is the set of absciss\ae\ of intersections points,
    represented by black disks.
    }
  \label{fig:network-posscp}
\end{figure}

Using the continuity of $\val_\game$
(Theorem~\ref{prop:continuity-of-val}), this shows that all the
cutpoints of $\val_\game$ are intersections of functions
from~$\F_{\game}$, i.e.~belong to the set of \emph{possible cutpoints}
\[\posscp_\game =\{\valuation\in [0,\rightpoint]\mid \exists
  f_1,f_2\in\F_{\game}: f_1\neq f_2\land
  f_1(\valuation)=f_2(\valuation)\}\,.\] This set is depicted in
Figure~\ref{fig:network-posscp} on an example.  Observe that
$\posscp_\game$ contains at most
$|\F_{\game}|^2=4|\Locs|^4(\maxPriceTrans)^2$ points (also
pseudo-polynomial in the size of~$\game$) since all functions
in~$\F_{\game}$ are affine, and can thus intersect at most once with
every other function.  Moreover, $\posscp_\game\subseteq \Q$, since
all functions of $\F_{\game}$ take rational values in $0$ and
$\rightpoint\in\Q$. Thus, for all $\loc$, $\val_\game(\loc)$ is a cost
function (with cutpoints in $\posscp_\game$ and pieces from
$\F_\game$).  Since $\val_\game(\loc)$ is a piecewise affine function,
we can characterise it completely by computing only its value on its
cutpoints. Hence, we can reconstruct $\val_\game(\loc)$ by calling
\SolveInstant on each rational clock value
$\valuation \in \posscp_\game$.  From the optimal strategies computed
along \SolveInstant, we can also reconstruct a fake-optimal
NC-strategy for $\MinPl$ and an optimal FP-strategy for $\MaxPl$,
hence:
\begin{prop}\label{prop:baseCase}
  Every $\rightpoint$-\SPTG $\game$ with only urgent locations is finitely
  optimal. Moreover, for all locations~$\loc$, the piecewise affine
  function $\val_\game(\loc)$ has cutpoints in $\posscp_{\game}$ of
  cardinality $4|\Locs|^4(\maxPriceTrans)^2$, pseudo-polynomial in
  the size of $\game$.
\end{prop}

Let us establish this proposition. Notice, that it allows us to
compute $\val(\loc)$ for every $\loc\in\Locs$. First, we compute the
set $\posscp_{\game}=\{y_1,y_2,\ldots,y_k\}$, which can be done in
pseudo-polynomial time in the size of $\game$. Then, for all
$1\leq i\leq k$, we can compute the vectors
$\big(\val(\loc,y_i)\big)_{\loc\in\Locs}$ of values in each location
when the clock takes value $y_i$ using
\algorithmcfname~\ref{algo:value-iteration-fixed}. This provides the
value of $\val(\loc)$ in each cutpoint, for all locations $\loc$,
which is sufficient to characterise the whole value function, as it is
continuous and piecewise affine. Observe that all cutpoints, and
values at the cutpoints, in the value function are rational numbers,
so \algorithmcfname~\ref{algo:value-iteration-fixed} is effective.
Thanks to the above discussions, this procedure consists in a
pseudo-polynomial number of calls to a pseudo-polynomial algorithm,
hence, it runs in pseudo-polynomial time. This allows us to conclude
that $\val_\game(\loc)$ is a cost function for all $\loc$. This proves
item~\ref{item:finitely-optimal3} of the definition of finite
optimality for $\rightpoint$-\SPTG{s} with only urgent locations.

Let us conclude the proof that $\rightpoint$-\SPTG{s} with only urgent
locations are finitely optimal by showing that $\MinPl$ has a
fake-optimal NC-strategy, and $\MaxPl$ has an optimal FP-strategy. Let
$\valuation_1, \valuation_2,\ldots,\valuation_k$ be the sequence of
elements from $\posscp_\game$ in increasing order, and let us assume
$\valuation_0=0$. For all $1\leq i\leq k$, let $f_i^\loc$ be the
function from $\F_\game$ that defines the piece of $\val_\game(\loc)$
in the interval $[\valuation_{i-1},\valuation_i]$ (we have shown above
that such an $f_i^\loc$ always exists). Formally, for all
$1\leq i\leq k$, $f_i^\loc\in\F_\game$ verifies
$\val(\loc,\valuation)=f_i^\loc(\valuation)$, for all
$\valuation\in
[\valuation^\loc_{i-1},\valuation^\loc_i]$. 
Next, for all $1\leq i\leq k$, let $\mu_{i}$ be a value taken in the
middle\footnote{Taking the middle is an arbitrary choice, any point
  strictly within the interval would work} of
$[\valuation_{i-1}, \valuation_i]$,
i.e.~$\mu_i=\frac{\valuation_i+\valuation_{i-1}}{2}$. Note that all
$\mu_i$'s are rational values since all $\valuation_i$'s are. By
applying \SolveInstant in each $\mu_i$, we can compute
$(\val_\game(\loc,\mu_i))_{\loc\in\Locs}$, and we can extract an
optimal memoryless strategy $\stratmax^i$ for $\MaxPl$ and an optimal
switching strategy $\stratmin^i$ for $\MinPl$. Thus we know that, for
all $\loc\in\Locs$, playing $\stratmin^i$ (respectively,
$\stratmax^i$) from $(\loc, \mu_i)$ allows $\MinPl$ (respectively,
$\MaxPl$) to ensure a price at most (respectively, at least)
$\val_\game(\loc,\mu_i)=f_i^\loc(\mu_i)$. However, it is easy to check
that the bound given by $f_i^\loc(\mu_i)$ holds in every clock value,
i.e.~for all $\loc$, for all
$\valuation\in [\valuation_{i-1}, \valuation_i]$
\[\cost{(\loc,\valuation), \stratmin^i}\leq f_i^\loc(\valuation)
  \qquad \text{ and } \qquad \cost{(\loc,\valuation), \stratmax^i}\geq
  f_i^\loc(\valuation)\,.\] This holds because:
\begin{enumerate}
\item $\MinPl$ can play $\stratmin^i$ from all clock values (in
  $[0,r]$) since we are considering an $\rightpoint$-\SPTG;\@ and
\item $\MaxPl$ does not have more possible strategies from an
  arbitrary clock value $\valuation\in [0,r]$ than from~$\mu_i$, because
  all locations are urgent and time cannot elapse (neither from
  $\valuation$, nor from $\mu_i$).
\end{enumerate}
And symmetrically for $\MaxPl$.

We conclude that $\MinPl$ can consistently play the same strategy
$\stratmin^i$ from all configurations~$(\loc,\valuation)$ with
$\valuation\in [\valuation_{i-1},\valuation_i]$ and secure a price
which is at most $f_i^\loc(\valuation)=\val_\game(\loc,\valuation)$,
i.e.~$\stratmin^i$ is optimal on this interval. By definition of
$\stratmin^i$, it is easy to extract from it a fake-optimal
NC-strategy (actually, $\stratmin^i$ is a switching strategy described
by a pair $(\stratmin^1,\stratmin^2)$, and~$\stratmin^1$ can be used
to obtain the fake-optimal NC-strategy). The same reasoning applies to
strategies of $\MaxPl$ and we conclude that $\MaxPl$ has an optimal
FP-strategy.

\section{Finite optimality of general SPTGs}
\label{sec:solving-sptg}

In this section, we consider \SPTG{s} with non-urgent locations. We
first prove that all such \SPTG{s} are finitely optimal. Then, we
introduce Algorithm~\ref{alg:solve} to compute optimal values and
strategies of \SPTG{s}. Throughout the section, we fix an \SPTG
$\game = (\LocsMin, \LocsMax, \LocsFin, \LocsUrg, \fgoalvec,
\transitions, \price)$ with non-urgent locations.  Before presenting
our core contributions, let us explain how we can detect locations
with infinite values. As already argued, we can compute $\val(\loc,1)$
for all $\loc$ assuming all locations are urgent, since time cannot
elapse anymore when the clock has value $1$. This can be done with
\SolveInstant (\algorithmcfname~\ref{algo:value-iteration-fixed}).
Then, from the absence of guards in \SPTG{s} and
Theorem~\ref{prop:continuity-of-val} we have that $\val(\loc,1)=+\infty$
(respectively, $\val(\loc,1)=-\infty$) if and only if
$\val(\loc,\valuation)=+\infty$ (respectively,
$\val(\loc,\valuation)=-\infty$) for all $\valuation \in[0,1]$.  We
can thus remove from the game all locations with infinite value, and
this does not affect the values of other locations. Thus, we
henceforth assume that $\val(\loc,\valuation)\in\R$ for all
$(\loc,\valuation)\in\confgame$.

\subsection{The $\game_{\Locs',r}$ construction} To prove finite
optimality of \SPTG{s} and to establish correctness of our algorithm,
we rely in both cases on a construction that consists in
decomposing~$\game$ into a sequence of \SPTG{s} with \emph{fewer
  non-urgent locations}. Intuitively, a game with fewer non-urgent
locations is easier to solve since it is closer to an untimed game (in
particular, when all locations are urgent, we can apply the techniques
of Section~\ref{sec:urgentSPTG}). More precisely, given a set~$\Locs'$
of non-urgent locations,
we will define a
(possibly infinite) sequence of clock values $1=r_0> r_1>\cdots$ and a
sequence $\game_{\Locs',r_0}$, $\game_{\Locs',r_1},\ldots$ of \SPTG{s}
such that
\begin{enumerate}
\item the non-final locations of $\game_{\Locs', r_i}$ are exactly the ones of
  $\game$, except that the locations of $\Locs'$ are now
  urgent (some final locations are added to allow players to wait until $r_i$); and
\item for all $i\geq 0$, the value function of $\game_{\Locs',r_i}$ is
  equal to $\val_\game$ on the interval $[r_{i+1},r_i]$. Hence, we can
  re-construct $\val_\game$ by assembling well-chosen parts of the
  value functions of the games~$\game_{\Locs',r_i}$ (assuming
  $\inf_i r_i=0$).
\end{enumerate}
In fact, we will show later (see Lemma~\ref{lem:strictlySmaller}) that
we can assume the sequence $r_0,\dots$ to be finite.  This basic
result will be exploited in two directions. First, we prove by
induction on the number of non-urgent locations that all \SPTG{s} are
finitely optimal, by re-constructing $\val_\game$ (as well as optimal
strategies) as a $\opcf$-concatenation of the value functions of a
finite sequence of \SPTG{s} with one non-urgent locations less. The
base case, with only urgent locations, is solved by
Proposition~\ref{prop:baseCase}. This construction suggests a
\emph{recursive} algorithm in the spirit of~\cite{BouLar06,Rut11} (for
non-negative weights). Second, we show that this recursion can be
\emph{avoided} (see Algorithm~\ref{alg:solve}). Instead of turning
locations urgent one at a time, this algorithm makes them all urgent
and computes directly the sequence of \SPTG{s} with only urgent
locations. Its proof of correctness relies on the finite optimality of
\SPTG{s} and, again, on our basic result linking the value functions
of $\game$ and games $\game_{\Locs',r_i}$.

Let us formalise these constructions. Let $\game$ be an \SPTG,
$r\in[0,1]$ be an endpoint, and $\vec x = (x_\loc)_{\loc\in\locs}$ be
a vector of rational values. Then, $\Waiting(\game,r,\vec x)$ is an
$r$-\SPTG in which both players may now decide, in all non-urgent
locations $\loc$, to \emph{wait} until the clock takes value~$r$, and
then to stop the game, adding the weight $x_\loc$ to the current \pname
of the play. Formally,
$\Waiting(\game,r,\vec x) = (\LocsMin,\LocsMax,\LocsFin',\LocsUrg,
\fgoalvec', T', \price')$ is such that
\begin{itemize}
\item $\LocsFin' = \LocsFin \uplus \{\loc^f \mid \loc\in \locs\setminus
  (\LocsUrg\cup \LocsFin)\}$;
\item for all $\loc' \in \LocsFin$ and $\valuation\in[0,r]$,
  $\fgoal'_{\loc'}(\valuation) =\fgoal_{\loc'}(\valuation)$, for all
  $\loc\in\locs\setminus (\LocsUrg\cup \LocsFin)$,
  $\fgoal'_{\loc^f}(\valuation)=(r-\valuation)\cdot
  \price(\loc)+x_\loc$;
\item $T'=T\cup \{(\loc, [0,r], \bot,\loc^f)\mid \loc\in \locs\setminus
  (\LocsUrg\cup \LocsFin)\}$;
\item for all $\transition\in T'$,
  $\price'(\transition)=\price(\transition)$ if $\transition\in T$,
  and $\price'(\transition) =0$ otherwise.
\end{itemize}
Then, we let
$\game_r=\Waiting\big(\game,r,(\Value_\game(\loc,r))_{\loc\in\Locs}\big)$,
i.e.~the game obtained thanks to $\Waiting$ by letting $\vec x$ be the
value of $\game$ in $r$. It is easy to check that this first
transformation does not alter the value of the game, for clock values
before $r$:

\begin{lem}\label{lem:waiting} 
  For all $\valuation\in [0,r]$ and locations $\loc$,
  $\Value_{\game}(\loc,\valuation) =
  \Value_{\game_r}(\loc,\valuation)$.
\end{lem}

Next, we make locations urgent. For a set
$\Locs'\subseteq \Locs\setminus(\LocsUrg\cup \LocsFin)$ of non-urgent locations, we
let~$\game_{\Locs',r}$ be the \SPTG obtained from $\game_r$ by making
urgent every location $\loc$ of $\Locs'$. Observe that, although all
locations $\loc\in \Locs'$ are now urgent in $\game_{\Locs',r}$, their
clones $\loc^f$ allow the players to wait until $r$. When $\Locs'$ is
a singleton $\{\loc\}$, we write $\game_{\loc,r}$ instead of
$\game_{\{\loc\},r}$.
%
%
%
%
%
%
\begin{figure}[tbp]
  \centering
  \begin{tikzpicture}
    \node[above] at (0,5) {$\Value_{\game_{\loc,r}}(\loc,\valuation)$};
    \node[below] at (7,0) {$\valuation$};
    \draw[->] (-0.3,0) -- (7,0);
    \draw[->] (0,-0.3) -- (0,5);
    \draw[dotted] (6,0) -- (6,5);
    \draw[dotted] (3,0) -- (3,5);

    \node[below] at (3,0) {$a$};
    \node[below] at (6,0) {$r$};
    
    \path[name path=courbe,draw] (0,0.5)  .. controls (0.5,7) and ( 1.5,0) .. (6,4); 
    \path[name path = x] (3.5,0) -- (3.5,7);
    \path[name path = xprime] (5.5,0) -- (5.5,7);
    
    \draw[thick, name intersections={of=courbe and x,by={c}},name intersections={of=courbe and xprime,by={d}}]  (c)--(d);
    
    \node at (c) {$\bullet$};
    \node at (d) {$\bullet$};
    
    \coordinate (p) at ($(d)-(2,-0.5)$); 
    
    \draw[dashed] (c) -- ($(c)+(2,-0.5)$);
    
    \draw[dotted] (3.5,0) -- (c);
    \draw[dotted] (5.5,0) -- (d);
    \node[below] at (3.5,0) {$\valuation_1$};
    \node[below] at (5.5,0) {$\valuation_2$};
    
    \coordinate(cy) at ($(c)-(3.5,0)$);
    \coordinate(dy) at ($(d)-(5.5,0)$);

    \draw[dotted] (cy) -- (c);
    \draw[dotted] (dy) -- (d);
    
    \node[left] at (cy) {$\Value_{\game_{\loc,r}}(\loc,\valuation_1)$};
    \node[left] at (dy) {$\Value_{\game_{\loc,r}}(\loc,\valuation_2)$};
    
  \end{tikzpicture}
  \caption{The condition~\eqref{eq:SlopeBigger} (in the case
    $\Locs'=\emptyset$ and $\loc\in \LocsMin$): graphically, it means
    that the slope between every two points of the plot in $[a,r]$
    (represented with a thick line) is greater than or equal to
    $-\price(\loc)$ (represented with dashed line). The value function
    is depicted here as a non-piecewise-affine function, as this is
    not the crucial property we want to highlight.}
  \label{fig:slopeBigger}
\end{figure}
  
While the construction of $\game_r$ does not change the value of the
game, turning locations urgent \emph{does}. Yet, we can characterise
an interval $[a,r]$ on which the value functions
of~$\hame=\game_{L',r}$ and $\hame^+=\game_{L'\cup\{\ell\},r}$
coincide, as stated by the next proposition. The interval~$[a,r]$
depends on the \emph{slopes} of the pieces of $\val_{\hame^+}$ as
depicted in Figure~\ref{fig:slopeBigger}: for each location $\ell$ of
$\MinPl$, the slopes of the pieces of $\val_{\hame^+}$ contained in
$[a,r]$ should be $\geq-\price(\loc)$ (and
$\leq-\price(\loc)$ when $\loc$ belongs to $\MaxPl$). It is proved by
lifting optimal strategies of $\hame^+$ into $\hame$, and strongly
relies on the determinacy result of
Theorem~\ref{thm:determined}. Hereafter, we denote the slope of
$\Value_\game(\loc)$ in-between $\valuation$ and $\valuation'$ by
$\slope^\ell_{\game}(\valuation,\valuation')$, formally defined by
$\slope^\ell_{\game}(\valuation,\valuation') =
\frac{\Value_\game(\loc,\valuation') -
  \Value_\game(\loc,\valuation)}{\valuation'-\valuation}$.

\begin{prop}\label{lem:SameValue}
  Let $0\leq a < r \leq 1$, $\Locs'\subseteq \Locs\setminus(\LocsUrg\cup \LocsFin)$
  and $\loc\notin \Locs'\cup\LocsUrg$ a non-urgent location of \MinPl
  (respectively, \MaxPl). Assume that $\game_{\Locs'\cup\{\loc\},r}$
  is finitely optimal, and that, for all
  $a\leq \valuation_1<\valuation_2 \leq r$:
  \begin{equation}
    \slope^\ell_{\game_{\Locs'\cup\{\loc\},r}}(\valuation_1,\valuation_2)\geq
    -\price(\loc) \quad (\textrm{respectively, }\leq
    -\price(\loc))\,.\label{eq:SlopeBigger} 
  \end{equation} Then,
  for all $\valuation\in [a,r]$ and $\loc'\in \locs$,
  $\Value_{\game_{\Locs'\cup\{\loc\},r}}(\loc',\valuation) =
  \Value_{\game_{\Locs',r}}(\loc',\valuation)$.  Furthermore,
  fake-optimal NC-strategies and optimal FP-strategies in
  $\game_{\Locs'\cup\{\loc\},r}$ are also fake-optimal and optimal
  over $[a,r]$ in $\game_{\Locs',r}$.
\end{prop}

Before proving this result, we start with an auxiliary lemma showing a
property of the rates of change of the value functions associated to
non-urgent locations

\begin{lem}\label{lem:rate}
  Let $\game$ be an $r$-\SPTG, $\loc$ and $\loc'$ be non-urgent
  locations of $\MinPl$ and $\MaxPl$, respectively. Then for all
  $0\leq \valuation<\valuation'\leq r$:
  \[\slope^\ell_{\game}(\valuation,\valuation') \geq
  -\price(\loc)\qquad \textrm{ and } \qquad
  \slope^{\ell'}_{\game}(\valuation,\valuation') \leq
  -\price(\loc')\,.\]
\end{lem}

\begin{proof}
  For the location $\loc$, the inequality rewrites in
  \[\Value_\game(\loc,\valuation) \leq
    (\valuation'-\valuation)\price(\loc) +
    \Value_\game(\loc,\valuation')\,.\] Using the upper definition of
  the value (thanks to the determinacy result of
  Theorem~\ref{thm:determined}), it suffices to prove, for all
  $\varepsilon>0$, the existence of a strategy $\minstrategy$ of
  \MinPl such that for all strategies~$\maxstrategy$ of \MaxPl:
    \begin{align}
      \cost{\CPlay{(\loc,\valuation),\minstrategy,\maxstrategy}} &\leq
                                                                  (\valuation'-\valuation)\price(\loc) +
                                                                  \Value_\game(\loc,\valuation') + \varepsilon\,.\label{eq:2}
    \end{align}

    To prove the
    existence of such a $\minstrategy$, we first fix, given
    $\varepsilon$, a strategy $\minstrategy'$ such that for all strategies $\maxstrategy$:
  \[\cost{\CPlay{(\loc,\valuation'),\minstrategy',\maxstrategy}} \leq
    \Value_\game(\loc,\valuation') + \varepsilon\,.\] Such a strategy
  necessarily exists by definition of the value.  Then, $\minstrategy$
  can be obtained as follows. Under $\minstrategy$, \MinPl will all
  always play as indicated by $\minstrategy'$, except in the first
  round. In this first round, the game is still in $\loc$ and \MinPl
  will play like $\minstrategy'$, adding an extra delay of
  $\valuation'-\valuation$ time units (observe \MinPl is allowed to do
  so, since $\loc$ is non-urgent). Clearly, this extra delay in $\ell$
  will incur a cost of $(\valuation'-\valuation)\price(\loc)$, hence,
  we obtain~\eqref{eq:2}.

  A similar reasoning allows us to obtain the result for $\loc'$.
\end{proof}

Now, we show that, even if the locations in $\Locs'$ are turned into
urgent locations, we may still obtain for them a similar result of the
rates of change as the one of Lemma~\ref{lem:rate}:

\begin{lem}\label{lem:bounded}
  For all locations $\loc \in \Locs'\cap\LocsMin$ (respectively,
  $\loc \in \Locs'\cap\LocsMax$), and $\valuation \in [0,r]$,
  $\Value_{\game_{\Locs',r}}(\loc,\valuation) \leq
  (r-\valuation)\price(\loc)+\Value_{\game}(\loc,r)$ 
  (respectively,  $\Value_{\game_{\Locs',r}}(\loc,\valuation)
  \geq (r-\valuation)\price(\loc) + \Value_{\game}(\loc,r)\;)$.
\end{lem}
\begin{proof}
  It suffices to notice that from $(\loc,\valuation)$, $\MinPl$
  (respectively, $\MaxPl$) may choose to go directly in $\loc^f$
  ensuring the value
  $(r-\valuation)\price(\loc)+ \Value_{\game}(\loc,r)$.
\end{proof}

We are now ready to establish Proposition~\ref{lem:SameValue}:

\begin{proof}[Proof of Proposition~\ref{lem:SameValue}]
  Let $\stratmin$ and $\stratmax$ be respectively a fake-optimal
  NC-strategy of $\MinPl$ and an optimal FP-strategy of $\MaxPl$ in
  $\game_{\Locs'\cup\{\loc\},r}$. Notice that both strategies are also
  well-defined finite positional strategies in $\game_{\Locs',r}$.

  First, let us show that $\stratmin$ is indeed an NC-strategy in
  $\game_{\Locs',r}$. Take a finite play
  $(\loc_0,\valuation_0) \xrightarrow{c_0} \cdots
  \xrightarrow{c_{k-1}} (\loc_k,\valuation_k)$, of length $k\geq 2$,
  that conforms with $\stratmin$ in $\game_{\Locs',r}$, and with
  $\loc_0=\loc_k$ and $\valuation_0,\valuation_k$ in the same interval
  $I$ of $\intervals(\stratmin)$. To show that $\stratmin$ is an
  NC-strategy, we need to show that the total \pname of the transitions in this play is at most
  $-1$. As $\stratmin$ is finite positional and $\nu_0$ and 
  $\nu_k$ are in the same interval, the play
  $(\loc_0,\valuation_k) \xrightarrow{c'_0} \cdots
  \xrightarrow{c'_{k-1}} (\loc_k,\valuation_k)$ also conforms with
  $\stratmin$ (with possibly different weights).
%
  Furthermore, as all
  the delays in this new play are $0$ we are sure that this play is also a valid play
  in $\game_{\Locs'\cup\{\loc\},r}$, in which $\stratmin$ is an
  NC-strategy. Therefore,
  $\price(\loc_0,\loc_1) + \cdots + \price (\loc_{k-1},\loc_k)\leq
  -1$, and $\stratmin$ is an NC-strategy in $\game_{\Locs',r}$.

  We now show the result for $\loc\in \LocsMin$. The proof for
  $\loc\in\LocsMax$ is a straightforward adaptation. Notice that every
  play in $\game_{\Locs',r}$ that conforms with $\stratmin$ is also a
  play in $\game_{\Locs'\cup\{\loc\},r}$ that conforms with
  $\stratmin$, as $\stratmin$ is defined in
  $\game_{\Locs'\cup\{\loc\},r}$ and thus plays with no delay in
  location $\loc$. Thus, for all $\valuation\in [a,r]$ and
  $\loc'\in \locs$, by Lemma~\ref{lem:fake-optimality},
  \begin{equation}\label{eq:fake-value-inequ}
    \Value_{\game_{\Locs',r}}(\loc',\valuation) \leq
    \fakeValue_{\game_{\Locs',r}}^{\stratmin}(\loc',\valuation) =
    \fakeValue_{\game_{\Locs'\cup\{\loc\},r}}^{\stratmin}(\loc',\valuation)=
    \Value_{\game_{\Locs'\cup\{\loc\},r}}(\loc',\valuation)\,.
  \end{equation}

  To obtain that
  $\Value_{\game_{\Locs',r}}(\loc',\valuation)=
  \Value_{\game_{\Locs'\cup\{\loc\},r}}(\loc',\valuation)$, it remains
  to show the reverse inequality. To that extent, let $\run$ be a
  finite play in $\game_{\Locs',r}$ that conforms with $\stratmax$,
  starts in a configuration~$(\loc',\valuation)$ with
  $\valuation\in [a,r]$, and ends in a final location. We show by
  induction on the length of $\run$ that
  $\cost{\run}\geq
  \Value_{\game_{\Locs'\cup\{\loc\},r}}(\loc',\valuation)$.  If $\run$
  has size $1$ then $\loc'$ is a final configuration and
  $\cost{\run} =
  \Value_{\game_{\Locs'\cup\{\loc\},r}}(\loc',\valuation) =
  \fgoal'_{\loc'}(\valuation)$.

  Otherwise $\run= (\loc',\valuation) \xrightarrow{c} \run'$ where
  $\run'$ is a play that conforms with $\stratmax$, starting in a
  configuration $(\loc'',\valuation'')$ and ending in a final
  configuration. By induction hypothesis, we have
  $\cost{\run'} \geq
  \Value_{\game_{\Locs'\cup\{\loc\},r}}(\loc'',\valuation'')$.  We now
  distinguish three cases, the two first being immediate:
  \begin{itemize}
  \item \underline{If $\loc'\in \LocsMax$}, then
    $\stratmax (\loc',\valuation)$ leads to the next configuration
    $(\loc'',\valuation'')$, thus
    \begin{align*}
      \Value_{\game_{\Locs'\cup\{\loc\},r}}(\loc',\valuation)
      &=
        \costgame{\game_{\Locs'\cup\{\loc\},r}}{(\loc',\valuation),
        \stratmax}\\
      &= c +
        \costgame{\game_{\Locs'\cup\{\loc\},r}}{(\loc'',\valuation''),
        \stratmax} \\ &\leq c+ \cost{\run'}= \cost{\run}\,.
    \end{align*}

  \item \underline{If $\loc'\in \LocsMin$, and $\loc'\neq \loc$ or
      $\valuation''=\valuation$}, we have that
    $(\loc',\valuation) \xrightarrow{c} (\loc'',\valuation'')$ is a
    valid transition in~$\game'_{\Locs'\cup\{\loc\},r}$. Therefore,
    $\Value_{\game_{\Locs'\cup\{\loc\},r}}(\loc',\valuation) \leq c +
    \Value_{\game_{\Locs'\cup\{\loc\},r}}(\loc'',\valuation'')$, hence
    \[\cost{\run} = c+\cost{\run'} \geq c+
      \Value_{\game_{\Locs'\cup\{\loc\},r}}(\loc'',\valuation'') \geq
      \Value_{\game_{\Locs'\cup\{\loc\},r}}(\loc',\valuation).\]

  \item Finally, \underline{if $\loc'=\loc$ and
      $\valuation'' >\valuation$}, then
    $c = (\valuation''-\valuation)\price(\loc)+\price (\loc,\loc'')$.
    As
    $(\loc,\valuation'') \xrightarrow{\price (\loc,\loc'')}
    (\loc'',\valuation'')$ is a valid transition in
    $\game_{\Locs'\cup\{\loc\},r}$, we have
    $\Value_{\game_{\Locs'\cup\{\loc\},r}}(\loc,\valuation'') \leq
    \price (\loc,\loc'') +
    \Value_{\game_{\Locs'\cup\{\loc\},r}}(\loc'',\valuation'')$.
    Furthermore, since $\valuation''\in [a,r]$, we can use~\eqref{eq:SlopeBigger} to obtain
    \begin{align*}
      \Value_{\game_{\Locs'\cup\{\loc\},r}}(\loc,\valuation) &\leq
      \Value_{\game_{\Locs'\cup\{\loc\},r}}(\loc,\valuation'') +
      (\valuation''-\valuation)\price(\loc)\\ &\leq
      \Value_{\game_{\Locs'\cup\{\loc\},r}}(\loc'',\valuation'') +
      \price(\loc,\loc'') + (\valuation''-\valuation)\price(\loc)\,.
    \end{align*}
    Therefore
    \begin{align*}
      \cost{\run} 
      &=
        (\valuation''-\valuation) \price(\loc)+
        \price(\loc,\loc'')+\cost{\run'} \\
      &\geq  
        (\valuation''-\valuation) \price(\loc)+ \price(\loc,\loc'')+
        \Value_{\game_{\Locs'\cup\{\loc\},r}}(\loc'',\valuation'') 
        \geq \Value_{\game_{\Locs'\cup\{\loc\},r}}(\loc',\valuation)\,.
    \end{align*}
  \end{itemize}
  This concludes the induction.  As a consequence,
  \[\inf_{\stratmin'\in \stratsofmin(\game_{\Locs',r})}
    \costgame{\game_{\Locs',r}}{\CPlay{(\loc',\valuation),\stratmin',\stratmax}}
    \geq \Value_{\game_{\Locs'\cup\{\loc\},r}}(\loc',\valuation)\] for
  all locations $\loc'$ and $\valuation\in [a,r]$, which finally
  proves that
  $\Value_{\game_{\Locs',r}}(\loc',\valuation)\geq
  \Value_{\game_{\Locs'\cup\{\loc\},r}}(\loc',\valuation)$.
  Fake-optimality of $\minstrategy$ over $[a,r]$ in
  $\game_{\Locs'\cup\{\loc\},r}$ is then obtained by~\eqref{eq:fake-value-inequ}.
\end{proof}

Given an \SPTG $\game$ and some \emph{finitely optimal}
$\game_{\Locs',r}$, we now characterise precisely the left endpoint of
the maximal interval ending in $r$ where the value functions of
$\game$ and $\game_{L',r}$ coincide. To this end, we use the operator
$\Next_{\Locs'}\colon (0,1]\to [0,1]$
defined as:
\[\Next^{\game}_{\Locs'}(r)=  \inf \{ r'\leq r \mid \forall \loc\in \Locs \
\forall \valuation \in [r',r] \
\Value_{\game_{\Locs',r}}(\loc,\valuation) =
\Value_{\game}(\loc,\valuation)\}\,.\]
Most of the time, we will forget about the $\game$ exponent in $\Next^{\game}_{\Locs'}(r)$, but we keep it since it will become useful in later proofs. 
By %
continuity of the value (Theorem~\ref{prop:continuity-of-val}), the
infimum in the definition exists and
$\Value_{\game}(\loc,\Next_{\Locs'}(r)) =
\Value_{\game_{\Locs',r}}(\loc,\Next_{\Locs'}(r))$.  Moreover,
$\Value_{\game}(\loc)$ is a cost function on $[\Next_{\Locs'}(r),r]$, since
$\game_{\Locs',r}$ is finitely optimal. 

However, this definition of
$\Next_{\Locs'}(r)$ is semantical. Yet, building on the ideas of
Proposition~\ref{lem:SameValue}, we can effectively compute
$\Next_{\Locs'}(r)$, given $\Value_{\game_{\Locs',r}}$. We claim that
$\Next_{\Locs'}(r)$ is the \emph{minimal clock value} such that for all
locations $\loc\in \Locs'\cap\LocsMin$ (respectively,
$\loc \in \Locs'\cap\LocsMax$), the slopes of the affine sections of
the cost function $\Value_{\game_{\Locs',r}}(\loc)$ on $[\Next_{\Locs'}(r),r]$
are at least (at most) $-\price(\loc)$. Notice that while this condition (that we show formally in Lemma~\ref{lem:operatorNext}) only speaks about locations of $\Locs'$ that are made urgent, the semantical definition of $\Next^{\game}_{\Locs'}(r)$ gives an equality of values for \emph{all} locations of $\Locs$. Via this condition, $\Next_{\Locs'}(r)$ can be
obtained (see Figure~\ref{fig:HowToConstructri}) by inspecting
iteratively, for all $\loc$ of $\MinPl$ (respectively, $\MaxPl$), the
slopes of $\Value_{\game_{\Locs',r}}(\loc)$, for $\loc\in\Locs'$, by decreasing clock values
until we find a piece with a slope greater than $-\price(\loc)$
(respectively, smaller than $-\price(\loc)$). This enumeration of the
slopes is effective as $\Value_{\game_{\Locs',r}}$ has finitely many
pieces, by hypothesis. Moreover, this guarantees that $\Next_{\Locs'}(r)<r$, as
shown in the following lemma.

\begin{lem}\label{lem:operatorNext}
  Let $\game$ be an $\SPTG$, $\Locs'\subseteq \Locs\setminus(\LocsUrg\cup \LocsFin)$,
  and $r\in(0,1]$, such that $\game_{\Locs'',r}$ is finitely optimal
  for all $\Locs''\subseteq \Locs'$.
  Then, $\Next_{\Locs'}(r)$ is the minimal clock value such that for all
  locations~$\loc\in \Locs'\cap\LocsMin$ (respectively,
  $\loc \in \Locs'\cap\LocsMax$), the slopes of the affine sections of
  the cost function $\Value_{\game_{\Locs',r}}(\loc)$ on
  $[\Next_{\Locs'}(r),r]$ are at least (respectively, at most) $-\price(\loc)$.
  Moreover, $\Next_{\Locs'}(r)<r$.
\end{lem}

\begin{proof}
  Since $\Value_{\game_{\Locs',r}}(\loc)=\Value_\game(\loc)$ on
  $[\Next_{\Locs'}(r),r]$, and as $\loc$ is non-urgent in $\game$,
  Lemma~\ref{lem:rate} states that all the slopes of
  $\Value_\game(\loc)$ are at least (respectively, at most)
  $-\price(\loc)$ on $[\Next_{\Locs'}(r),r]$.

  We now show the minimality property by contradiction. Therefore, let
  $r'<\Next_{\Locs'}(r)$ such that all cost functions
  $\Value_{\game_{\Locs',r}}(\loc)$ are affine on $[r',\Next_{\Locs'}(r)]$, and
  assume that for all $\loc'\in \Locs'\cap\LocsMin$ (respectively,
  $\loc' \in \Locs'\cap\LocsMax$), the slopes of
  $\Value_{\game_{\Locs',r}}(\loc')$ on $[r',\Next_{\Locs'}(r)]$ are at least
  (respectively, at most) $-\price(\loc')$. Hence, this property holds
  on $[r',r]$. Then, by applying Proposition~\ref{lem:SameValue}
  $|\Locs'|$ times (here, we use the finite optimality of the games
  $\game_{\Locs'',r}$ with $\Locs''\subseteq \Locs'$), we have that
  for all $\valuation\in [r',r]$
  $\Value_{\game_{r}}(\loc,\valuation) =
  \Value_{\game_{\Locs',r}}(\loc,\valuation)$.  Using
  Lemma~\ref{lem:waiting}, we also know that for
  all~$\valuation\leq r$, and $\loc$,
  $\Value_{\game_{r}}(\loc,\valuation)=\Value_\game(\loc,\valuation)$.
  Thus,
  $\Value_{\game_{\Locs',r}}(\loc,\valuation) =
  \Value_\game(\loc,\valuation)$.  As $r'<\Next_{\Locs'}(r)$, this contradicts
  the definition of $\Next_{\Locs'}(r)$.

  We finally prove that $\Next_{\Locs'}(r)<r$. This is immediate in case
  $\Next_{\Locs'}(r)=0$, since $r>0$. Otherwise, from the result obtained
  previously, we know that there exists $r'<\Next_{\Locs'}(r)$, and
  $\locMin\in \Locs'$ such that $\Value_{\game_{\Locs',r}}(\locMin)$
  is affine on $[r',\Next_{\Locs'}(r)]$ of slope smaller (respectively,
  greater) than $-\price(\locMin)$ if $\locMin\in\LocsMin$
  (respectively, $\locMin\in\LocsMax$), i.e.\
  \[
  \begin{cases}
    \Value_{\game_{\Locs',r}}(\locMin,r') >
    \Value_{\game_{\Locs',r}}(\locMin,\Next_{\Locs'}(r))+(\Next_{\Locs'}(r)-r')
    \price(\locMin) & \text{if } \locMin\in \LocsMin \\
    \Value_{\game_{\Locs',r}}(\locMin,r') <
    \Value_{\game_{\Locs',r}}(\locMin,\Next_{\Locs'}(r))+(\Next_{\Locs'}(r)-r')
    \price(\locMin) & \text{if } \locMin\in \LocsMax \,.
  \end{cases}
  \]
  From Lemma~\ref{lem:bounded}, we also know that
  \[
  \begin{cases}
    \Value_{\game_{\Locs',r}}(\locMin,r') \leq
    \Value_{\game_{\Locs',r}}(\locMin,r)+(r-r') \price(\locMin) &
    \text{ if } \locMin\in \LocsMin \\
    \Value_{\game_{\Locs',r}}(\locMin,r') \geq
    \Value_{\game_{\Locs',r}}(\locMin,r)+(r-r') \price(\locMin) &
    \text{ if } \locMin\in \LocsMax\,.
  \end{cases}\] Both equations combined imply
  \[
  \begin{cases}
    \Value_{\game_{\Locs',r}}(\locMin,r) >
    \Value_{\game_{\Locs',r}}(\locMin,\Next_{\Locs'}(r)) + (\Next_{\Locs'}(r)-r)
    \price(\locMin) & \text{if } \locMin\in \LocsMin \\
    \Value_{\game_{\Locs',r}}(\locMin,r) <
    \Value_{\game_{\Locs',r}}(\locMin,\Next_{\Locs'}(r)) + (\Next_{\Locs'}(r)-r)
    \price(\locMin) & \text{if } \locMin\in \LocsMax
  \end{cases}
  \]
  which is not possible if $\Next_{\Locs'}(r)=r$.
\end{proof}


Thus, one can reconstruct $\val_\game$ on $[\inf_i r_i,r_0]$ from the
value functions of the (potentially infinite) sequence of games
$\game_{\Locs',r_0}$, $\game_{\Locs',r_1},\ldots$ where
$r_{i+1}=\Next_{\Locs'}(r_i)$ for all $i$ such that $r_i>0$, for all
possible choices of non-urgent locations $\Locs'$.  Another
interesting fact, that we formally state and prove in the next lemma,
is that just on the left of such a point $r_i$ the slope of
$\Value_{\game}(\loc)$, for $\loc\in L'$, is $-\price(\loc)$:

\begin{lem}\label{lem:r_2-r_1-r_0}
  Let  $\game$ be an $\SPTG$, $L'\subseteq \Locs\setminus(\LocsUrg\cup
  \LocsFin)$ and $r_0\in(0,1]$ such that $\game_{\Locs',r_0}$ is finitely
  optimal. Suppose that $r_1=\Next_{\Locs'}(r_0)>0$, and let
  $r_2=\Next_{\Locs'}(r_1)$. Then, there exists $r'\in [r_2,r_1)$ and
  $\loc\in \Locs'$ such that
  \begin{enumerate}
  \item\label{item:prop1} $\Value_{\game}(\loc)$ is affine on $[r',r_1]$, of slope
    equal to $-\price(\loc)$, and
  \item\label{item:prop2}
    $\Value_{\game}(\loc,r_1) \neq \Value_{\game}(\loc,r_0) +
    \price(\loc) (r_0-r_1)$.
  \end{enumerate}
  As a consequence, $\Value_\game(\loc)$ has a cutpoint in
  $[r_1,r_0)$.
\end{lem}
\begin{proof}

We denote by $r'$ the smallest clock value (smaller than
$r_1$) such that for all locations~$\loc$, $\Value_{\game}(\loc)$ is
affine over $[r',r_1]$. Then, the proof goes by contradiction: using
Lemma~\ref{lem:operatorNext}, we assume that for all
$\loc \in \Locs'\cap \LocsMin$ (respectively,
$\loc \in \Locs'\cap \LocsMax$)
\begin{itemize}
\item either ($\neg\ref{item:prop1}$) the slope of $\Value_{\game}(\loc)$ on
  $[r',r_1]$ is greater (respectively, smaller) than
  $-\price(\loc)$,
\item or ($\ref{item:prop1}\wedge \neg\ref{item:prop2}$) for all $\valuation\in [r',r_1]$,
  $\Value_{\game}(\loc,\valuation) = \Value_{\game}(\loc,r_0) +
  \price(\loc) (r_0-\valuation)$.
\end{itemize}

Let $\stratmin^0$ and $\stratmax^0$ (respectively, $\stratmin^1$ and
$\stratmax^1$) be a fake-optimal NC-strategy and an optimal
FP-strategy in $\game_{\Locs',r_0}$ (respectively,
$\game_{\Locs',r_1}$). Let
$r'' = \max ( [r',r_1) \cap
(\points(\stratmin^1)\cup\points(\stratmax^1))
)$,
so that strategies $\stratmin^1$ and $\stratmax^1$ have the
\emph{same behaviour} on all clock values of the interval $(r'',r_1)$,
i.e.~either always play urgently the same transition, or wait, in
a non-urgent location, until reaching some clock value greater than or
equal to $r_1$ and then play the same transition (recall that $\points$ represent
the set of endpoints in which an FP-strategy may change its behaviour).

Observe first that for all $\loc \in \Locs'\cap \LocsMin$
(respectively, $\loc \in \Locs'\cap \LocsMax$), if on the interval
$(r'',r_1)$, $\stratmin^1$ (respectively, $\stratmax^1$) goes to
$\loc^f$ then the slope on $[r'',r_1]$ (and thus on~$[r',r_1]$) is
$-\price(\loc)$. This implies that $(\ref{item:prop1})$ holds and as
either $(\neg~\ref{item:prop1})$ or
($\ref{item:prop1}\wedge \neg\ref{item:prop2}$) is true by assumption,
for such a location $\loc$, we know that
$(\ref{item:prop1}\wedge \neg\ref{item:prop2})$ holds (by
letting~$r'=r''$).

For other locations $\loc$ (notice that there necessarily exists a
location that does not satisfy
$(\ref{item:prop1}\wedge \neg\ref{item:prop2})$ otherwise
$\Next_{\Locs'}(r_0)\leq r' <r_1$), we will construct a new pair of
NC- and FP-strategies $\stratmin$ and~$\stratmax$ in
$\game_{\Locs',r_0}$ such that for all locations $\loc$ and clock
values $\valuation \in (r'',r_1)$
\begin{equation}\label{eq:fake-notfake}
  \fakeValue_{\game_{\Locs',r_0}}^{\stratmin}(\loc,\valuation) \leq
  \Value_{\game}(\loc,\valuation)\leq
  \costgame{\game_{\Locs',r_0}}{(\loc,\valuation), \stratmax} \,.
\end{equation}
As a consequence, with Lemma~\ref{lem:fake-optimality} (over game
$\game_{\Locs',r_0}$), one would have that
$\Value_{\game_{\Locs',r_0}}(\loc,\valuation) =
\Value_\game(\loc,\valuation)$,
which will raise a contradiction with the definition of $r_1$ as
$\Next_{\Locs'}(r_0)<r_0$ (by Lemma~\ref{lem:operatorNext}), and
conclude the proof.

We only show the construction for $\stratmin$, as it is very similar
for $\stratmax$. Strategy $\stratmin$ is obtained by combining
strategies $\stratmin^1$ over $[0,r_1]$, and $\stratmin^0$ over
$[r_1,r_0]$: a special care has to be spent in case $\stratmin^1$
performs a jump to a location $\loc^f$, since then, in $\stratmin$,
we rather glue this move with the decision of strategy $\stratmin^0$
in $(\loc,r_1)$. Formally, let $(\loc,\valuation)$ be a
configuration of $\game_{\Locs',r_0}$ with $\loc \in \LocsMin$. We
construct $\stratmin(\loc,\valuation)$ as follows:
\begin{itemize}
\item if $\valuation \geq r_1$,
  $\stratmin(\loc,\valuation) = \stratmin^0(\loc,\valuation)$;
\item if $\valuation < r_1$, $\loc\not\in\Locs'$ and
  $\stratmin^1(\loc,\valuation) = \big(t,(\loc, \loc^f)\big)$ for
  some delay $t$ (such that $\valuation+t\geq r_1$), we let
  $\stratmin(\loc,\valuation) = \big (r_1-\valuation+t',(\loc,
  \loc') \big )$ where $(t',(\loc, \loc'))=\stratmin^0(\loc,r_1)$;
\item otherwise
  $\stratmin(\loc,\valuation) = \stratmin^1(\loc,\valuation)$.
\end{itemize}

For all finite plays $\run$ in $\game_{\Locs',r_0}$ that conform to
$\stratmin$, start in a configuration $(\loc,\valuation)$ such that
$\valuation\in (r'',r_0]$ and
$\loc \notin \{{\loc'}^f\mid \loc'\in \Locs\}$, and end in a final
location, we show by induction that
$\costgame{\game_{\Locs',r_0}}{\run} \leq
\Value_\game(\loc,\valuation)$.
Note that $\run$ either only contains clock values in $[r_1,r_0]$, or
is of the form
$(\loc,\valuation) \xrightarrow{c} (\loc^f,\valuation')$, or is of
the form $(\loc,\valuation)\xrightarrow{c} \run'$ with $\run'$ a play
that satisfies the above restriction.
\begin{itemize}
\item If $\valuation \in [r_1,r_0]$, then $\run$ conforms with
  $\stratmin^0$, thus, as $\stratmin^0$ is fake-optimal,
  $\costgame{\game_{\Locs',r_0}}{\run} \leq
  \Value_{\game_{\Locs',r_0}}(\loc,\valuation) =
  \Value_\game(\loc,\valuation)$
  (the last inequality comes from the definition of
  $r_1=\Next_{\Locs'}(r_0)$). Therefore, in the following cases, we
  assume that $\valuation \in (r'',r_1)$.
\item Consider then the case where $\run$ is of the form
  $(\loc,\valuation) \xrightarrow{c} (\loc^f,\valuation')$.
  \begin{itemize}
  \item if $\loc\in \Locs'\cap\LocsMin$, $\loc$ is urgent in
    $\game_{\Locs',r_0}$, thus $\valuation'=\valuation$. Furthermore,
    since $\run$ conforms with~$\stratmin$, by construction of
    $\stratmin$, the choice of $\stratmin^1$ on $(r'',r_1)$ consists
    in going to $\loc^f$, thus, as observed above,
    $\ref{item:prop1}\wedge \neg\ref{item:prop2}$ holds for
    $\loc$. Therefore,
    \[\Value_{\game}(\loc,\valuation) = \Value_{\game}(\loc,r_0) +
    \price(\loc) (r_0-\valuation) = \fgoal_{\loc_f}(\valuation) =
    \costgame{\game_{\Locs',r_0}}{\run}\,.\]
  \item If $\loc \in \LocsMin\setminus\Locs'$, by construction, it
    must be the case that
    $\stratmin(\loc,\valuation) = \big(r_1-\valuation+t',(\loc,
    \loc^f) \big)$
    where
    $\big(t,(\loc, \loc^f)\big) = \stratmin^1(\loc,\valuation)$ and
    $\big(t',(\loc, \loc^f)\big) = \stratmin^0(\loc,r_1)$.  Thus,
    $\valuation' = r_1 +t'$. In particular, observe that
    \[\costgame{\game_{\Locs',r_0}}{\run} = (r_1-\valuation)
    \price(\loc) + \costgame{\game_{\Locs',r_0}}{\run'}\]
    where
    $\run' = (\loc,r_1)\xrightarrow {c'} (\loc^f,\valuation')$. As
    $\run'$ conforms with $\stratmin^0$ which is fake-optimal in
    $\game_{\Locs',r_0}$, and $r_1=\Next_{\Locs'}(r_0)$,
    \[\costgame{\game_{\Locs',r_0}}{\run'} \leq
      \Value_{\game_{\Locs',r_0}}(\loc,r_1) =
      \Value_{\game}(\loc,r_1)\,.\] Thus
    \[\costgame{\game_{\Locs',r_0}}{\run} \leq (r_1-\valuation)
      \price(\loc) + \Value_{\game}(\loc,r_1) =
      \costgame{\game_{\Locs',r_1}}{\run''}\] where
    $\run''= (\loc,\valuation) \xrightarrow{c''}(\loc^f,\valuation+t)$
    conforms with $\stratmin^1$ which is fake-optimal in
    $\game_{\Locs',r_1}$. Therefore, since $r_1=\Next_{\Locs'}(r_0)$,
    \[\costgame{\game_{\Locs',r_0}}{\run} \leq
      \Value_{\game_{\Locs',r_1}}(\loc,\valuation) =
      \Value_\game(\loc,\valuation)\,.\]
  \item If $\loc\in \LocsMax$ then
    \begin{align*}
      \costgame{\game_{\Locs',r_0}}{\run}
      &=
        (\valuation'-\valuation)\price(\loc) +
        \fgoal_{\loc_f}(\valuation') \\
      &=
        (\valuation'-\valuation)\price(\loc) + (r_0-\valuation')
        \price(\loc) + \Value_\game(\loc,r_0)\\
      &= (r_0-\valuation)
        \price(\loc)+ \Value_\game(\loc,r_0)\,.
    \end{align*} By
    Lemma~\ref{lem:rate}, since $\loc\in \LocsMax\setminus(\LocsUrg\cup \LocsFin)$
    ($\loc$ is not urgent in $\game$ since $\loc^f$ exists),
    $\Value_\game(\loc,r_1) \geq (r_0-r_1) \price(\loc)
    +\Value_\game(\loc,r_0)$.  Furthermore, observe that if we define
    $\run'$ as the play
    $(\loc,\valuation)\xrightarrow{c'} (\loc^f,\valuation)$ in
    $\game_{\Locs',r_1}$, then $\run'$ conforms with $\stratmin^1$ and
    \begin{align*}
      \costgame{\game_{\Locs',r_1}}{\run'} 
      & = (r_1-\valuation) \price(\loc) + \Value_\game(\loc,r_1) \\
      & \geq (r_1-\valuation) \price(\loc) + (r_0-r_1)
        \price(\loc) +\Value_\game(\loc,r_0) \\ 
      & = (r_0-\valuation) \price(\loc) + \Value_\game(\loc,r_0) \\
      & = \costgame{\game_{\Locs',r_0}}{\run}\,.
    \end{align*}
    Thus, as $\stratmin^1$ is fake-optimal in $\game_{\Locs',r_1}$,
    \[\costgame{\game_{\Locs',r_0}}{\run} \leq
      \costgame{\game_{\Locs',r_1}}{\run'} \leq
      \Value_{\game_{\Locs',r_1}}(\loc,\valuation) =
      \Value_{\game}(\loc, \valuation)\,.\]
  \end{itemize}
\item We finally consider the case where
  $\run = (\loc,\valuation)\xrightarrow{c}\run'$ with $\run'$ that
  starts in configuration~$(\loc',\valuation')$ such that
  $\loc'\notin \{{\loc''}^f\mid \loc''\in \Locs\}$. By induction
  hypothesis
  $\costgame{\game_{\Locs',r_0}}{\run'}\leq
  \Value_\game(\loc',\valuation')$.
  \begin{itemize}
  \item If $\valuation' \leq r_1$, let $\run''$ be the play of
    $\game_{\Locs',r_1}$ starting in $(\loc',\valuation')$ that
    conforms with $\stratmin^1$ and $\stratmax^1$. If $\run''$ does
    not reach a final location, since $\stratmin^1$ is an NC-strategy,
    the \pnames of its prefixes tend to $-\infty$. By considering the
    switching strategy of Lemma~\ref{lem:fake-optimality}, we would
    obtain a completed play conforming with $\stratmax^1$ of price smaller than
    $\Value_{\game_{\Locs',r_1}}(\loc',\valuation')$ which would
    contradict the optimality of $\stratmax^1$. Hence, $\run''$
    reaches the target. Moreover, since $\stratmax^1$ is optimal and
    $\stratmin^1$ is fake-optimal, we finally know that
    $\costgame{\game_{\Locs',r_1}}{\run''} =
    \Value_{\game_{\Locs',r_1}}(\loc',\valuation')
    =\Value_{\game}(\loc',\valuation')$ (since
    $\valuation'\in [\Next_{\Locs'}(r_1),r_1]$). Therefore,
    \begin{align*}
      \costgame{\game_{\Locs',r_0}}{\run} 
      & = (\valuation'-\valuation) \price(\loc) +
        \price(\loc,\loc') +\costgame{\game_{\Locs',r_0}}{\run'} \\ 
      & \leq (\valuation'-\valuation) \price(\loc) +
        \price(\loc,\loc') + 
        \Value_\game(\loc',\valuation') \\ 
      & = (\valuation'-\valuation) \price(\loc) + \price(\loc,\loc')
        + \cost{\run''} = \cost{(\loc,\valuation) \xrightarrow{c'} \run''} 
    \end{align*}
    Since the play $(\loc,\valuation) \xrightarrow{c'} \run''$
    conforms with $\stratmin^1$, we finally have
    \[\costgame{\game_{\Locs',r_0}}{\run} \leq
      \cost{(\loc,\valuation) \xrightarrow{c'} \run''} \leq
      \Value_{\game_{\Locs',r_1}}(\loc,\valuation)=
      \Value_\game(\loc,\valuation)\,.\]
  \item If $\valuation'> r_1$ and $\loc\in \LocsMax$, let $\run^1$
    be the play in $\game_{\Locs',r_1}$ defined by
    $\run^1 = (\loc,\valuation) \xrightarrow{c'}
    (\loc^f,\valuation)$
    and $\run^0$ the play in $\game_{\Locs',r_0}$ defined by
    $\run^0= (\loc,r_1) \xrightarrow{c''} \run'$. We have
    \begin{align*}
      \costgame{\game_{\Locs',r_0}}{\run} 
      &=
        (\valuation'-\valuation)\price(\loc) + \price(\loc,\loc') +
        \costgame{\game_{\Locs',r_0}}{\run'} \\
      &=
        \underbrace{\fgoal_{\loc_f}(\valuation)}_{
        \makebox[0pt][c]{\scriptsize{$=\costgame{\game_{\Locs',r_1}}{\run^1}$}}}
        -\Value_\game(\loc,r_1) 
        + 
        \underbrace{(\valuation'-r_1)\price(\loc) + \price(\loc,\loc') +
        \costgame{\game_{\Locs',r_0}}{\run'}}_{=\;\costgame{\game_{\Locs',r_0}}{\run^0}}\,. 
    \end{align*}
    Since $\run^{0}$ conforms with $\stratmin^0$, fake-optimal, and
    reaches a final location, and since $r_1=\Next_{\Locs'}(r_0)$,
    \[\costgame{\game_{\Locs',r_0}}{\run^0}\leq
      \Value_{\game_{\Locs',r_0}}(\loc,r_1) =
      \Value_\game(\loc,r_1)\,.\] We also have that $\run^1$ conforms
    with $\stratmin^1$, so the previous explanations already proved
    that
    $\costgame{\game_{\Locs',r_1}}{\run^1} \leq
    \Value_\game(\loc,\valuation)$.  As a consequence
    $\costgame{\game_{\Locs',r_0}}{\run} \leq
    \Value_\game(\loc,\valuation)$.
  \item If $\valuation'> r_1$ and $\loc\in \LocsMin$, we know that
    $\loc$ is non-urgent, so that $\loc\not\in \Locs'$. Therefore,
    by definition of $\stratmin$,
    $\stratmin(\loc,\valuation) = ( r_1-\valuation+t',(\loc,\loc'))$
    where $\stratmin^1(\loc,\valuation) = (t,(\loc, \loc^f))$ for
    some delay $t$, and
    $\stratmin^0(\loc,r_1) = (t',(\loc, \loc'))$. If we let $\run^1$
    be the play in $\game_{\Locs',r_1}$ defined by
    $\run^1 = (\loc,\valuation) \xrightarrow{c'}(\loc^f,\valuation)$
    and $\run^0$ the play in $\game_{\Locs',r_0}$ defined by
    $\run^0= (\loc,r_1)\xrightarrow{c''} \run'$, as in the previous
    case, we obtain that
    $\costgame{\game_{\Locs',r_0}}{\run} \leq
    \Value_\game(\loc,\valuation)$.
  \end{itemize}
\end{itemize}

As a consequence of this induction, we have shown that for
all~$\loc\in \Locs$, and~$\valuation\in (r'',r_1)$,
$\fakeValue_{\game_{\Locs',r_0}}^{\stratmin}(\loc,\valuation) \leq
\Value_\game(\loc,\valuation)$, which shows one inequality of~\eqref{eq:fake-notfake}, the other being obtained very similarly.
\end{proof}

Next, we will define two different ways of choosing the subset $L'\subseteq
\Locs\setminus(\LocsUrg\cup \LocsFin)$: the former (one at a time) to prove finite optimality of
all \SPTG{s}, the latter (all at once) to bound the number of cutpoints of the value functions
and obtain an efficient algorithm to solve them.



\begin{figure}[tbp]
  \centering
  \begin{tikzpicture}
    \node[above] at (0,5) {$\Value_{\game_{\locMin,r}}(\locMin,\valuation)$};
    \node[below] at (7,0) {$\valuation$};
    \draw[->] (-0.3,0) -- (7,0);
    \draw[->] (0,-0.3) -- (0,5);
    \draw[dotted] (6,0) -- (6,5);
    \draw[dotted] (3,0) -- (3,5);
    \draw[dotted] (0,4) -- (7,4);
    \node[left] at (0,4) {$\Value_\game(\locMin,r)$};
    \node[below] at (3,0) {$\Next_{\locMin}(r)$};
    \node[below] at (6,0) {$r$};
    
    \draw (6,4) -- (5,3) -- (4,2.7) -- (3.7,2) -- (3,1) -- (2,1.7)
    -- (1.5,1) -- (1,1.3) -- (0.5,0.5) -- (0,0.5);    
    \begin{scope}[xshift=5cm, yshift=3cm]
      \draw[dashed] (0,0) -- (0.5,0.125);
    \end{scope}
    \begin{scope}[xshift=4cm, yshift=2.5cm]
      \draw[dashed] (0,0) -- (0.5,0.125);
    \end{scope}
    \begin{scope}[xshift=3.7cm, yshift=2cm]
      \draw[dashed] (0,0) -- (0.5,0.125);
    \end{scope}
    \begin{scope}[xshift=3cm, yshift=1cm]
      \draw[dashed] (0,0) -- (0.5,0.125);
    \end{scope}
    \begin{scope}[xshift=2cm, yshift=1.7cm]
      \draw[dashed] (0,0) -- (0.5,0.125);
    \end{scope}
    \begin{scope}[xshift=1.5cm, yshift=1cm]
      \draw[dashed] (0,0) -- (0.5,0.125);
    \end{scope}
    \begin{scope}[xshift=1cm, yshift=1.3cm]
      \draw[dashed] (0,0) -- (0.5,0.125);
    \end{scope}
    \begin{scope}[xshift=0.5cm, yshift=0.5cm]
      \draw[dashed] (0,0) -- (0.5,0.125);
    \end{scope}
    \begin{scope}[xshift=0cm, yshift=0.5cm]
      \draw[dashed] (0,0) -- (0.5,0.125);
    \end{scope}
  \end{tikzpicture}
  \caption{In this example $\Locs'=\{\locMin\}$
    and $\locMin\in \LocsMin$. $\Next_{\locMin}(r)$ is the
    leftmost point such that all slopes on its right are 
    at least $-\price(\locMin)$ in the graph of
    $\Value_{\game_{\locMin,r}}(\locMin,\valuation)$. Dashed lines
    have slope $-\price(\locMin)$.}
  \label{fig:HowToConstructri}
\end{figure}

\subsection{\SPTG{s} are finitely optimal} To prove finite
optimality of all \SPTG{s} we reason by induction on the number of
non-urgent locations and instantiate the previous results to the case
where $L'=\{\locMin\}$ where $\locMin$ is a non-urgent location of
\emph{minimum weight} (i.e.~for all $\loc\in \Locs\setminus(\LocsFin\cup\LocsUrg)$,
$\price(\locMin)\leq \price(\loc)$). Given $r_0\in [0,1]$, we let
$r_0> r_1> \cdots$ be the decreasing sequence of clock values such that
$r_i=\Next_{\locMin}(r_{i-1})$ for all $i>0$ with $r_{i-1}>0$. As explained before, we
will build $\val_\game$ on $[\inf_i r_i, r_0]$ from the value
functions of games $\game_{\locMin,r_i}$. Assuming finite optimality
of those games, this will prove that $\game$ is finitely optimal
\emph{under the condition} that $r_0> r_1> \cdots$ eventually stops,
i.e.~$r_i=0$ for some $i$. Lemma~\ref{lem:strictlySmaller} will prove
this property. First, we relate the optimal value functions with the
final cost functions. 

\begin{lem}\label{lem:eqGoal}
  Assume that $\game_{\locMin,r}$ is finitely optimal. If
  $\Value_{\game_{\locMin,r}}(\locMin)$ is affine on a non-singleton
  interval $I\subseteq [0,r]$ with a slope greater\footnote{For this
    result, the order does not depend on the owner of the location,
    but on the fact that $\locMin$ has minimal weight amongst locations
    of $\game$.}  than $-\price(\locMin)$, then there exists
  $f\in \F_\game$ (see definition in page~\pageref{page:FG}) such that
  for all $\valuation \in I$,
  $\Value_{\game_{\locMin,r}}(\locMin,\valuation) = f(\valuation)$.
\end{lem}
\begin{proof}
  Let $\stratmin^1$ and $\stratmax$ be some fake-optimal NC-strategy
  and optimal FP-strategy in~$\game_{\locMin,r}$. As $I$ is a
  non-singleton interval, there exists a subinterval $I'\subset I$,
  which is not a singleton and is contained in an interval of
  $\stratmin^1$ and of $\stratmax$. Let $\stratmin$ be the switching
  strategy obtained from $\stratmin^1$ in
  Lemma~\ref{lem:fake-optimality}: notice that both strategies have
  the same intervals.

  Let $\valuation \in I'$. Since
  $\Value_{\game_{\locMin,r}}(\locMin,\valuation)\notin\{+\infty,-\infty\}$,
  the completed play \[\CPlay{(\locMin,\valuation),\stratmin,\stratmax}\]
  necessarily reaches a final location and has price
  $\Value_{\game_{\locMin,r}}(\locMin,\valuation)$. Thus 
  it is a finite completed play
  $(\loc_0,\valuation_0) \xrightarrow{c_0} \cdots
  (\loc_k,\valuation_k)$ where $(\loc_0,\valuation_0)=(\locMin,\valuation) $ 
  and $\loc_k\in \LocsFin'$.
  We also let $\valuation'\in I'$ be a clock value such that
  $\valuation<\valuation'$. We now explain successively why:
  \begin{enumerate}
  \item for all $i$, $\valuation_i=\valuation$;
  \item $\loc_k \in \LocsFin$;
  \item $\CPlay{(\locMin,\valuation),\stratmin,\stratmax}$ contains no cycles.
  \end{enumerate}
  We will then use these properties to conclude. 

  \begin{enumerate}
  \item Assume by contradiction that there exists an index $i$ such
    that $\valuation<\valuation_i$ and let $i$ be the smallest of such
    indices. For each $j< i$, if $\loc_j\in \LocsMin$, let
    $(t,\transition) = \stratmin(\loc_j,\valuation)$ and
    $(t',\transition') = \stratmin(\loc_j,\valuation')$. Similarly, if
    $\loc_j\in \LocsMax$, we let
    $(t,\transition) = \stratmax(\loc_j,\valuation)$ and
    $(t',\transition') = \stratmax(\loc_j,\valuation')$. As $I'$ is
    contained in an interval of $\stratmin$ and $\stratmax$, we have
    $\transition = \transition'$ and either $t = t' = 0$, or
    $\valuation + t= \valuation'+t'$. Applying this result for all
    $j<i$, we obtain that
    $(\loc_0,\valuation') \xrightarrow{c'_0} \cdots
    (\loc_{i-1},\valuation') \xrightarrow{c'_{i-1}}
    (\loc_i,\valuation_i) \xrightarrow{c_i} \cdots
    (\loc_k,\valuation_k)$ is a prefix of
    $\CPlay{(\locMin,\valuation'),\stratmin,\stratmax}$: notice
    moreover that, as before, this prefix has \pname
    $\Value_{\game_{\locMin,r}}(\locMin,\valuation')$.  In particular,
    \[\Value_{\game_{\locMin,r}}(\locMin,\valuation') =
      \Value_{\game_{\locMin,r}}(\locMin,\valuation)
      -(\valuation'-\valuation) \price(\loc_{i-1})\leq
      \Value_{\game_{\locMin,r}}(\locMin,\valuation)
      -(\valuation'-\valuation) \price(\locMin)\] which implies that
    the slope of $\Value_{\game_{\locMin,r}}(\locMin)$ is at most
    $-\price(\locMin)$, and therefore contradicts the hypothesis. As a
    consequence, we have that $\valuation_i = \valuation$ for all $i$.
  
  \item Again by contradiction, assume now that $\loc_k = \loc^f$ for some
    $\loc\in \Locs\setminus(\LocsUrg\cup \LocsFin)$. By the same reasoning as before,
    we then would have
    $\Value_{\game_{\locMin,r}}(\locMin,\valuation') =
    \Value_{\game_{\locMin,r}}(\locMin,\valuation)
    -(\valuation'-\valuation) \price(\loc)$, which again contradicts
    the hypothesis. Therefore, $\loc_k\in \LocsFin$.

  \item Suppose, for a contradiction, that the prefix
    $(\loc_0,\valuation) \xrightarrow{c_0} \cdots (\loc_k,\valuation)$
    contains a cycle. Since $\stratmin$ is a switching strategy and
    $\stratmax$ is a memoryless strategy, this implies that the cycle
    is contained in the part of $\stratmin$ where the decision is
    taken by the strategy $\stratmin^1$: since it is an NC-strategy,
    this implies that the sum of the weights along the cycle is at
    most $-1$. But if this is the case, we may modify the switching
    strategy~$\stratmin$ to loop more in the same cycle (this is
    indeed a cycle in the timed game, not only in the untimed region
    game): against the optimal memoryless strategy $\stratmax$, this
    would imply that \MinPl has a sequence of strategies to obtain a
    value as small as \he wants, and thus
    $\Value_{\game_{\locMin,r}}(\locMin,\valuation)=-\infty$. This
    contradicts the absence of values $-\infty$ in the game. Thus, the
    prefix
    $(\loc_0,\valuation) \xrightarrow{c_0} \cdots (\loc_k,\valuation)$
    contains no cycles.
  \end{enumerate}
  
  \noindent We now explain how to conclude. The absence of cycles
  implies that the sum of the discrete weights
  $w=\price(\loc_0,\loc_1)+\cdots+\price(\loc_{k-1},\loc_k)$ belongs
  to the set
  $[-(|\Locs|-1)\maxPriceTrans,(|\Locs|-1)\maxPriceTrans]\cap \Z$, and
  we have
  $\Value_{\game_{\locMin,r}}(\locMin,\valuation) =
  w+\fgoal_{\loc_k}(\valuation)$.  Notice that the previous
  developments also show that for all~$\valuation'\in I'$ (here,
  $\valuation<\valuation'$ is not needed),
  $\Value_{\game_{\locMin,r}}(\locMin,\valuation') =
  w+\fgoal_{\loc_k}(\valuation')$, with the same location $\loc_k$,
  and length $k$. Since this equality holds on $I'\subseteq I$ which
  is not a singleton, and $\Value_{\game_{\locMin,r}}(\locMin)$ is
  affine on $I$, it holds everywhere on $I$. This shows the result
  since $w+\fgoal_{\loc_k} \in \F_\game$.
\end{proof}

We now prove the termination of the sequence of $r_i$'s described
earlier. This is achieved by showing why, for all $i$, the owner of
$\locMin$ has a strictly better strategy in
configuration~$(\locMin,r_{i+1})$ than waiting until $r_i$ in location
$\locMin$.

\begin{lem}\label{lem:strictlySmaller}\label{lem:stationarysequence-locMin}
  If $\game_{\locMin,r_i}$ is finitely optimal for all $i\geq 0$ for which $r_i$ is defined, then
  \begin{enumerate}
  \item there exists $j\leq |\F_\game|^2+2$ such that $r_j=0$; and
  \item\label{item:strictlySmaller} denoting $j$ the number such that $r_{j}=0$ we have
  for all $0\leq i\leq j-2$ that if
    $\locMin\in\LocsMin$ (respectively, $\LocsMax$),
    $\Val_\game(\locMin,r_{i+1}) <
    \Val_\game(\locMin,r_i)+(r_i-r_{i+1})\price(\locMin)$
    (respectively,
    $\Val_\game(\locMin,r_{i+1}) >
    \Val_\game(\locMin,r_i)+(r_i-r_{i+1})\price(\locMin)$).
  \end{enumerate}
\end{lem}
\begin{proof}
  \begin{enumerate}
  \item We consider first the case where
    \underline{$\locMin\in \LocsMax$}, showing a better bound
    $j\leq |\F_\game|+2$. The main ingredient is to show that a
    function of $\F_\game$ cannot be used twice in
    $\Val_\game(\locMin)$. Let $i>0$ such that $r_i \neq 0$ (if there
    exist no such~$i$ then $r_1=0$). Recall from
    Lemma~\ref{lem:operatorNext} that there exists $r'_i<r_i$ such
    that $\Value_{\game_{\locMin,r_{i-1}}}(\locMin)$ is affine on
    $[r'_i,r_i]$, of slope greater than $-\price(\locMin)$. In
    particular,
  \[
    \frac{\Value_{\game_{\locMin,r_{i-1}}}(\locMin,r_i)-
      \Value_{\game_{\locMin,r_{i-1}}}(\locMin,r'_i)} {r_i-r'_i} >
    -\price(\locMin)\,.\] Lemma~\ref{lem:eqGoal} states that on
  $[r'_i,r_i]$, $\Value_{\game_{\locMin,r_{i-1}}}(\locMin)$ is equal
  to some $f_i\in \F_\game$. As $f_i$ is an affine function,
  $f_i(r_i) = \Value_{\game_{\locMin,r_{i-1}}}(\locMin,r_i)$, and
  $f_i(r'_i) = \Value_{\game_{\locMin,r_{i-1}}}(\locMin,r'_i)$. Thus,
  for all~$\valuation$,
  \[f_i(\valuation) = \Value_{\game_{\locMin,r_{i-1}}}(\locMin,r_i) +
    \frac{\Value_{\game_{\locMin,r_{i-1}}}(\locMin,r'_i)-
      \Value_{\game_{\locMin,r_{i-1}}}(\locMin,r_i)} {r_i - r'_i} (r_i
    -\valuation).\] Since $\game_{\locMin,r_{i-1}}$ is assumed to be
  finitely optimal, we know that
  $\Value_{\game_{\locMin,r_{i-1}}}(\locMin,r_i) =
  \Value_{\game}(\locMin,r_i)$, by definition of
  $r_i=\Next_{\locMin}(r_{i-1})$.  Therefore, combining both
  equalities above, for all clock values $\valuation<r_i$, we have
  $f_i(\valuation) < \Value_\game(\locMin,r_i) +\price(\locMin)
  (r_i-\valuation)$.

  Consider then $j> i$ such that $r_j \neq 0$.  We claim that
  $f_j \neq f_i$.  Indeed, we have
  $\Val_\game(\locMin,r_j) = f_j(r_j)$.  As, in $\game$, $\locMin$ is
  a non-urgent location, Lemma~\ref{lem:rate} ensures that 
  \[\Value_\game(\locMin,r_j) \geq \Value_\game(\locMin,r_i)
    +\price(\locMin) (r_i-r_j)\,.\] As for all $i'$,
  $\Value_\game(\locMin,r_{i'}) = f_{i'}(r_{i'})$, the equality above
  is equivalent to
  $f_j(r_j) \geq f_i(r_i) +\price(\locMin) (r_i-r_j)$.  Recall that
  $f_i$ has a slope strictly greater that $-\price(\locMin)$,
  therefore
  $f_i(r_j) < f_i(r_i) +\price(\locMin) (r_i-r_j) \leq f_j(r_j)$. As a
  consequence $f_i \neq f_j$ (this is depicted in
  Figure~\ref{fig:fineqfj}).

  \begin{figure}\centering
    \begin{tikzpicture}
      \node[above] at (0,5) {$\Value_{\game}(\locMin,\valuation)$};
      \node[below] at (7,0) {$\valuation$};
      \draw[->] (-0.3,0) -- (7,0);
      \draw[->] (0,-0.3) -- (0,5);
      \draw[gray] (6,0) -- (6,5);
      \draw[gray] (3,0) -- (3,5);
      \draw[dashed] (7,4) -- (-0.5,0.25);
      \node[above] at (7,4) {$-\price(\locMin)$};
      
      \draw (6,3.5) -- ++(-0.75,-0.375) -- ++ (-0.75,0.2) -- ++ (-0.75,-0.1) -- ++ (-0.75,-0.375) -- ++ (-1,-0.5) -- ++ (-1,0.4);
      \draw (6,3.5) -- ++(0.5,-0.3);
      \draw[dashed] (3,2.85) -- ++(-2,-1);
      
      \draw[dotted] (6,3.5) -- ++(-1.5,-2);
      \draw[dotted] (3,2.85) -- ++(-1.5,-1.5);
      \node[below] at (4.5,1.5) {$f_i$};
      \node[above] at (1.5,1.35) {$f_j$};
      
      \node[below] at (3,0) {$r_j$};
      \node[below] at (6,0) {$r_i$};
      
    \end{tikzpicture}
    
    \caption{The case $\locMin\in\LocsMax$: a geometric proof of
      $f_i\neq f_j$. The dotted lines represents $f_i$ and $f_j$, the
      dashed lines have slope $-\price(\locMin)$, and the plain line
      depicts $\Value_\game(\locMin)$. Because the slope of
      $f_i$ is strictly smaller than $-\price(\locMin)$, and the value
      at $r_j$ is above the dashed line it cannot be the case that
      $f_i(r_j) = \Value_\game(\locMin,r_j)=f_j(r_j)$.}
    \label{fig:fineqfj}
  \end{figure}

  Therefore, there cannot be more than $|\F_\game|+1$ non-null
  elements in the sequence $r_0\geq r_1\geq \cdots$, which proves that
  there exists $i\leq |\F_\game|+2$ such that $r_i=0$.

  \medskip
  
  We continue with the case where \underline{$\locMin\in
    \LocsMin$}. We generalise the previous arguments that may no
  longer be true in this case (the same function of $\F_\game$ could
  be used twice in $\Value_\game(\locMin)$), by showing that
  in-between two successive points $r_{i+1}$ and $r_i$, there is
  always one ``full segment'' of $\F_\game$ (i.e.~it encounters at
  least one point that is the intersection of two functions of
  $\F_\game$, and there are $|\F_\game|^2$ many such points). Let
  $r_\infty = \inf\{r_i \mid i\geq 0\}$.  In this case, we look at the
  affine parts of $\Val_\game(\locMin)$ with a slope greater than
  $-\price(\locMin)$, and we show that there can only be finitely many
  such segments in $[r_\infty,1]$.  We then show that there is at
  least one such segment contained in $[r_{i+1},r_i]$ for all $i$,
  bounding the size of the sequence.

  In the following, we call \emph{segment} every interval
  $[a,b]\subset (r_\infty,1]$ such that $a$ and $b$ are two
  consecutive cutpoints of the cost function $\Val_\game(\locMin)$
  over $(r_\infty,1]$. Recall that it means that~$\Val_\game(\locMin)$
  is affine on $[a,b]$, and if we let $a'$ be the greatest cutpoint
  smaller than $a$, and~$b'$ be the smallest cutpoint greater than
  $b$, the slopes of $\Val_\game(\locMin)$ on $[a',a]$ and $[b,b']$
  are different from the slope on $[a,b]$.  We abuse the notations by
  referring to \emph{the slope of a segment $[a,b]$} for the slope of
  $\Val_\game(\locMin)$ on $[a,b]$ and simply call \emph{cutpoint} a
  cutpoint of $\Val_\game(\locMin)$.

  To every segment $[a,b]$ with a slope greater than
  $-\price(\locMin)$, we associate a function $f_{[a,b]}\in \F_\game$
  as follows. Let $i$ be the smallest index such that
  $[a,b]\cap[r_{i+1},r_i]$ is a non singleton interval $[a',b']$.
  Lemma~\ref{lem:eqGoal} ensures that there exists
  $f_{[a,b]}\in \F_\game$ such that for all $\valuation \in [a',b']$,
  $\Value_\game(\locMin,\valuation) = f_{[a,b]}(\valuation)$.

  Consider now two disjoint segments $[a,b]$ and $[c,d]$ with a slope
  greater than~$-\price(\locMin)$, and assume that
  $f_{[a,b]}=f_{[c,d]}$ (in particular both segments have the same
  slope). Without loss of generality, assume that $b<c$. We claim that
  there exists a segment $[e,g]$ in-between $[a,b]$ and $[c,d]$ with a
  slope greater than the slope of $[c,d]$, and that $f_{[e,g]}$ and
  $f_{[a,b]}$ intersect over $x\in[b,c]$,
  i.e.~$f_{[e,g]}(x) = f_{[a,b]}(x)$ (depicted in
  Figure~\ref{fig:pureGeometry}). We prove it now.

  \begin{figure}\centering
    \begin{tikzpicture}
      \draw (0,0) -- (1,1) -- (2,0.5) -- (2.25,4) -- (4,6.5) -- (5,5) -- (6,6);
      \draw[dashed] (-0.5,-0.5) -- (6.5,6.5);
      
      \node[above] at (0,0) {$a$}; 
      \node[above] at (1,1) {$b$};
      \node[below] at (2,0.5) {$e$}; 
      \node[above] at (2.25,4) {$g$};
      \node[above] at (4,6.5) {$\alpha$}; 
      \node[below] at (5,5) {$c$}; 
      \node[below] at (6,6) {$d$};
      \node at (2.115,2.115) {$\bullet$};
      \node[below right] at (2.115,2.115) {$x$};
    \end{tikzpicture}
    \caption{In order for the segments $[a,b]$ and $[c,d]$ to be
      aligned, there must exist a segment with a biggest slope
      crossing $f_{[a,b]}$ (represented by a dashed line) between $b$
      and $c$.}
    \label{fig:pureGeometry}
  \end{figure}
  
  Let $\alpha$ be the greatest cutpoint smaller than $c$. We know that
  the slope of $[\alpha,c]$ is different from the one of $[c,d]$.  If
  it is greater then define $e=\alpha$ and $x=g=c$, those indeed
  satisfy the property.  If the slope of $[\alpha,c]$ is smaller than
  the one of $[c,d]$, then for all $\valuation \in [\alpha,c)$,
  $\Val_\game(\locMin,\valuation) > f_{[c,d]}(\valuation)$.  Let $x$
  be the greatest point in $[b,\alpha]$ such that
  $\Val_\game(\locMin,x) = f_{[c,d]}(x)$. We know that it exists since
  $\Val_\game(\locMin,b) = f_{[c,d]}(b)$, and $\Val_\game(\locMin)$ is
  continuous.  Observe that
  $\Val_\game(\locMin,\valuation) > f_{[c,d]}(\valuation)$, for all
  $x<\valuation<c$.  Finally, let $g$ be the smallest cutpoint of
  $\Val_\game(\locMin)$ strictly greater than $x$, and $e$ the
  greatest cutpoint of $\Val_\game(\locMin)$ smaller than or equal to
  $x$. By construction, $[e,g]$ is a segment that contains~$x$. The
  slope of the segment $[e,g]$ is
  $s_{[e,g]}=\frac{\Val_\game(\locMin,g) -
    \Val_\game(\locMin,x)}{g-x}$, and the slope of the segment~$[c,d]$
  is equal to $s_{[c,d]} = \frac{f_{[c,d]}(g) -
    f_{[c,d]}(x)}{g-x}$. Remembering that
  $\Val_\game(\locMin,x) = f_{[c,d]}(x)$, and that
  $\Val_\game(\locMin,g) > f_{[c,d]}(g)$ since $g\in (x,c)$, we obtain
  that $s_{[e,g]} > s_{[c,d]}$. Finally, since
  $\Val_\game(\locMin,x) = f_{[c,d]}(x) = f_{[e,g]}(x)$, $x$ is indeed
  the intersection point of $f_{[c,d]}=f_{[a,b]}$ and $f_{[e,g]}$,
  which concludes the proof of the previous claim.

  For every function $f\in \F_\game$, there are less than $|\F_\game|$
  intersection points between $f$ and the other functions of
  $\F_\game$ (at most one for each pair $(f,f')$). If $f$ has a slope
  greater than~$-\price(\locMin)$, thanks to the previous paragraph,
  we know that there are at most $|\F_\game|$ segments $[a,b]$ such
  that $f_{[a,b]}=f$. Summing over all possible functions $f$, there
  are at most $|\F_\game|^2$ segments with a slope greater than
  $-\price(\locMin)$.

  Now, we link those segments with the clock values $r_i$'s, for
  $i>0$.  By item~\ref{item:strictlySmaller}, thanks to the
  finite-optimality of $\game_{\locMin,r_{i}}$,
  $\Val_\game(\locMin,r_{i+1}) < (r_i-r_{i+1}) \price(\locMin) +
  \Val_\game(\locMin,r_i)$.  Furthermore,
  Lemma~\ref{lem:r_2-r_1-r_0} states that the slope of the
  segment directly on the left of $r_i$ is equal
  to~$-\price(\locMin)$. With the previous inequality in mind, this
  cannot be the case if $\Value_\game(\locMin)$ is affine over the
  whole interval $[r_{i+1},r_{i}]$. Thus, there exists a segment
  $[a,b]$ of slope strictly greater than~$-\price(\locMin)$ such that
  $b\in [r_{i+1},r_{i}]$. As we also know that the slope left to
  $r_{i+1}$ is $-\price(\locMin)$, it must be the case that
  $a\in [r_{i+1},r_i]$. Hence, we have shown that in-between $r_{i+1}$
  and $r_i$, there is always a segment (this is depicted in
  Figure~\ref{fig:slopesLocMin}). As the number of such segments is
  bounded by $|\F_\game|^2$, we know that the sequence $r_i$ is
  stationary in at most $|\F_\game|^2+1$ steps, i.e.~that there exists
  $i\leq |\F_\game|^2+1$ such that $r_i=0$.
  \begin{figure}\centering
    \begin{tikzpicture}
      \node[above] at (0,5) {$\Value_{\game}(\locMin,\valuation)$};
      \node[below] at (7,0) {$\valuation$};
      \draw[->] (-0.3,0) -- (7,0);
      \draw[->] (0,-0.3) -- (0,5);
      \draw[gray] (6,0) -- (6,5);
      \draw[gray] (3,0) -- (3,5);
      \draw[dashed] (7,4) -- (-0.5,0.25);
      \node[above] at (7,4) {$-\price(\locMin)$};
      
      \node[below] at (3,0) {$r_{i+1}$};
      \node[below] at (6,0) {$r_i$};
      
      \draw (6,3.5) -- (5,3);
      \draw[double] (5,3) --  (4,1.5);
      \draw (4,1.5) -- (3,1) -- (2,0.5) --  (1,-0.3);
      \node at (5,3) {$\bullet$};
      \node at (4,1.5) {$\bullet$};
      \node[left, text width = 2.5cm, text centered] at (0,2) {$\Value_\game(\locMin, r_i) +$\\ $\price(\locMin) (r_{i}-r_{i+1})$};
      \draw[gray] (0,2) -- (3,2);
      
      \draw [dashed] (3,1) -- ++(-2.5,-1.25);

    \end{tikzpicture}
    \caption{The case $\locMin\in\LocsMin$: as the value at $r_{i+1}$
      is strictly below
      $\Value_\game(\locMin, r_i) + \price(\locMin) (r_{i}-r_{i+1})$, as the
      slope on the left of $r_i$ and of $r_{i+1}$ is
      $-\price(\locMin)$, there must exist a segment (represented with
      a double line) with slope greater than $-\price(\locMin)$ in
      $[r_{i+1},r_i)$. }
    \label{fig:slopesLocMin}
  \end{figure}
\item We assume $\locMin\in \LocsMin$, since the proof of the other
  case only differs with respect to the sense of the
  inequalities. From Lemma~\ref{lem:operatorNext}, we know that in
  $\game_{\locMin,r_i}$, if $r_{i+1}>0$, there exists $r'<r_{i+1}$
  such that $\Value_{\game_{\locMin,r_i}}(\locMin)$ is affine on
  $[r',r_{i+1}]$ and its slope is smaller than~$-\price(\locMin)$,
  i.e.~$\Value_{\game_{\locMin,r_i}}(\locMin,r_{i+1})<
  \Value_{\game_{\locMin,r_i}}(\locMin, r') -
  (r_{i+1}-r')\price(\locMin)$.  Lemma~\ref{lem:bounded} also ensures
  that
  $\Value_{\game_{\locMin,r_i}}(\locMin,r')\leq
  \Value_\game(\locMin,r_i) + (r_i-r')\price(\locMin)$.  Combining
  both inequalities allows us to conclude.\qedhere
  \end{enumerate}
\end{proof}

We iterate this construction to obtain the finite optimality: 

\begin{thm}\label{the:finiteOptimality}
  Every \SPTG $\game$ is finitely optimal.
\end{thm}
\begin{proof} 
  We show by induction on $n\geq 0$ that every $r$-\SPTG
  $\game$ with $n$ non-urgent non-final locations is finitely optimal.

  The base case $n=0$ is given by
  Proposition~\ref{prop:baseCase}. 
  
  Now, assume that $\game$ has at
  least one non-urgent location, and assume $\locMin$ is a non-urgent
  location with minimum weight. By induction hypothesis, all
  $r'$-\SPTG{s} $\game_{\locMin,r'}$ are finitely optimal for all
  $r'\in [0,r]$.  Let $r_0> r_1> \cdots$ be the decreasing sequence
  defined by $r_0=r$ and $r_i=\Next_{\locMin}(r_{i-1})$ for all
  $i\geq 1$. By Lemma~\ref{lem:stationarysequence-locMin}, there
  exists $j\leq |\F_\game|^2+2$ such that $r_j = 0$. Moreover, for all
  $0<i\leq j$, $\Value_\game=\Value_{\game_{\locMin,r_{i-1}}}$ on
  $[r_{i},r_{i-1}]$ by definition of $r_i = \Next_{\locMin}(r_{i-1})$,
  so that $\Value_\game(\loc)$ is a cost function on this interval
   by induction hypothesis.
  Finally, by using Proposition~\ref{lem:SameValue}, we can reconstruct fake-optimal and optimal strategies in
  $\game$ from the fake-optimal and optimal strategies of
  $\game_{\locMin,r_i}$.
\end{proof}

\subsection{SPTGs have a pseudo-polynomial number of cutpoints}
To prove that the number of cutpoints of value functions of SPTGs is at most pseudo-polynomial, we will need more knowledge about the $\Next_{\locMin}$
operator for all \SPTG{s} $\game$. First, if we are at the left of a given position $r_1$, the
next jump is further than $\Next_{\locMin}(r_{1})$:

\begin{lem}\label{lem:left-left}
  Let $\game$ be an \SPTG and $\locMin$ be a non-urgent location of minimum weight. Let $r_1$ and $r_2$ be such that $r_2 \leq r_1$. Then,
  $\Next_{\locMin}(r_{2}) \leq \Next_{\locMin}(r_{1})$.
\end{lem}
\begin{proof}
  Let $r'_1 = \Next_{\locMin}(r_{1})$. If $r_2\leq r'_1$, then the
  result is trivially true. We now suppose that $r'_1 < r_2 < r_1$. By
  definition, it suffices to show that
  \begin{equation}
    \forall \valuation\in[r'_{1}, r_2]\qquad
    \Value_{\game_{\locMin,r_1}}(\locMin,\valuation) = 
    \Value_{\game_{\locMin,r_2}}(\locMin,\valuation) 
    \label{eq:values-equal}
  \end{equation}
  Indeed, since $\Value_{\game_{\locMin,r_1}}(\locMin,\valuation) =
  \Value_{\game}(\locMin,\valuation)$ for all $\valuation \in [r'_1,r_1]$, this
  implies that  for all $\valuation \in [r'_{1},r_2]$,
  $\Value_{\game_{\locMin,r_2}}(\locMin,\valuation) =
  \Value_{\game}(\locMin,\valuation)$, and thus that $\Next_{\locMin}(r_{2})
  \leq r'_{1}$.

  To show~\eqref{eq:values-equal}, for $\varepsilon>0$, we will need
  $\varepsilon$-optimal strategies for both players in $\game_{\locMin,r_1}$ and
  $\game_{r_1}$. We build them considering two separate cases. 
  \begin{itemize}
    \item if $\locMin$ belongs to $\MaxPl$, let $\stratmax^*$ be an
    $\varepsilon$-optimal strategy of $\MaxPl$ in $\game_{\locMin,r_1}$: then it
    is also $\varepsilon$-optimal strategy in $\game_{r_1}$ since the values of
    those games are the same on $[r'_1,r_1]$, and strategies of $\MinPl$ are
    identical too;

    \item  if $\locMin$ belongs to $\MinPl$, let $\stratmax^*$ be an
    $\varepsilon$-optimal strategy of $\MaxPl$ in $\game_{r_1}$: then it is also
    $\varepsilon$-optimal in $\game_{\locMin,r_1}$ since $\MinPl$ has less
    capabilities in this game than in $\game_{r_1}$, while strategies of
    $\MaxPl$ are unchanged. 
  \end{itemize}
  We do the same case distinction to define an $\varepsilon$-optimal strategy $\stratmin^*$ both in $\game_{\locMin,r_1}$ and
  $\game_{r_1}$. In particular, we have, for all
  $\valuation\in[r'_1,r_1]$
  \[ \Value_{\game_{r_1}}(\loc,\valuation) - \varepsilon
    \leq \Value^{\stratmax^*}_{\game_{r_1}}(\loc,\valuation)
    \qquad \text{ and } \qquad
    \Value^{\stratmin^*}_{\game_{r_1}}(\loc,\valuation)  \leq
    \Value_{\game_{r_1}}(\loc,\valuation) + \varepsilon 
  \]
  
  First, let us show that $\Value_{\game_{\locMin,r_1}}(\locMin,\valuation)
  \leq \Value_{\game_{\locMin,r_2}}(\locMin,\valuation)$,
  i.e.~$\Value_{\game_{r_1}}(\locMin,\valuation) \leq
  \Value_{\game_{\locMin,r_2}}(\locMin,\valuation)$. To do so, we consider the
  definition of the value and show that for all $\valuation\in [r'_1,r_2]$,
  \begin{equation}
    \Value^{\stratmax^*}_{\game_{r_1}}(\locMin,\valuation)
    \leq \sup_{\stratmax} \inf_{\stratmin}
    \cost{\CPlay{(\locMin,\valuation),\stratmin,\stratmax}}+\varepsilon
    \label{eq:leq-bothgames}
  \end{equation}
  where the play on the right is a play of the game
  $\game_{\locMin,r_2}$. In particular, this will imply that
  $\Value_{\game_{r_1}}(\locMin,\valuation)-\varepsilon \leq
  \Value_{\game_{\locMin,r_2}}(\locMin,\valuation) + \varepsilon$
  which allows us to conclude by letting $\varepsilon$ go to $0$. We
  thus build a strategy $\stratmax$ in $\game_{\locMin,r_2}$ as
  follows: it simply follows what $\stratmax^*$ prescribes to do in
  $\game_{r_1}$ (especially when it jumps from $\locMin$ to
  ${\locMin}^f$ if $\locMin$ belongs to $\MaxPl$) except when
  $\stratmax^*$ wants to jump on the right of $r_2$, from any location
  $\loc$, in which case $\stratmax$ goes to the location $\loc^f$ in
  valuation $r_2$. We now explain why
  \[\Value^{\stratmax^*}_{\game_{r_1}}(\locMin,\valuation)
    \leq \inf_{\stratmin}
    \cost{\CPlay{(\locMin,\valuation),\stratmin,\stratmax}} +
    \varepsilon\] To do so, we consider any strategy $\stratmin$ of
  $\MinPl$ in $\game_{\locMin,r_2}$, and build a strategy $\stratmin'$
  of $\MinPl$ in $\game_{r_1}$ that gets a smaller payoff. The
  strategy $\stratmin'$ mimics $\stratmin$ except when it jumps in a
  location $\loc^f$: instead $\stratmin'$ delays in $\loc$ until $r_2$
  and then performs the action prescribed by $\stratmin^*$ in
  $(\loc,r_2)$. Notice that this is a legal move since we play in
  $\game_{r_1}$ where the location $\locMin$ has not been made urgent.

  We now compare the prices of two plays:
  \begin{itemize}
  \item the play $\rho_1$ obtained from $(\loc,\valuation)$ in
    $\game_{r_1}$ by following $\stratmax^*$ (the
    $\varepsilon$-optimal strategy we have fixed) and $\stratmin'$
    (that we have built);
  \item and the play $\rho_2$ obtained from $(\loc,\valuation)$ in
    $\game_{\locMin,r_2}$ by following $\stratmax$ (that we have built)
    and $\stratmin$ (that we have fixed).
  \end{itemize}
  We need to show that
  $\cost{\rho_1} \leq \cost{\rho_2} + \varepsilon$ to conclude.

  If $\rho_2$ stops in a location different from $\loc^f$ for any
  $\loc$, then this play is also a play of $\game_{r_1}$ conforming to
  $\stratmax^*$ (by construction of $\stratmax$) and $\stratmin'$
  (that we have built), and is thus equal to $\rho_1$. We conclude
  directly that $\cost{\rho_1} = \cost{\rho_2}$.

  Otherwise, $\rho_2$ stops in a configuration
  $(\loc^f,\valuation)$. Let $\rho'_2$ be the partial play obtained
  from $\rho_2$ by removing its last transition.
  Then,
  \[\cost{\rho_2} = \cost{\rho'_2} + (r_2-\valuation)
    \price(\loc)+\Value_{\game_{\locMin,r_1}}(\loc,r_2)\] Let
  $\rho'_1$ be the play obtained by following $\stratmin^*$ and
  $\stratmax^*$ in $\game_{r_1}$ from $(\loc,r_2)$: $\rho'_1$ has
  price at most
  $\Value^{\stratmin^*}_{\game_{r_1}}(\loc,r_2)\leq
  \Value_{\game_{r_1}}(\loc,r_2)+\varepsilon$ since it follows
  $\stratmin^*$. However, the play $\rho_1$ is the concatenation of
  the play $\rho'_2$, a delay of $r_2-\valuation$ in $\loc$, and the
  play $\rho'_1$. Thus
  \begin{align*}
    \cost{\rho_1}
    &=\cost{\rho'_2}+(r_2-\valuation)
      \price(\loc) + \cost{\rho'_1}\\
    &=
      \cost{\rho_2}-\Value_{\game_{\locMin,r_1}}(\loc,r_2) +
      \cost{\rho'_1} \\
    &\leq \cost{\rho_2}+\varepsilon
  \end{align*} This concludes all the cases and thus the proof of~\eqref{eq:leq-bothgames}.


  The other inequality $\Value_{\game_{\locMin,r_1}}(\locMin,\valuation) \geq
  \Value_{\game_{\locMin,r_2}}(\locMin,\valuation)$ is obtained
  symmetrically by showing that
  \[\Value^{\stratmin^*}_{\game_{r_1}}(\locMin,\valuation)
    \geq \inf_{\stratmin}\sup_{\stratmax}
    \cost{\CPlay{(\locMin,\valuation),\stratmin,\stratmax}}-\varepsilon\tag*{\qedhere}\]
\end{proof}

Then, we change our policy to make locations urgent: instead of making them urgent one by one by increasing order of weight, we make them all urgent at once. We now show that this makes us progress at least as fast in the $\Next$ functions (from now on, we reuse the exponents in the $\Next$ function to explain which game we consider):

\begin{lem}\label{lem:nextLl} Let $\game$ be an \SPTG with non-urgent
  locations $L'=\{\loc_1,\ldots, \loc_n\}$ ordered in increasing order of weight. Then
  for all valuations $r\leq 1$, \[\Next_{L'}^\game(r) \leq \max_{1\leq i\leq
  n}\Next_{\loc_i}^{\game_{\{\loc_1,\dots,\loc_{i-1}\},r}}(r)\]
\end{lem} 
\begin{proof}
  Let $r' = \max_{1\leq i\leq n }\Next_{\loc_i}^{\game_{\{\loc_1,\dots,\loc_{i-1}\},r}}(r)$.
  By definition of $\Next$ function, for all $\loc\in L$, $0\leq i\leq n-1$, and $\valuation\in [r',r]$, 
  $\Value_{\game_{\loc_0,\dots,\loc_{i-1}}}(\loc,\valuation ) = 
  \Value_{\game_{\loc_0,\dots,\loc_{i}}}(\loc,\valuation)$. Thus 
  $\Value_{\game}(\loc,\valuation)=\Value_{\game_{L',r}}(\loc,\valuation)$ for all $\valuation\in [r',r]$, meaning that $\Next_{L'}^\game(r)\leq r'$.
  \end{proof}
  
  This allows us to bound the number of steps of $\Next$ when making all locations urgent at once, generalising the bound obtained in Lemma~\ref{lem:stationarysequence-locMin} when making only one location urgent:
  \begin{lem}\label{lem:all-at-once}
  Let $\game$ be an \SPTG with non-urgent locations $L'$.
  Denoting $(r_k)_{k\in \mathbb{N}}$ the sequence defined by $r_0 = 1$ and for all $i$, $r_{i+1}=\Next^\game_{L'}(r_i)$, then there exists $j\leq |L|(|\F_\game|^2+2)$ such that $r_j=0$.
  \end{lem}
  \begin{proof}
  Let $\loc_1,\dots, \loc_n$ be the locations of $L'$ by increasing order of
  weight. By Lemma~\ref{lem:stationarysequence-locMin}, for all $1\leq i \leq
  n$, the sequence $(j^{(i)}_{k})_{k\in\N}$ defined by  
  $j^{(i)}_0=1$ and $j^{(i)}_{k+1} =
  \Next_{\loc_i}^{\game_{\loc_1,\dots,\loc_{i-1}}}(j^{(i)}_k)$ for all $k\in \N$
  is stationary to $0$: there exists $k_i\leq |\F_\game|^2+2$ such that
  $j^i_{k_i}=0$. We now build the decreasing sequence $(j_k)_{0\leq k\leq t}$ by interleaving those
  $n$ sequences. We thus have $t\leq |L'|(|\F_\game|^2+2)\leq |L|(|\F_\game|^2+2)$.
  
  For all $k\leq t$, we have $r_k \leq j_k$. Indeed,
  $r_0 = 1 = j_0$ (as $j_0^{(i)}=1$ for all $i\leq n$).
  Assume that $r_k \leq j_k$ for some $k\leq t$
  and, for all $1\leq i\leq n$, let $n_i$ be the greatest index such that 
  $j_k\leq j^{(i)}_{n_i}$, so that $r_k \leq j^{(i)}_{n_i}$. By definition of the sequence $(j_k)_{0\leq k\leq t}$, we then have $j_{k+1} = \max_{1\leq i\leq n }(j^{(i)}_{n_i+1})$. Thus
  \begin{align*}
  r_{k+1} = \Next_{L'}^\game(r_k) &\leq
  \max_{1\leq i\leq n}\Next_{\loc_i}^{\game_{\loc_1,\dots,\loc_{i-1}}}(r_k) & \mbox{(by Lemma~\ref{lem:nextLl})}\\
  &\leq   \max_{1\leq i\leq n }\Next_{\loc_i}^{\game_{\loc_1,\dots,\loc_{i-1}}}(j^{(i)}_{n_i})& \mbox{(by Lemma~\ref{lem:left-left})}\\
  &=  \max_{1\leq i\leq n }(j^{(i)}_{n_i+1})= j_{k+1}
  \end{align*}
  which concludes the induction.
  Hence the sequence $(r_k)_{k\in \mathbb{N}}$ reaches $0$ in at most $t$ steps, thus in at most than $|L|(|\F_\game|^2+2)$ steps.
  \end{proof}


\begin{thm}\label{the:ExpCutpoints} Let $\game$ be an \SPTG.\@ For all
  locations $\loc$, $\Value_\game(\loc)$ has at most
  $O\left((\maxPriceTrans)^4|\Locs|^9\right)$ cutpoints.
\end{thm}
\begin{proof} 
  By using the notations of Lemma~\ref{lem:all-at-once}, it suffices to show
  that the number of cutpoints of $\Value_\game(\loc)$ in the interval
  $[r_{i+1},r_{i}]$ (with $i$ from $1$ to $j-1\leq |L|(|\F_\game|^2+2)-1 =
  O\big({\maxPriceTrans}^2 |\Locs|^5\big)$) is at most
  $O((\maxPriceTrans)^2|\Locs|^4)$. However, on such an interval, we know that
  the value function $\Value_\game(\loc)$ is equal to
  $\Value_{\game_{L',r_{i}}}(\loc)$. But $\game_{L',r_{i}}$ is a game where all locations are urgent, and thus by Proposition~\ref{prop:baseCase}, its number of cutpoints is indeed bounded by $O((\maxPriceTrans)^2|\Locs|^4)$.
\end{proof}

\subsection{Algorithms to compute the value function}\label{sec:algor-comp-value}

The finite optimality of \SPTG{s} allows us to compute the value
functions. The proof of Theorem~\ref{the:finiteOptimality} suggests a
\emph{recursive} algorithm to do so:
from an \SPTG $\game$ with minimal non-urgent location $\locMin$,
solve recursively $\game_{\locMin,1}$,
$\game_{\locMin,\Next_{\locMin}(1)}$,
$\game_{\locMin,\Next_{\locMin}(\Next_{\locMin}(1))}$, etc.\ handling
the base case where all locations are urgent with
Algorithm~\ref{algo:value-iteration-fixed}. While our results above
show that this is correct and terminates with a pseudo-polynomial time
complexity, we propose instead to solve---without the need for
recursion---the sequence of games
$\game_{\Locs\setminus(\LocsUrg\cup \LocsFin),1}$,
$\game_{\Locs\setminus(\LocsUrg\cup
  \LocsFin),\Next_{\Locs\setminus(\LocsUrg\cup \LocsFin)}(1)}, \ldots$
i.e.~making all locations urgent at once. Lemma~\ref{lem:all-at-once} explains why this sequence of games correctly computes the value function of $\game$ and terminates after a pseudo-polynomial number of steps.  


\begin{algorithm}[tbp]
  \caption{\texttt{solve}($\game$)}\label{alg:solve}
  \KwIn{\SPTG
    $\game=(\LocsMin, \LocsMax, \LocsFin, \LocsUrg, \fgoalvec,
    \transitions, \price)$} %
  \DontPrintSemicolon%

  $\vec f = (f_\loc)_{\loc\in \locs} :=
  \SolveInstant(\game,1)$\tcc*[r]{$f_\loc\colon \{1\}\to \Rbar$}\label{alg:init}%
  $r := 1$\;%
  \While(\tcc*[f]{Invariant: $f_\loc\colon [r,1] \to \Rbar$}){$0<r$}{%
    $\game' := \Waiting(\game,r,\vec f(r))$ \tcc*[r]{$r$-\SPTG
      $\game' = (\LocsMin,\LocsMax,\LocsFin',\LocsUrg', \fgoalvec',
      T', \price')$}%
    $\LocsUrg' := \LocsUrg' \cup \locs$\tcc*[r]{every location is made
      urgent}%
    $b := r$\label{alg:game-constructed}\;%
    \Repeat(\tcc*[f]{Invariant:$f_\loc\colon [b,1]\to \Rbar$}){ $b=0$
      or $stop$}%
    {%
      $a := \max (\posscp_{\game'}\cap [0,b))$\label{alg:cutpoint}\;%
      $\vec x = (x_\loc)_{\loc\in \locs}
      :=\SolveInstant(\game',a)$\tcc*[r]{$x_\loc
        = \val_{\game'}(\loc,a)$}\label{alg:call-si}%
      \If{$\forall \loc\in\LocsMin \;
        \frac{f_\loc(b)-x_\loc}{b-a}\leq-\price(\loc)\land \forall
        \loc\in\LocsMax \;
        \frac{f_\loc(b)-x_\loc}{b-a}\geq-\price(\loc)$\label{line:begin}}%
      {%
        \lForEach{$\loc\in\locs$}
        {$f_\loc := \Big(\valuation\in[a,b] \mapsto f_\loc(b) +
          (\valuation-b)\frac{f_\loc(b)-x_\loc}{b-a}\Big) \opcf
          f_\loc$}\label{alg:concat}%
        $b:=a$\,; $stop := false$ }%
      \lElse{$stop := true$\label{line:end}}
    }%
    $r:= b$%
  }%
  \Return{$\vec f$}%
\end{algorithm}

Algorithm~\ref{alg:solve} implements these ideas. Each iteration of
the \textbf{while} loop computes a new game in the sequence
$\game_{\Locs\setminus(\LocsUrg\cup \LocsFin),1}$,
$\game_{\Locs\setminus(\LocsUrg\cup
  \LocsFin),\Next_{\Locs\setminus(\LocsUrg\cup \LocsFin)}(1)},\ldots$; solves it thanks to \SolveInstant; and thus computes
a new portion of $\val_\game$ on an interval on the left of the
current point $r\in [0,1]$. More precisely, the vector
$(\val_\game(\loc,1))_{\loc\in \Locs}$ is first computed in
line~\ref{alg:init}. Then, the algorithm enters the {\bf while} loop,
and the game $\game'$ obtained when reaching
line~\ref{alg:game-constructed} is
$\game_{\Locs\setminus(\LocsUrg\cup \LocsFin),1}$. Then, the algorithm enters the
{\bf repeat} loop to analyse this game. Instead of building the whole
value function of $\game'$, Algorithm~\ref{alg:solve} builds only the
parts of $\val_{\game'}$ that coincide with $\val_{\game}$. It
proceeds by enumerating the possible cutpoints $a$ of $\val_{\game'}$,
starting in~$r$, by decreasing clock values (line~\ref{alg:cutpoint}),
and computes the value of $\val_{\game'}$ in each cutpoint thanks to
\SolveInstant (line~\ref{alg:call-si}), which yields a new piece of
$\val_{\game'}$. Then, the {\bf if} in line~\ref{line:begin} checks
whether this new piece coincides with $\val_\game$, using the
condition given by Proposition~\ref{lem:SameValue}. If it is the case,
the piece of $\val_{\game'}$ is added to $f_\loc$
(line~\ref{alg:concat}); \textbf{repeat} is stopped otherwise. When
exiting the \textbf{repeat} loop, variable $b$ has value
$\Next_{\Locs\setminus(\LocsUrg\cup \LocsFin)}(1)$. Hence, at the next
iteration of the \textbf{while} loop,
$\game'=\game_{\Locs\setminus(\LocsUrg\cup
  \LocsFin),\Next_{\Locs\setminus(\LocsUrg\cup \LocsFin)}(1)}$ when
reaching line~\ref{alg:game-constructed}. By continuing this reasoning
inductively, one concludes that the successive iterations of the
\textbf{while} loop compute the sequence
$\game_{\Locs\setminus(\LocsUrg\cup \LocsFin),1}$,
$\game_{\Locs\setminus(\LocsUrg\cup
  \LocsFin),\Next_{\Locs\setminus(\LocsUrg\cup \LocsFin)}(1)},\ldots$
as announced, and rebuilds $\val_\game$ from them. 

Termination of the \textbf{while} loop in pseudo-polynomially
many steps is then ensured by Lemma~\ref{lem:all-at-once}. Similarly, the
termination of the internal \textbf{repeat} loop is ensured by the at
most pseudo-polynomial number of possible cutpoints and the $stop$
variable. As each of the non-trivial calls requires at most
pseudo-polynomial time, Algorithm~\ref{alg:solve} finishes in
pseudo-polynomial time, in total. Note that some \SPTG{s} indeed have
a pseudo-polynomial number of cutpoints~\cite{FeaIbs20} (even in the
case of only non-negative prices), which shows that our bound is
asymptotically tight.

\begin{rem}
  The pseudo-polynomial lower-bound on the number of cutpoints shown
  in~\cite{FeaIbs20} helps getting a PSPACE-hardness of the value
  problem consisting in deciding whether the value
  $\val_\game(\loc,0)$ is below a given rational threshold. We might
  thus wonder whether our upper-bound techniques help closing the
  gap. Unfortunately, this does not seem to be the case. Indeed, even
  if we transform our algorithm to only record the current values
  $(f_\loc(r))_{\loc\in\Locs}$ of the value function, we are not able
  to obtain that such values can be stored in polynomial space (should
  we obtain such a result, it would easily imply a polynomial space
  algorithm to compute the initial values
  $(f_\loc(0))_{\loc\in\Locs}$, since the rest of the algorithm, in
  particular \SolveInstant, performs in polynomial space). The problem
  comes from the growth of the various coefficients appearing during
  the algorithm, in particular the granularity of the rational
  cutpoints we encounter through the computation. Though unrealistic,
  if the cutpoint on the left of $r$ was always in the middle of the
  interval $[0,r]$, cutpoints would have the shape $1/2^x$ with $x$ an
  integer bounded pseudo-polynomially. Unfortunately, the denominator
  of this ratio cannot be stored in polynomial space. Thus, getting a
  polynomial space algorithm to solve SPTGs requires a better
  understanding of the \emph{granularity} of cutpoints, and not only a
  bound on their \emph{number}.
\end{rem}

\begin{exa} Figure~\ref{fig:val_sptg} shows the value functions of
  the \SPTG of Figure~\ref{fig:ex-ptg2}. Here is how
  Algorithm~\ref{alg:solve} obtains those functions. During the first
  iteration of the {\bf while} loop, the algorithm computes the
  correct value functions until the cutpoint $\frac 3 4$: 
  in the
  $repeat$ loop, at first $a= 9/10$ but the slope in $\loc_1$ is
  smaller than the slope that would be granted by waiting, as depicted
  in Figure~\ref{fig:ex-ptg2}. 
  Then, $a=3/4$ where the algorithm gives
  a slope of value $-16$ in $\loc_2$ while the weight of this location
  of \MaxPl is $-14$. During the first iteration of the {\bf while}
  loop, the inner {\bf repeat} loop thus ends with $r=3/4$. The next
  iterations of the {\bf while} loop end with $r=\frac 1 2$ (because
  $\loc_1$ does not pass the test in line~\ref{line:begin});
  $r=\frac 1 4$ (because of $\loc_2$) and finally with $r=0$, giving
  us the value functions on the entire interval $[0,1]$.
\end{exa}

\begin{figure}
  \begin{center}
    \begin{tikzpicture}[xscale=.8,yscale=0.55]

      \draw[->] (6,0) -- (11,0) node[anchor=north] {$x$};
      \draw	(6,0) node[anchor=south] {$0$}
      (7,0) node[anchor=south] {$\frac 1 4$}
      (8,0) node[anchor=south] {$\frac 1 2$}
      (9,0) node[anchor=south] {$\frac 3 4$}
      (10,0) node[anchor=north] {$1$};

      \draw[->] (6,0) -- (6,-4) node[anchor=east] {$\val(\loc_2,x)$};
      \draw	(6,-3.3) node[anchor=east] {$-9.5$}
      (6,-2) node[anchor=east] {$-6$}
      (6.05,-1.85) node[anchor=west] {$-5.5$}
      (6,-0.6) node[anchor=east] {$-2$}
      (6,0.3) node[anchor=east] {$1$};

      \draw[thick] (6,-3.3) -- (7,-2) -- (8,-1.85)--(9,-0.6)--(10,0.3);

      \draw[->] (6,-5) -- (11,-5) node[anchor=north] {$x$};
      \draw	(6,-5) node[anchor=south] {$0$}
      (7,-5) node[anchor=south] {$\frac 1 4$}
      (8,-5) node[anchor=south] {$\frac 1 2$}
      (9,-5) node[anchor=south] {$\frac 3 4$}
      (9.6,-5) node[anchor=south] {$\frac 9 {10}$}
      (10,-5) node[anchor=south] {$1$};

      \draw[->] (6,-5) -- (6,-9) node[anchor=east] {$\val(\loc_1,x)$};
      \draw	(6,-8.3) node[anchor=east] {$-9.5$}
      (6,-7) node[anchor=east] {$-6$}
      (6.05,-6.85) node[anchor=west] {$-5.5$}
      (6,-5.6) node[anchor=east] {$-2$}
      (6,-5.1) node[anchor=east] {$-0.2$};

      \draw[thick] (6,-8.3) -- (7,-7) -- (8,-6.85)--(9,-5.6)--(9.6,-5.1)--(10,-5);

      \draw[->] (-2,0) -- (3,0) node[anchor=north] {$x$};
      \draw	(-2,0) node[anchor=south] {$0$}
      (-1,0) node[anchor=south] {$\frac 1 4$}
      (0,0) node[anchor=south] {$\frac 1 2$}
      (2,0) node[anchor=south] {$1$};

      \draw[->] (-2,0) -- (-2,-4) node[anchor=east] {$\val(\loc_3,x)$};
      \draw	(-2,-3) node[anchor=east] {$-10$}
      (-2,-1.7) node[anchor=east] {$-6$}
      (-2.05,-1.5) node[anchor=west] {$-5.5$}
      (-2,-2) node[anchor=west] {$-7$};

      \draw[thick] (-2,-3) -- (-1,-1.7) -- (0,-1.5)--(2,-2);

      \draw[->] (-2,-5) -- (3,-5) node[anchor=north] {$x$};
      \draw	(-2,-5) node[anchor=south] {$0$}
      (2,-5) node[anchor=south] {$1$};

      \draw[->] (-2,-5) -- (-2,-9) node[anchor=east] {$\val(\loc_4,x)$};
      \draw	(-2,-6.3) node[anchor=east] {$-4$}
      (-2,-7.3) node[anchor=east] {$-7$};

      \draw[thick] (-2,-6.3) -- (2,-7.3);

      \draw[->] (6,-10) -- (11,-10) node[anchor=north] {$x$};
      \draw	(6,-10) node[anchor=south] {$0$}
      (9,-10) node[anchor=south] {$\frac 3 4$}
      (10,-10) node[anchor=south] {$1$};

      \draw[->] (6,-10) -- (6,-14) node[anchor=east] {$\val(\loc_5,x)$};
      \draw	(6,-13.5) node[anchor=east] {$-14$}
      (6,-10.5) node[anchor=east] {$-2$}
      (6,-9.8) node[anchor=east] {$1$};

      \draw[thick] (6,-13.5) -- (9,-10.5) --(10,-9.8);

      \draw[->] (6,-15) -- (11,-15) node[anchor=north] {$x$};
      \draw	(6,-15) node[anchor=south] {$0$}
      (10,-15) node[anchor=south] {$1$};

      \draw[->] (6,-15) -- (6,-19) node[anchor=east] {$\val(\loc_6,x)$};
      \draw	(6,-18) node[anchor=east] {$-11$}
      (6,-14.8) node[anchor=east] {$1$};

      \draw[thick] (6,-18) -- (10,-14.8);

      \draw[->] (-2,-10) -- (3,-10) node[anchor=north] {$x$};
      \draw	(-2,-10) node[anchor=south] {$0$}
      (2,-10) node[anchor=south] {$1$};

      \draw[->] (-2,-10) -- (-2,-14) node[anchor=east] {$\val(\loc_7,x)$};
      \draw	(-2,-13.5) node[anchor=east] {$-16$}
      (-2,-10) node[anchor=east] {$0$};

      \draw[thick] (-2,-13.5) -- (2,-10);

    \end{tikzpicture}

    \caption{Value functions of the \SPTG of Figure~\ref{fig:ex-ptg2}}
    \label{fig:val_sptg}
  \end{center}
\end{figure}


\section{Towards more complex \PTG{s}}
\label{app:nraptg2}

In~\cite{BouLar06,Rut11,DueIbs13}, \emph{general} \PTG{s} with
\emph{non-negative weights} are solved by reducing them to a finite
sequence of \SPTG{s}, by eliminating guards and resets. It is thus
natural to try and adapt these techniques to our general case, in
which case Algorithm~\ref{alg:solve} would allow us to solve
\emph{general \PTG{s} with arbitrary weights}.  Let us explain where are
the difficulties of such a generalisation.

The technique used to remove strict guards from the transitions of the
\PTG{s}, i.e.~guards of the form $(a,b]$, $[b,a)$ or $(a,b)$ with
$a,b\in \N$, consists in enhancing the locations with regions while
keeping an equivalent game. This technique \emph{can} be adapted to
arbitrary weights.  Formally, let
$\game=
(\LocsMin,\LocsMax,\LocsFin,\LocsUrg,\fgoalvec,\transitions,\price)$
be a \PTG.\@ We define the region-\PTG of~$\game$ as
$\game'=
(\LocsMin',\LocsMax',\LocsFin',\LocsUrg',\fgoalvec',\transitions',\price')$
where:
\begin{itemize}
\item $\LocsMin'=\{(\loc,I)\mid \loc\in \LocsMin, I\in \reggame\}$;
\item $\LocsMax'=\{(\loc,I)\mid \loc\in \LocsMax, I\in \reggame\}$;
\item $\LocsFin=\{(\loc,I)\mid \loc\in \LocsFin, I\in \reggame\}$;
\item $\LocsUrg=\{(\loc,I)\mid \loc\in \LocsUrg, I\in \reggame\}$;
\item for all $(\loc,I)\in \LocsFin'$, if $I$ is a
  singleton $\{a\}$ then $\fgoal'_{\loc,I}(a)=\fgoal_\loc(a)$,
  otherwise $I$ is an interval $(a,b)$, we then define for $x\in I$,
  $\fgoal'_{\loc,I}(x)=\fgoal_\loc(x)$ and extend $\fgoal'_{\loc,I}$
  on the borders of $I$ by continuity;
\item transitions given by
  \[\transitions'
    =\Bigg\{((\loc,I),\overline{I_g \cap I},R,(\loc',I'))\mid
      (\loc,I_g,R,\loc')\in \transitions, I' =
      {\footnotesize\begin{cases}
          I &\text{if } R=\bot\\
          \{0\} &\text{otherwise}
        \end{cases}}\Bigg\} \cup \WaitTr\]
    with 
    \begin{align*}
      \WaitTr    &=\big\{((\loc,(M_k,M_{k+1})),\{M_{k+1}\},\bot,(\loc,\{M_{k+1}\}))
                  \mid \loc\in \Locs, (M_k,M_{k+1}) \in \reggame\big\} \\
                & \quad\cup\big\{((\loc,\{M_{k}\}),\{M_{k}\},\bot,(\loc,(M_k,M_{k+1})))
                  \mid \loc\in \Locs, (M_k,M_{k+1}) \in \reggame\big\}\,;
  \end{align*}
\item
  $\forall (\loc,I)\in \Locs', \price'(\loc,I) = \price(\loc)$; and
  $\forall \delta'\in \transitions'$, we let $\price'(\delta')$ being
  the maximal (resp.~minimal) weight of a transition of $\transitions$
  giving rise to $\delta'$ in the definition above, knowing that
  transitions coming from $\WaitTr$ are given weight $0$, if $\loc$
  belongs to $\MaxPl$ (resp.~$\MinPl$).\footnote{Indeed, notice that a
    transition $((\loc,I),I'',R,(\loc',I'))\in \transitions'$ can be
    originated from two different transitions $(\loc,I_g,R,\loc')$ and
    $(\loc,I'_g,R,\loc')$ of $\transitions$ if $I_g$ and $I'_g$ both
    intersect $I$.}
\end{itemize}

It is easy to check that the region-PTG fulfils certain
\emph{invariants}. In all configurations $((\loc,\{M_k\}),\valuation)$
reachable from the clock value 0, the clock value $\valuation$ is
$M_k$. More interestingly, in all configurations
$((\loc,(M_k,M_{k+1})),\valuation)$ reachable from the clock value 0,
the clock value $\valuation$ is in $[M_k, M_{k+1}]$, and not only in
$(M_k,M_{k+1})$ as one might expect. Intuitively, we rely on
$\valuation=M_k$, for example, to denote a configuration of the
original game with a clock value arbitrarily close to $M_k$, but
greater than $M_k$. The game can thus take transitions with guard
$x>M_k$, but cannot take transitions with guard~$x=M_k$ anymore.

\begin{lem}
  \label{lem:region_ptg}
  Let $\game$ be a \PTG, and $\game'$ be its region-\PTG
  defined as before. For $(\loc,I)\in \Locs\times\reggame$ and
  $\valuation\in I$,
  $\val_\game(\loc,\valuation)=\val_{\game'}((\loc,I),\valuation)$.
  Moreover, we can transform an $\varepsilon$-optimal strategy of
  $\game'$\footnote{Recall that a strategy $\stratmin$ of $\MinPl$ is
  $\varepsilon$-optimal from location $\loc$ in $\game$ if 
  $\cost{\loc, \stratmin}\leq\valgame(\loc)+\varepsilon$.  
  } into an $\varepsilon'$-optimal strategy of $\game$ with
  $\varepsilon' < 2 \varepsilon$ and vice-versa.
\end{lem}
\begin{proof}
  Intuitively, the proof consists in replacing strategies of $\game'$
  where players can play on the borders of regions, by strategies of
  $\game$ that play increasingly close to the border as time
  passes. If played close enough, the loss created can be chosen as
  small as we want.

  Formally, let $\game$ be a \PTG, $\game'$ be its region-\PTG.\@
  First, for $\varepsilon>0$, we create a transformation $g$ of the
  plays of $\game'$ which do not end with a waiting transition to the
  plays of $\game$.  It is defined by induction on the length $n$ of
  the plays so that for a play $\run$ of length $n$ we have
  \begin{itemize}
  \item $|\puse{(\run)}- \puse{(g(\run}))| \leq
    2\maxPriceLoc(1-\frac{1}{2^n})\varepsilon$; and
  \item there exists $\loc\in \Locs$ and $I \in \reggame$ such that
    $g(\run)$ and $\run$ end in the respective locations $\loc$ and $(\loc,I)$, and
    their clock values are both in $I$ and differ of at most
    $\frac{1}{2^{n+1}}\varepsilon$.
  \end{itemize}

  If $n=0$, let $\run = ((\loc,I),\valuation)$ be a play of $\game'$
  of length $0$, then $g(\run) = (\loc,\valuation')$, where
  $\valuation'=\valuation \pm \frac \varepsilon 2$ if $I$ is not a
  singleton and $\valuation $ is an endpoint of $I$, and
  $\valuation'=\valuation$ otherwise (so that $\valuation'\in I$ in
  every case).
  
  For $n>0$, we suppose $g$ defined on every play of length at most
  $n$ which does not end with a waiting transition. Let
  $\run = ((q_1,I_1),\valuation_1)\xrightarrow{t_1,\transition_1,c_1}
  \dots \xrightarrow{t_n,\transition_n,c_n} ((q_n,I_n),\valuation_n)
  \xrightarrow{t_{n+1},\transition_{n+1},c_{n+1}}((q_{n+1},I_{n+1}),\valuation_{n+1})$
  with $\transition_{n+1}\notin \WaitTr$.  Let
  $last = \max(\{k\leq n\mid tr_k \notin \WaitTr\})$ (with
  $\max \emptyset = 0$). Then, by induction, there exists
  $\run'=(q_1,\valuation_1)\rightarrow\dots \rightarrow (q_{last
    +1},v'_{last+1})$ such that
  \begin{itemize}
  \item $g(\run_{|last}) = \run'$ (where $\run_{|last}$ is the prefix
    of length $last$ of $\run$),
  \item $|\puse{(\run_{|last})}- \puse{(g(\run_{|last}))}| \leq
    2\maxPriceLoc(1-\frac{1}{2^{last}})\epsilon$,
    and
  \item
    $|\valuation'_{last+1}-\valuation_{last+1}|\leq
    \frac{1}{2^{last+1}}\epsilon$.
  \end{itemize}
  Then we choose
  $g(\run)=\run'\xrightarrow{t,\transition'_{n+1},c}(q_{n+1},\valuation'_{n+1})$,
  with $\transition'_{n+1}$ the transition giving rise
  to $\transition_{n+1}$ (with the correct price) in the definition of
  the region-PTG, and where
\begin{itemize}
\item if $\transition_{n+1}$ is enabled in configuration
  $(q_{last+1}=q_{n}, \valuation_{n}+t_{n+1})$ of $\game$, then,
  $t = \valuation_{n}+t_{n+1}-\valuation'_{last}$;
\item otherwise, as the guards of $\game'$ are contained in the
  closure\footnote{By closure, we mean that, for example, a guard of the
    form $x>1$ becomes $x\geq 1$ in the region $[1,2]$.} of the guards
  of $\game$, then there exists~$z\in\{1,-1\}$ such that for
  $t = \valuation_{n}+t_{n+1}-\valuation'_{last} +
  \frac{z\epsilon}{2^{n+2}}$, $\transition_{n+1}$ is enabled in
  $\game$ and $\valuation_{last}'+t$ and $\valuation_{n}+t_{n+1}$
  belong to the same region.
\end{itemize}  
Thus, in both cases,
$|\valuation_{n+1}-\valuation'_{n+1} |\leq \frac{\epsilon}{2^{n+2}}$
and $\valuation_{n+1}\neq \valuation'_{n+1}$ iff $I$ is not a
singleton, $\valuation_{n+1}$ is on a border, $\valuation'_{n+1}$ is
close to this border and $\transition_{n+1}$ does not contain a reset.
Moreover,
  \begin{align*}
    |\puse{(\run)}-\puse{(g(\run))}|
    &=  |\puse{(\run_{|last})} + (\valuation_{n+1}-\valuation_{last}) \price(q_{last})+\price(\transition_{n+1})-\puse{(g(\run))}| \\
    &\leq  |\puse{(\run_{|last})}-\puse{(g(\run_{|last}))}| \\
    & \hspace{5mm}+
      | (\valuation_{n+1}-\valuation_{last}) \price(q_{last}) + \price(\transition_{n+1})+\puse{(g(\run_{|last}))}-\puse{(g(\run))}| \\
    &\leq 2\maxPriceLoc(1-\frac{1}{2^{last}})\epsilon+ |
      (\valuation'_{last}-\valuation_{last}) \price(q_{last})+
      (\valuation_{n+1}-\valuation'_{n+1})\price(q_{last})|\\ 
    &\leq 2\maxPriceLoc(1-\frac{1}{2^{last}})\epsilon+
      \left|\frac{\epsilon}{2^{last+1}}
      \price(q_{last})\right|+
      \left|\frac{\epsilon}{2^{n+2}}\price(q_{last})\right|\\ 
    &\leq 2\maxPriceLoc(1-\frac{1}{2^{last}})\epsilon+ \frac{\maxPriceLoc\epsilon}{2^{last+1}}+\frac{\maxPriceLoc\epsilon}{2^{n+2}}  \\
    &\leq  2\maxPriceLoc(1-\frac{1}{2^{last +1}})\epsilon \\
    &\leq  2\maxPriceLoc(1-\frac{1}{2^{n+1}})\epsilon \,.
  \end{align*}

  Let $\stratmin$ be a strategy of \MinPl in $\game$.  Using the
  transformation $g$, we will build by induction a strategy
  $\stratmin'$ in $\game'$ such that, for all plays $\run$ whose last
  transition does not belong to $\WaitTr$ and conforming with
  $\stratmin'$: $g(\run)$ conforms with $\stratmin$.

  Let $\run$ be a play of $\game'$ whose last transition does not
  belong to $\WaitTr$ such that $g(\run)$ conforms with $\stratmin$
  (which is the case of all plays of length 0). Then, $\run$ and
  $g(\run)$ end in locations $(q,I)$ and $q$ respectively.

  \begin{itemize}
  \item If $\run$ ends in a configuration of \MaxPl, then the choice
    of the next $(t,\transition)$-transition does not depend on
    $\stratmin$ or $\stratmin'$. Let $(t,\transition)$ be a choice of
    \MaxPl in $\game'$ with cost $c$.  If $\transition$ belongs to
    $\WaitTr$, then the new configuration also belongs to \MaxPl where
    \he will make another choice.  Let $\run'$ be the extension of
    $\run$ until the first transition $\transition'$ such that
    $\transition'\notin \WaitTr$.  The play $g(\run')$ conforms with
    $\stratmin$ as the configuration where $g(\run)$ ends is
    controlled by \MaxPl and $g(\run')$ only has one more transition
    than $g(\run)$.

  \item If $\run$ ends in a configuration of \MinPl, then there exists
    $t,\transition,c,q',\valuation'$ such that
    $g(\run)\xrightarrow{t,\transition,c}(q',\valuation')$ conforms
    with $\stratmin$. As taking a waiting transition does not change
    the ownership of the configuration, we consider here multiple
    successive choices of \MinPl as one choice: $\stratmin'(\run)$ is
    such that
    $\run'=\run\xrightarrow{t_1,\transition_1,c_1}\cdots
    \xrightarrow{t_k,\transition_k,c_k}
    ((q,I''),\valuation)\xrightarrow{t_{k+1},\transition,c_{k+1}}
    ((q',I'),\valuation')$ where
    $\forall i\leq k, \transition_i\in \WaitTr$ conforms with
    $\stratmin'$. This is possible as if $\transition$ is allowed in a
    configuration $(q,\valuation)$ in $\game$ then it is allowed too
    in a configuration $((q,I),\valuation)$ with the appropriate $I$.
    Then $g(\run') =g(\run)\xrightarrow{\delta,tr,c}(q',\valuation')$,
    thus $g(\run')$ conforms with $\stratmin$.
  \end{itemize}

  As no completed plays of $\game'$ end with a transition of $\WaitTr$,
  every completed play $\run$ conforming with $\stratmin'$ verifies
  that $g(\run)$ is a completed play conforming with $\stratmin$. 
  Moreover, the time valuation of $\run$ and $g(\run)$ differ by at most
  $\epsilon$ and belong to the same interval (potentially on one border), thus by
  definition of $\fgoal'$ the difference in final cost is bounded by 
  $K_{fin}\epsilon$ where $K_{fin}$ is the greatest absolute value of the slopes appearing
  in the piecewise affine functions within $\fgoalvec$.
  Thus for every
  configuration $s$, 
  $\costgame{\game'}{s, \stratmin'} \leq \costgame{\game}{s,
    \stratmin} + (2 \maxPriceLoc + K_{fin})\epsilon$. Therefore
  $\val_{\game'}(s)\leq \val_{\game}(s)$.

  \medskip

  Reciprocally, let $\stratmin'$ be a strategy of \MinPl in
  $\game'$. We will now build by induction a strategy $\stratmin$ in
  $\game$ such that for all plays $\run$ conforming with $\stratmin$,
  there exists a play in~$g^{-1}(\run)$ that conforms with
  $\stratmin'$.

  Let $\run$ be a play of $\game$ conforming with $\stratmin$ such
  that there exists $\run'\in g^{-1}(\run)$ conforming with
  $\stratmin'$ (which is the case of all plays of length 0). Plays
  $\run'$ and $\run$ end in the configurations~$((q,I),\valuation')$
  and $(q,\valuation)$ respectively.

  \begin{itemize}
  \item If $\run$ ends in configuration of \MaxPl, then the choice
    does not depend on $\stratmin$ or $\stratmin'$.
    Let~$(t,\transition)$ be a choice of \MaxPl in $\game$ with cost
    $c$ and let $\tilde{\run}$ be the extension of $\run$ by this
    choice.  There exists
    $(t_1,\transition_1,c_1),\dots,(t_{k+1},\transition_{k+1},c_{k+1})$
    such that $\forall i\leq k, \transition_i\in \WaitTr$,
    $\transition_{k+1}=\transition$ and
    $\sum_{i=1}^{k+1} t_i= \valuation+t-\valuation'$. Let
    $\run_c= \run'\xrightarrow{t_1,\transition_1,c_1}\cdots
    \xrightarrow{t_k,\transition_k,c_k}
    ((q,I''),\valuation_k)\xrightarrow{t_{k+1},
      \transition,c_{k+1}}((q',I'),\valuation_{k+1})$, then $\run_c$
    conforms with $\stratmin'$ (as \MinPl did not take a single
    decision) and $g(\run_c) = \tilde{\run}$.

  \item If $\run$ ends in a configuration of \MinPl, then there exists
    a play
    $\run_c=\run\xrightarrow{t_1,\transition_1,c_1}\cdots
    \xrightarrow{t_k,\transition_k,c_k}
    ((q,I''),\valuation_k)\xrightarrow{t_{k+1},
      \transition,c_{k+1}}((q',I'),\valuation_{k+1})$ such that
    $\run_c$ conforms with $\stratmin'$. We choose
    $\stratmin(\run)=(t,\transition)$ such that for the adequate cost
    $c$, $g(\run_c)=\run\xrightarrow{t,\transition,c}(q',v'')$. This
    is possible as $t +\valuation' \in I''$.
  \end{itemize}
  Every completed play $\run$ conforming with $\stratmin$ verifies
  $\exists \run' \in g^{-1}(\run)$ conforming with $\stratmin$. Thus,
  taking the final cost function into account as before,
  for every configuration $s$,
  $\costgame{\game} {s, \stratmin} \leq \costgame{\game'}{s,
    \stratmin}+(2 \maxPriceLoc + K_{fin})\epsilon$. Therefore
  $\val_{\game'}(s)\geq \val_{\game}(s)$. Hence
  $\val_{\game'}(s)= \val_{\game}(s)$.
\end{proof}

The technique used in~\cite{BouLar06,Rut11,DueIbs13} to remove resets
from \PTG{s}, however, consists in \emph{bounding} the number of clock
resets that can occur in each play following an optimal strategy of
\MinPl or \MaxPl. Then, the \PTG can be \emph{unfolded} into a
\emph{reset-acyclic} \PTG with the same value. By reset-acyclic, we
mean that no cycles in the configuration graph visit a transition with
a reset. This reset-acyclic \PTG can be decomposed into a finite
number of components that contain no reset and are linked by
transitions with resets. These components can be solved iteratively,
from the bottom to the top, turning them into \SPTG{s}.  Thus, if we
\emph{assume} that the \PTG{s} we are given as input \emph{are}
reset-acyclic, we can solve them in \emph{pseudo-polynomial time}, and
show that their value functions are cost functions with at most a
pseudo-polynomial number of cutpoints, using our techniques.

In~\cite{BouLar06} the authors showed that with one-clock \PTG and
non-negative weights only we could bound the number of resets by the
number of locations, without changing the value functions.
Unfortunately, these arguments do not hold for arbitrary weights, as
shown by the \PTG in \figurename~\ref{fig:ex-ptgrr}. In that \PTG, we
claim that $\val(\loc_0)=0$; that \MinPl has no optimal strategies,
but a family of $\varepsilon$-optimal strategies
$\stratmin^\varepsilon$ each with value~$\varepsilon$; and that each
$\stratmin^\varepsilon$ requires \emph{memory whose size depends on
  $\varepsilon$} and might \emph{yield a play visiting at least
  $1/\varepsilon$ times the reset} between $\loc_1$ and $\loc_0$
(hence the number of resets cannot be bounded). For all
$\varepsilon>0$, $\stratmin^\varepsilon$ consists in: waiting
$1-\varepsilon$ time units in $\loc_0$, then going to $\loc_1$ during
the $\lceil 1/\varepsilon\rceil$ first visits to $\loc_0$; and to go
directly to $\loc_f$ afterwards. Against $\stratmin^\varepsilon$,
$\MaxPl$ has several possible choices:
\begin{enumerate}
\item either wait $\eta\in[0,\varepsilon]$ time units in $\loc_1$,
  wait $\varepsilon-\eta$ time units in $\loc_2$, then reach $\loc_f$;
  or
\item wait $\varepsilon$ time unit in $\loc_1$ to have the clock
  equal to $1$, and force the cycle by going back to $\loc_0$, where
  the game will wait for $\MinPl$'s next move.
\end{enumerate}
Thus, all plays according to $\stratmin^\varepsilon$ will visit a
sequence of locations which is either of the form
$\loc_0 (\loc_1 \loc_0)^k\loc_1\loc_2\loc_f$, with
$0\leq k <\lceil 1/\varepsilon\rceil$; or of the form
$\loc_0 (\loc_1 \loc_0)^{\left\lceil1/\varepsilon\right\rceil}\loc_f$.
In the former case, the price of the play will be
$-k\varepsilon-\eta+(\varepsilon-\eta)=-(k-1)\varepsilon-2\eta\leq\varepsilon$;
in the latter, $-\varepsilon(\lceil 1/\varepsilon\rceil)+1\leq
0$. This shows that $\val(\loc_0)=0$, but there are no optimal
strategies as none of these strategies allow one to guarantee a price
of $0$ (neither does the strategy that waits $1$ time unit in
$\loc_0$).

\begin{figure}[tbp] 
  \centering 
  \begin{tikzpicture}[minimum size=5mm, node distance = 2cm] 
    \everymath{\footnotesize}

    \node[draw,circle, label=above:$0$] at (0,0) (q0) {\makebox[0pt][c]{$\loc_0$}};
    \node[draw,rectangle,right of=q0,node distance=4cm,label=above:$-1$] (q1) {\makebox[0pt][c]{$\loc_1$}};
    \node[draw,rectangle, below of=q1, label=right:$1$,xshift=-1cm] (q2) {\makebox[0pt][c]{$\loc_2$}};
    \node[draw,circle,accepting, below of=q2,xshift=-1cm] (qf) {\makebox[0pt][c]{$\loc_f$}};

    \path[->] 
    (q2) edge node[below right] {$x=1$} (qf) 
    (q1) edge node[below right] {$x\leq 1$} (q2) 
    (q0) edge[bend left=10]   node[above] {$x\leq 1$} (q1)
    (q1) edge[bend left=10] node[below] {$x=1, x:=0$} (q0)
    (q0) edge node[below left] {$1$} (qf)
    (qf) edge[loop right] (qf);


  \end{tikzpicture}
  \caption{A \PTG where the number of resets in optimal plays cannot
    be bounded a priori.} \label{fig:ex-ptgrr}
\end{figure}

If bounding the number of resets is not possible in the general case,
it could be done if one adds constraints on the cycles of the
game. This kind of restriction was used in~\cite{BCR14} where the
authors introduce the notion of \emph{robust} games (and a more
restrictive one of \emph{divergent} games was used
in~\cite{BMR17a}). Such games require among other things that there
exists $\kappa>0$ such that every play starting and ending in the same
pair location and time region has either a positive \pname or a \pname
smaller than $-\kappa$. Here we require a less powerful assumption as
we put this restriction only on cycles containing a reset.
\begin{defi}
  Given $\kappa>0$, a $\kappa$-negative-reset-acyclic \PTG
  ($\kappa$-\NRAPTG) is a \PTG where for every location
  $\loc\in \locs$ and every cyclic finite play $\rho$ starting and
  ending in $(\loc,0)$, either $\puse{(\rho)}\geq 0$ or
  $\puse{(\rho)} <-\kappa$.
\end{defi}

The \PTG of Figure~\ref{fig:ex-ptgrr} is not a $\kappa$-\NRAPTG for
any $\kappa>0$ as the play
$(\loc_0,0)\xrightarrow{0} (\loc_1,1-\kappa/2) \xrightarrow{-\kappa/2}
(\loc_0,0)$ is a cycle containing a reset and with a negative \pname
strictly greater than $-\kappa$. On the contrary, in
Figure~\ref{fig:ex-nraptg} we show a $1$-\NRAPTG and its region
\PTG.\@ Here, every cycle containing a reset is between $\loc_0$ and
$\loc_1$ and such cycles have at most \pname $-1$. The value of this
\PTG is $0$ but no strategies for $\MaxPl$ can achieve it because of
the guard $x>0$ on the transition from $\loc_1$ to $\loc_f$. As this
guard is not strict anymore in the region \PTG, both players have an
optimal strategy in this game (this is not always the case).

\begin{figure}[tbp]
  \begin{center}
    \begin{tikzpicture}[minimum size=5mm, node distance = 4cm] 
      \node[draw,circle,label=left:$0$] (q0) {$\loc_0$};
      \node[draw,rectangle,label=right:$-1$,right of=q0] (q1) {$\loc_1$};
      \node[draw,circle,accepting,below of=q0,xshift=2cm] (q2) {$\loc_f$};

      \path[->] (q0) edge node[below left] {$1, x<1$}(q2)
      (q1) edge node[below right] {$0<x<1$} (q2)
      (q0) edge[bend right=10] node[below] {$x<1$} (q1)
      (q1) edge[bend right=10] node[above] {$-1, x<1, x:=0$} (q0)
      (q2) edge[loop right] (q2);

      \begin{scope}[xshift=7cm,yshift=1cm]
        \node[draw,circle,label=left:$0$] (q00) {\footnotesize$\loc_0,\{0\}$};
        \node[draw,circle,label=left:$0$,below of=q00,node distance=2.5cm] (q001) {\footnotesize$\loc_0,[0,1]$};

        \node[draw,rectangle,label=right:$-1$,right of=q001] (q101) {$\loc_1,[0,1]$};

        \node[draw,rectangle,above of=q101,label=right:$-1$,node distance=2.5cm] (q10) {$\loc_1,\{0\}$};

        \node[draw,circle,accepting,below of=q001,xshift=2cm,yshift=1.5cm] (q2bis) {$\loc_f$};
        
        \path[->] (q00) edge node[right] {$x=0$} (q001)
        (q10) edge node[right] {$x=0$} (q101)
        (q001) edge node[above right] {$1$}  (q2bis)
        (q00) edge[bend right=60,looseness=1.5] node[below left] {$1$} (q2bis)
        (q101) edge (q2bis)
        (q001) edge (q101)
        (q101) edge node[sloped,above] {$-1,x:=0$} (q00)
        (q00) edge[bend right=10] (q10)
        (q10) edge[bend right=10] node[above] {$-1$}(q00)
        (q2bis) edge[loop right] (q2bis);

      \end{scope}
    \end{tikzpicture}

    \caption{A $1$-NRAPTG and its region \PTG (some guards removed 
    for better readability)}
    \label{fig:ex-nraptg}
  \end{center}
\end{figure}

In order to bound the number of resets of a $\kappa$-\NRAPTG, we 
first prove a bound on the value of such
games, that will be useful in the following. We let $k=|\reggame|$ be
the number of regions.  Recall
that $M$ is a bound on the valuations taken by the clock in
$\game$, as discussed on page~\pageref{sec:notat-defin}.
\begin{lem}
  \label{lem:bound} For all $\kappa$-\NRAPTG{s} $\game$ and
  $(\loc,\valuation)\in\confgame$: either
  $\valgame(\loc,\valuation)\in\{-\infty,+\infty\}$,~or
  \[-|\locs|M\maxPriceLoc-|\locs|^2(|\locs|+2)\maxPriceTrans - \maxPriceFin \leq
    \valgame(\loc,\valuation)\leq
    |\locs|M\maxPriceLoc+|\locs|k\maxPriceTrans + \maxPriceFin\,.\]
\end{lem}
\begin{proof}
  Consider the case where
  $\valgame(\loc,\valuation)\notin\{-\infty,+\infty\}$. Let
  $\kappa>2\varepsilon>0$. Then, there exist $\stratmin$ and
  $\stratmax$ $\varepsilon$-optimal strategies for $\MinPl$ and
  $\MaxPl$, respectively.

  Let $\stratmin^{\neg c}$ be any memoryless strategy of $\MinPl$ in
  the reachability timed game induced by $\game$ such that no play
  consistent with $\stratmin^{\neg c}$ goes twice in the same couple
  (location, region). If such a strategy does not exist, as the clock
  constraints are the same during the first and second occurrences of
  this couple, $\MaxPl$ can enforce the cycle infinitely often, thus
  the reachability game is winning for $\MaxPl$ and the value of $\game$ is
  $+\infty$. Let us note
  $\rho=\CPlay{(\loc,\valuation),\stratmin^{\neg c},\stratmax}$.  By
  $\varepsilon$-optimality of $\stratmax$,
  $\cost{\rho}\geq \valgame(\loc,\valuation)-\varepsilon$.  Let
  $\costTrans{\rho}$ be the price of $\rho$ due to the weights of the
  transitions, and $\costLoc{\rho}$ be the weight due to the time
  elapsed in the locations of the game:
  $\puse{(\rho)} = \costTrans{\rho}+\costLoc{\rho}$. As there are no
  cycles in the game according to couples (location, region), there
  are at most $|\locs|k$ transitions, thus
  $\costTrans{\rho}\leq |\locs|k \maxPriceTrans$.  Moreover, the
  absence of cycles also implies that we do not take two transitions
  with a reset ending in the same location or one transition with a
  reset ending in the initial location, thus we take at most
  $|\locs|-1$ such transitions. Therefore at most $|\locs|M$ units of
  time elapsed and $\costLoc{\rho}\leq |\locs|M \maxPriceLoc$. 
  Adding the final cost, this
  implies that
  \[\valgame(\loc,\valuation)-\varepsilon \leq \cost\rho\leq
  |\locs|M \maxPriceLoc+|\locs|k \maxPriceTrans + \maxPriceFin\,.\]
  By taking the limit of $\varepsilon$ towards $0$, we obtain the
  announced upper bound.

  We now prove the lower bound on the value. To that extent, consider
  now the completed play $\rho =
  \CPlay{(\loc,\valuation),\stratmin,\stratmax}$. We have that
  $\cost{\rho}\leq \valgame(\loc,\valuation)+\varepsilon$.

  We want to lower bound the price of $\rho$, therefore
  non-negative cycles can be safely ignored. Let us show that there
  are no negative cycles around a transition with a reset. If it was
  the case, since the game is a $\kappa$-\NRAPTG, this cycle has
 \pname at most $-\kappa$. Since the strategy $\stratmax$ is
  $\epsilon$-optimal, and $\kappa > \epsilon$, it is not possible that
  $\stratmax$ decides alone to take this bad cycle. Therefore,
  $\stratmin$ has the capability to enforce this cycle, and to exit it
  (otherwise, \MaxPl would keep \him inside to get value $+\infty$):
  but then, \MinPl could decide to cycle as long as \he wants, then
  guaranteeing a value as low as possible, which contradicts the fact
  that $\val(\loc,\valuation)\notin\{-\infty,+\infty\}$. Therefore,
  the only cycles in $\rho$ around transitions with resets, are
  non-negative cycles. This implies that its price is bounded below by
  the price of a sub-play obtained by removing the cycles in $\rho$.

  We now consider a play where each reset transition is taken at most
  once in $\rho$, and lower-bound its price.

  If $\rho$ contains a cycle around a location $\loc'\in\LocsMax$
  without reset transitions, this cycle has the form
  $(\loc',\valuation')\xrightarrow{c'}(\loc'',\valuation'+t)
  \cdots\xrightarrow{c''} (\loc',\valuation'')$ with
  $\valuation''\geq \valuation'$, followed in $\rho$ by a transition
  towards configuration $(\loc''',\valuation''+t')$. Thus, another
  strategy for \MaxPl could have consisted in skipping the cycle by
  choosing as delay in the first location $\loc'$,
  $\valuation''-\valuation'+t'$ instead of $t$. This would get a new
  strategy that cannot make the price increase above
  $\valgame(\loc,\valuation)+\varepsilon$, since it is still playing
  against an $\epsilon$-optimal strategy of \MinPl. Therefore, we can
  consider the sub-play~$\rho_f$ of $\rho$ where all such cycles are
  removed: we still have
  $\cost{\rho_f}\leq \valgame(\loc,\valuation)+\varepsilon$.

  Suppose now that $\rho_f$ contains a cycle around a location
  $\loc'\in\LocsMin$ without reset transitions, of the form
  $(\loc',\valuation')\xrightarrow{c'}(\loc'',\valuation'+t)
  \cdots\xrightarrow{c''} (\loc',\valuation'')$ with $\valuation'$ and
  $\valuation''$ in the same region, composed of $\MinPl$'s locations
  only, and followed in $\rho$ by a transition towards configuration
  $(\loc''',\valuation'' + t')$. Then, the transition weight of this
  cycle is non-negative, otherwise $\MinPl$ could enforce this cycle
  \he entirely controls, while letting only a bounded time pass
  (smaller and smaller as the number of cycles grow). This is not
  possible. 

  Therefore, we have that two occurrences of a same $\MaxPl$'s
  location in $\rho_f$ are separated by a reset transition and two
  occurrences of a same $\MinPl$'s couple (location, region) are either
  separated by a reset or by a $\MaxPl$'s location. As there are at
  most $|\locs|-1$ resets, $|\locs|$ locations of $\MaxPl$ and
  $|\locs|k$ couples (location, region) for $\MinPl$, $\rho_t$ contains
  at most $|\locs|^2$ locations of $\MaxPl$ and
  $|\locs| k(|\locs|^2+|\locs|-1 +1)$ locations of $\MinPl$, which
  makes for at most $|\locs|^2(|\locs|k+k+1)$ locations. Thus
  $\costLoc{\rho_t}\geq -|\locs|^2(|\locs|k+k+1) \maxPriceLoc$.
  Moreover, as at most $|\locs|-1$ resets are taken in~$\rho_f$ and
  that the game is bounded by $M$,
  $\costLoc{\rho_f} \geq -|\locs|M \maxPriceLoc$. Adding the final cost, this implies that
  \[\valgame(\loc,\valuation)+\varepsilon\geq
  \costLoc{\rho_f}+\costTrans{\rho_t} \geq -|\locs|M
  \maxPriceLoc-|\locs|^2(|\locs|k+k+1) \maxPriceTrans - \maxPriceFin\,.\]
  Taking the limit when $\varepsilon$ tends to $0$, we obtain the
  desired lower bound.
\end{proof}

Using this bound on the value of a $\kappa$-\NRAPTG, one can give a
bound on the number of cycles needed to be allowed. The idea is that
if a reset is taken twice and the generated cycle has positive \pname,
either $\MinPl$ can modify its strategy so that it does not take this
cycle or the value of the game is $+\infty$ as $\MaxPl$ can prevent
$\MinPl$ from reaching a final location. On the contrary if the
cycle has negative \pname, then by definition of a $\kappa$-\NRAPTG,
this \pname is less than $-\kappa$. Thus by allowing enough such
cycles, as we have bounds on the values of the game, we know when we
will have enough cycles to get under the lower bound of the value of
the game. By solving the copies of the game, if we reach a value that
is smaller than the lower bound of the value, then it means that the
value is $-\infty$.

\begin{lem}
  \label{le:no-reset}
  For all $\kappa>0$, the value of a $\kappa$-\NRAPTG can be computed
  by solving $\lceil 2n (\val^{sup}-\val_{inf})/\kappa\rceil$ \PTG{s}
  without resets and using the same set of guards, where $\val^{sup}$
  and $\val_{inf}$ are the upper and lower bounds of the value of the
  game given by Lemma~\ref{lem:bound}. Moreover, from
  $\varepsilon$-optimal strategies on those $k$ games, we can build
  $\varepsilon \lceil 2n (\val^{sup}-\val_{inf})/\kappa\rceil$-optimal
  strategies in the original game.
\end{lem}

Solving the PTGs without resets can be
done by using the same algorithms as the one described before for
SPTGs: indeed, since we play in a region-PTG, we can focus on the
resolution of a subgame staying in the same region until the final
transition, and such a game can even be decomposed into simpler games
where every region has length $1$, which can then be interpreted as a
SPTG.\@ Another possibility would be to rescale the time and the weights
in order to transform a region $(a,b)$ into $(a,a+1)$, avoiding to
split the region into $b-a$ different subregions.

The values $\val^{sup}$ and $\val_{inf}$, and therefore also
$\lceil 2n (\val^{sup}-\val_{inf})/\kappa\rceil$, are
pseudo-polynomial in the size of the original game, which allows us to
conclude:

\begin{thm}
  \label{th:NRAPTGalg}
  Let $\kappa>0$ and $\game$ be a $\kappa$-\NRAPTG.\@ Then for every
  location $q\in Q$, the function
  $\valuation \mapsto \valgame(q,\valuation)$ is computable in
  pseudo-polynomial time and is piecewise-affine with at most a
  pseudo-polynomial number of cutpoints. Moreover, for every
  $\varepsilon>0$, there exist (and we can effectively compute)
  $\varepsilon$-optimal strategies for both players.
\end{thm}

The robust games defined in~\cite{BCR14} restricted to one-clock are a
subset of the \NRAPTG, therefore their value is computable with the
same complexity.  While we cannot extend the computation of the value
to all (one-clock) \PTG{s}, we can still obtain information on the
nature of the value function:
\begin{thm} 
  The value functions of all one-clock \PTG{s} are cost functions with
  at most a pseudo-polynomial number of cutpoints.
\end{thm} 
\begin{proof} 
  Let $\game$ be a one-clock \PTG.\@ Let us replace all transitions
  $(\loc,g,\top,\loc')$ resetting the clock by $(\loc,g,\bot,\loc'')$,
  where $\loc''$ is a new final location with
  $\fgoal_{\loc''}=\val_\game(\loc',0)$---observe that
  $\val_\game(\loc,0)$ exists even if we cannot compute it, so this
  transformation is well-defined. This yields a reset-acyclic \PTG
  $\game'$ such that $\val_{\game'}=\val_\game$. The pseudo-polynomial
  number of cutpoints of reset-acyclic \PTG, as for \SPTG{s}, does not
  depend on the size of final prices (but only on the price of
  transitions, and the number of locations), which allows us to
  conclude.
\end{proof}

As a consequence, in the particular case of non-negative prices only
where transitions with a reset can be unfolded to remove cycles in
which they are contained, this ensures that the exponential-time
algorithms of~\cite{BouLar06,Rut11,DueIbs13} indeed have a
pseudo-polynomial time complexity.

\begin{cor}
  The value functions of all one-clock \PTG{s} with only non-negative
  prices can be computed in pseudo-polynomial time. 
\end{cor}

\section{Conclusion}

In this work, we study, for the first time, priced timed games with
arbitrary weights and one clock, showing how to compute optimal values
and strategies in pseudo-polynomial time for the special case of
simple games. This complexity result
is better than previously obtained results in the case of non-negative
weights only~\cite{DueIbs13,Rut11} (where an exponential complexity
was obtained), and we follow different paths to prove termination and
(partial) correctness (due to the presence of negative weights). In
order to push our algorithm as far as we can, we introduce the class
of negative-reset-acyclic games for which we obtain the same result:
as a particular case, we can solve all priced timed games with one
clock for which the clock is reset in every cycle of the underlying
region automaton. As future works, it is appealing to solve the full
class of priced timed games with arbitrary weights and one clock. We
have shown why our technique seems to break in this more general
setting, thus it could be interesting to study the difficult negative
cycles without reset as their own, with different techniques.

\bibliographystyle{alphaurl}
\bibliography{biblio}

\end{document}